\documentclass[a4paper,USenglish,
autoref,
cleveref,
numberwithinsect
]{lipics-v2021}

\usepackage{cite}
\graphicspath{{./figures/}}

\nolinenumbers 

\title{Dynamic Connectivity in Disk Graphs}

\author{Alexander Baumann}{Institut f\"ur Informatik, Freie Universität Berlin, Germany}
  {akauer@inf.fu-berlin.de}{https://orcid.org/0000-0001-8112-0482}{Supported 
  in part by grant 1367/2016 from the German-Israeli Science Foundation (GIF),
  by the German Research Foundation within the collaborative DACH 
  project \emph{Arrangements and Drawings} as DFG Project MU 3501/3-1,
  and by ERC StG 757609.}
\author{Haim Kaplan}{School of Computer Science, Tel Aviv University, Israel}
  {haimk@tau.ac.il}{}{Partially supported by ISF grant 1595/19 
  and by the Blavatnik research foundation.}
\author{Katharina Klost}{Institut f\"ur Informatik, Freie Universität Berlin, Germany}
{kathklost@inf.fu-berlin.de}{https://orcid.org/0000-0002-9884-3297}{}
\author{Kristin Knorr}{Institut f\"ur Informatik, Freie Universität Berlin, Germany}
{knorrkri@inf.fu-berlin.de}{https://orcid.org/0000-0003-4239-424X}{Supported by the German Science Foundation 
within the research training group `Facets of Complexity' (GRK 2434).}
\author{Wolfgang Mulzer}{Institut f\"ur Informatik, Freie Universität 
Berlin, Germany}{mulzer@inf.fu-berlin.de}{https://orcid.org/0000-0002-1948-5840}
{Supported in part by ERC StG 757609.}
\author{Liam Roditty}{Department of Computer Science, Bar Ilan 
University, Israel}{liamr@macs.biu.ac.il}{}{}
\author{Paul Seiferth}{Institut f\"ur Informatik, Freie 
Universität Berlin, Germany}{pseiferth@inf.fu-berlin.de}{}{}

\authorrunning{A.Baumann, H.Kaplan, K.Klost, K.Knorr, W.Mulzer, L.Roditty,
and P.Seiferth} 

\Copyright{Alexander Baumann, Haim Kaplan, Kristin Knorr, Wolfgang Mulzer, 
Liam Roditty and Paul Seiferth}



\category{} 

\funding{Supported in part by grant 1367/2016 from the German-Israeli 
Science Foundation (GIF).}

\nolinenumbers 

\hideLIPIcs

\usepackage[obeyFinal]{todonotes}
\usepackage{csquotes}
\usepackage{placeins}
\usepackage{tikz}
\usetikzlibrary{calc}
\usetikzlibrary{decorations.pathreplacing}

\let\originalleft\left
\let\originalright\right
\renewcommand{\left}{\mathopen{}\mathclose\bgroup\originalleft}
\renewcommand{\right}{\aftergroup\egroup\originalright}

\usepackage{xspace}
\usepackage{xcolor}
\usepackage{enumitem}

\definecolor{mygreen}{rgb}{0, 0.56, 0.26}
\definecolor{mylightgreen}{rgb}{0, 0.9, 0.42}
\definecolor{myblue}{rgb}{0, 0.23, 0.7}
\definecolor{mylightblue}{rgb}{0.2, 0.6, 1}
\definecolor{mypurple}{rgb}{0.5, 0, 0.5}
\definecolor{myred}{rgb}{0.72, 0, 0}
\definecolor{myorange}{rgb}{1, 0.61, 0.21}
\definecolor{myyellow}{rgb}{0.96, 0.84, 0.19}
\usepackage{array,multirow}
\usepackage{arydshln}
\newcommand{\RDS}{\textsf{RDS}\xspace}
\newcommand{\DLE}{\textsf{DLE}\xspace}
\newcommand{\MBM}{\textsf{MBM}\xspace}
\newcommand{\AWNN}{\textsf{AWNN}\xspace}

\DeclareMathOperator{\polylog}{polylog}
\DeclareMathOperator{\diam}{diam}
\DeclareMathOperator{\CL}{CL}
\newcommand{\G}{\mathcal{G}}
\newcommand{\eps}{\varepsilon}
\newcommand{\Q}{\mathcal{Q}}
\newcommand{\R}{\mathbb{R}}
\newcommand{\T}{\mathcal{T}}
\newcommand{\cH}{\mathcal{H}}
\newcommand{\cR}{\mathcal{R}}
\newcommand{\cD}{\mathcal{D}}
\renewcommand{\P}{\mathcal{P}}

\begin{document}
\maketitle

\begin{abstract}
Let $S \subseteq \R^2$ be a set of $n$ \emph{sites} 
in the plane, so that every site $s \in S$ has an 
\emph{associated radius} $r_s > 0$. Let $\cD(S)$ be
the \emph{disk intersection graph} defined 
by $S$, i.e., the graph with vertex set $S$ and 
an edge between two distinct sites $s, t \in S$ 
if and only if the disks with centers $s$, $t$ 
and radii $r_s$, $r_t$ intersect.
Our goal is to design data structures that 
maintain the connectivity structure of $\cD(S)$ 
as sites are inserted and/or deleted in $S$.

First, we consider \emph{unit disk graphs}, i.e., 
we fix $r_s = 1$, for 
all sites $s \in S$.  For this case, we describe a data 
structure that has  $O(\log^2 n)$ amortized 
update time and $O(\log n/\log\log n)$ 
query time. Second, we look at 
disk graphs \emph{with bounded radius ratio} 
$\Psi$, i.e., for all $s \in S$, we have 
$1 \leq r_s \leq \Psi$, for a parameter $\Psi$ 
that is known in advance. Here, we not only investigate 
the fully dynamic case, but also the
incremental and the decremental scenario, where
only insertions or only deletions of sites are 
allowed. 
In the fully dynamic case, we achieve amortized expected
update time $O(\Psi \log^{4} n)$ and
query time $O(\log n/\log \log n)$.
This improves the currently best update time by a 
factor of $\Psi$.
In the incremental case, we achieve 
logarithmic dependency on $\Psi$, with a data structure 
that has $O(\alpha(n))$ amortized query time and 
$O(\log\Psi \log^{4} n)$ amortized expected update time,
where $\alpha(n)$ denotes the inverse Ackermann function.

For the decremental setting, we first develop an 
efficient decremental \emph{disk revealing} data structure:  
given two sets $R$ and $B$ of disks in the plane, we can
delete disks from $B$, and upon each deletion,
we receive a list of all disks in $R$
that no longer intersect the union of $B$.
Using this data structure, we get decremental data 
structures with a query time of 
$O(\log n/\log \log n)$ that supports deletions 
in $O(n\log \Psi \log^{4} n)$ 
overall expected time for disk graphs with bounded radius ratio $\Psi$ and 
$O(n\log^{5} n)$ 
overall expected time for  
disk graphs with arbitrary radii, assuming that the deletion
sequence is oblivious of the internal random choices of the data structures.
\end{abstract}

\section{Introduction}%
\label{sec:introduction}
Suppose we are given a simple, undirected 
graph $G$, and we would like to preprocess it
so that we can determine efficiently if any two 
given vertices of $G$ lie in the same connected 
component. If $G$
is fixed, there is a simple solution: we perform a  
graph search in $G$ (e.g., a BFS or a DFS~\cite{cormen_introduction_2009})
to label the vertices of each
connected component with a unique identifier.
After that, we can answer a query in $O(1)$ time by comparing
the corresponding identifiers,
with linear preprocessing time and space. 

A much harder situation occurs if $G$ changes over time. 
Now, the connectivity queries may be 
interleaved with operations that modify $G$. 
If $G$ has a fixed  set of $n$ vertices and
edges can only be inserted, the problem directly reduces to
disjoint-set-union~\cite{reif_topological_1987}.
Using a standard disjoint-set-union 
structure~\cite{cormen_introduction_2009}, an optimal 
data structure that achieves 
$O(\alpha(n))$ amortized time for edge insertions and 
queries can be constructed, 
where $\alpha(n)$ is the inverse Ackermann 
function~\cite{cormen_introduction_2009}.
This reduction seems to be folklore, we explain 
the details in \cref{thm:generalinc} below.
If the vertex set is fixed, but 
edges can be both inserted and deleted,
the  data structure with the currently 
fastest update time is due to Huang et al.~\cite{huang_fully_2017}, 
which is based on a result by Thorup~\cite{thorup_near-optimal_2000}.
It has amortized expected update time 
$O(\log n(\log \log n)^2)$ and query time $O(\log n / \log \log \log n)$.
Generally, there were several recent results in this area~\cite{huang_fully_2017,wulff-nilsen_faster_2013,kapron_dynamic_2013}.
In our constructions, we will use a slightly older, 
but simpler data structure by Holm, Lichtenberg, and 
Thorup~\cite{holm_poly-logarithmic_2001} with a memory optimization 
by Thorup~\cite{thorup_near-optimal_2000}.
We will refer to this data structure as the \emph{HLT-structure}.
It achieves $O(\log n / \log \log n)$  query time and $O(\log^2 n)$ 
amortized update time.
For the special case of planar graphs,  
Eppstein et al.~\cite{eppstein_maintenance_1992a} 
give a data structure with $O(\log n)$ amortized time for 
both queries and updates.

One way to model the case of a changing 
vertex set is the \emph{dynamic subgraph
connectivity} problem. Here, we have a known fixed
graph $H$ with $n$ vertices and $m$ edges,
and $G$ is a subgraph of $H$ that changes
dynamically by activating or deactivating
the vertices of $H$ (the edges of $G$ are 
induced by the active vertex set). 
In this setting, there is a data structure 
with $O(m^{1/3} \polylog n)$ 
query time and $O(m^{2/3}\polylog n)$
amortized update time, developed by Chan et al.~\cite{chan_dynamic_2011}. 

In this paper, we add a geometric twist: we
study the dynamic connectivity 
problem on different variants of \emph{disk
intersection graphs}.
Let $S \subset \R^2$ be a set of planar 
\emph{point sites}, where each site $s \in S$ 
has an \emph{associated radius} $r_s > 0$. 
The \emph{disk intersection graph} (\emph{disk graph}, 
for short) $\cD(S)$  is the undirected graph
with vertex set $S$ that has an edge between 
any two distinct sites $s$ and $t$ if and only 
if the Euclidean distance between $s$ and $t$ 
is at most $r_s + r_t$.
In other words, there is an edge between
$s$ and $t$ if and only
if the disks with centers $s$ and $t$ and
with radius $r_s$ and $r_t$, respectively, intersect. 
Note that even though $\cD(S)$ is fully described by 
the $n$ sites and their associated radii, 
it might have $\Theta(n^2)$ edges. Thus, 
our goal is to find algorithms whose 
running time depends only on the number of sites 
and not on the number of edges.
We consider three variants of disk graphs, 
characterized by the possible values for the associated
radii.  In the first variant,
\emph{unit disk graphs}, all radii are~$1$.  
In the second variant, \emph{disk graphs of bounded radius
ratio}, all radii must come from the interval $[1, \Psi]$,
where $\Psi$ is a parameter that
is known in advance and that may depend on the
number of sites $n$, i.e., it can be as large as $n^2, 2^n$, or even larger.
In the third variant, \emph{general disk graphs}, the radii can be 
arbitrary positive real numbers. 

The static connectivity problem in disk graphs can be
solved similarly as in general graphs: we compute the
connected components and label the vertices of each connected
component with the same unique identifier. The challenge for disk
graphs lies in finding a quick way to perform the preprocessing
step. For unit disk graphs, there are several methods to 
perform a BFS in worst-cast 
time $O(n \log n)$ (see, e.g.,~\cite{KLOST2023101979}
and the references therein). For general disk graphs,
Cabello and Jejčič~\cite{CabelloJ15} observed that a BFS can
be performed with the help of an efficient weighted nearest
neighbor structure. With this approach and the fastest available
data structures, this leads to an algorithm with expected running
time $O(n \log^4 n)$~\cite{kaplan_dynamic_2020,Liu20}.

We assume that $S$ is \emph{dynamic}, i.e., 
sites can be inserted and deleted over time.
At each update, the edges incident 
to the modified site appear or disappear in $\cD(S)$.
An update can change up to $n-1$ edges in $\cD(S)$, 
so simply storing $\cD(S)$ in the HLT-structure 
could lead to potentially superlinear 
update times and might even be slower than recomputing 
the connectivity information from scratch. 

For dynamic connectivity in general disk graphs,
Chan et al.~\cite{chan_dynamic_2011} give a data structure with 
$O(n^{1/7 + \eps})$ query time and $O(n^{20/21 +\eps})$ 
amortized update time, where $\eps > 0$ is a constant that can
be made arbitrarily small, but which must be fixed.
As far as we know, this is currently still the 
best fully dynamic connectivity structure 
for general disk graphs.
However, Chan et al.~present their 
data structure as a special case of a more 
general setting, so there is hope  that the specific 
geometry of disk graphs may allow for better running times. 

Indeed, several results show that for certain disk graphs,
one can achieve polylogarithmic update and query times.
For unit disk graphs, Chan et al.~\cite{chan_dynamic_2011} 
observe that there is a data structure with $O(\log^{6} n)$ 
amortized update time and $O(\log n/\log\log n)$ query time.\footnote{Actually,
Chan et al.~\cite{chan_dynamic_2011} claim an amortized expected update time
of $O(\log^{10} n)$. However, the underlying data structures
have been improved since their paper was published~\cite{Chan20a}. With
this, the approach yields the stated
improved running time; see \cref{sec:unit} for more details.}
For bounded radius ratio, 
Kaplan et al.~\cite{kaplan_dynamic_2020} show that there 
is a data structure with amortized expected update time 
$O(\Psi^2 \log^{4}n)$ and query
time $O(\log n /\log \log n)$.\footnote{The original paper
claims an amortized expected update time of
$O(\Psi^2 \lambda_6(\log n) \log^{9}n)$,
but recent improvements in the underlying data structure~\cite{Liu20}
lead to the better bound.}

Both results use the notion of a 
\emph{proxy graph}, a sparse graph that 
models the connectivity of the original disk 
graph and that can be updated efficiently
with suitable dynamic geometric data structures. 
The proxy graph is then stored in 
the HLT -structure, so 
the query procedure coincides with the one by HLT.
The update operations involve a combination 
of updating the proxy graph with the help of 
geometric data structures and of
modifying the edges in the HLT-structure.

\subparagraph{Our results.}

\begin{table}
\begin{center}

\renewcommand{\arraystretch}{1.1}
\begin{tabular}{p{2.52cm}|p{3.92cm}|p{2.85cm}|p{2.82cm}}
  \multicolumn{4}{c}{unit disks, fully dynamic} \\ \hline
  & update time & query time & space usage \\ \hline
  Chan et al.~\cite{chan_dynamic_2011} & $O(\log^4 n)$ amortized & $O(\log n/\log \log n)$ & $O(n \log^2 n)$ \\ \hline
  \cref{thm:dynamicUDG} & $O(\log^2 n)$ amortized & $O(\log n / \log \log n)$ & $O(n)$ \\ \hline
  \multicolumn{4}{c}{} \\
  \multicolumn{4}{c}{$r_s \in [1,\Psi]$, fully dynamic} \\ \hline
  & update time & query time & space usage \\ \hline
  Kaplan et al.~\cite{kaplan_dynamic_2020} & $O(\Psi^2 \log^{4}n)$ amortized expected & $O(\log n/\log \log n)$ & $O(\Psi^2 n \log n)$ \\ \hline
  \cref{thm:bounded-radius-ratio-psi} & $O(\Psi \log^{4} n)$ amortized expected & $O(\log n/\log \log n)$ & $O(\Psi n \log n)$ \\ \hline
  \multicolumn{4}{c}{} \\
  \multicolumn{4}{c}{$r_s \in [1,\Psi]$, insertion only} \\ \hline
  & update time & query time & space usage \\ \hline
  \cref{thm:bounded:insertion} & $O(\log\Psi \log^{4} n)$ amortized expected & $O(\alpha(n))$ amortized & $O(n \log \Psi \log n)$ \\ \hline
  \multicolumn{4}{c}{} \\
  \multicolumn{4}{c}{$r_s \in [1,\Psi]$, deletion only} \\ \hline
  & update time & query time & space usage \\ \hline
  \cref{thm:bounded:deletion} & $O( n\log\Psi\log^4n)$ overall expected
 with oblivious adversary & $O(\log n/\log\log n)$ & $O(n (\log n + \log \Psi))$ \\ \hline
  \multicolumn{4}{c}{} \\
  \multicolumn{4}{c}{general, fully dynamic} \\ \hline
  & update time & query time & space usage \\ \hline
  Chan et al.~\cite{chan_dynamic_2011} & $O(n^{1/7 + \eps})$ amortized & $O(n^{20/21 +\eps})$ & $O(n^{5/3 + \eps})$ \\ \hline
  \multicolumn{4}{c}{} \\
  \multicolumn{4}{c}{general, deletion only} \\ \hline
  & update time & query time & space usage \\ \hline
  \cref{thm:unbounded:deletion} & $O(n\log^{5} n)$ 
   overall expected 
   with oblivious adversary & $O(\log n/\log\log n)$ & $O(n \log^2 n)$ \\ \hline
\end{tabular}
\renewcommand{\arraystretch}{1}

\end{center}
\caption{The state of the art after our work. 
Time bounds are per operation, unless noted otherwise.
In the semi-dynamic cases where no explicit bounds are given 
the best known results coincide with the fully-dynamic case.
The space requirements for the results of Chan et al.~\cite{chan_dynamic_2011}
are not stated explicitly in their paper, but they were derived
according to our understanding of their method.}
\label{tab:intro:runningtimes}
\end{table}

For unit disk graphs, we significantly improve 
the observation of Chan et al.~\cite{chan_dynamic_2011}: 
with a direct approach utilizing a grid-based proxy graph
and dynamic data structures for lower envelopes,
we obtain $O(\log^2 n)$ amortized update and 
$O(\log n / \log \log n)$ query time (\cref{thm:dynamicUDG}).

For bounded radius ratio, we give a data structure 
that improves the update time in \cref{thm:bounded-radius-ratio-psi}.
Specifically, we achieve
amortized expected update time $O(\Psi \log^{4} n)$ 
and query time $O(\log n/\log \log n)$.
Compared to the previous data structure of Kaplan et 
al.~\cite{kaplan_dynamic_2020},
this improves the factor in the update 
time from $\Psi^2$ to $\Psi$.

We also provide partial results that 
push the dependency on $\Psi$ from linear to logarithmic.
For this, we consider the \emph{semi-dynamic} setting, in which
only insertions (\emph{incremental}) or only deletions 
(\emph{decremental}) are allowed.
In both semi-dynamic settings we use quadtrees and 
assign the input disks to certain sets defined by the cells of the 
quadtree. Similar approaches have been used 
before~\cite{loffler_dynamic_2013,BuchinLMM11,LofflerM12,KaplanMRS18,KaplanKMRSS19}. 
In the incremental setting, we use a 
dynamic additively weighted Voronoi diagram to obtain 
a data structure with $O(\alpha(n))$ amortized query time 
and $O(\log\Psi \log^{4} n)$ amortized expected update 
time, see \cref{thm:bounded:insertion}.

In the decremental setting, 
a main challenge is to identify those edges in $\cD(S)$ 
that were incident to a freshly removed site and whose removal 
changes the connectivity in $\cD(S)$. 
To address this, we first develop a data structure for
a related dynamic geometric problem which might be of 
independent interest.
In this problem, we have two sets $R$ and $B$ of
disks in the plane, such that the disks in $R$ 
can be both inserted and deleted, while the
disks in $B$ can only be deleted.
We would like to maintain $R$ and $B$ in
a data structure such that whenever we delete
a disk $b$ from $B$,
we receive a list of all the disks in the current set $R$ 
that intersect the disk $b$ but no remaining disk
in $B \setminus \{b\}$.
We say that these are the disks in $R$ that are
\emph{revealed} by the deletion of $b$.
If $B$ initially contains $n$ disks, 
we can process a sequence of $m$ updates to $R$ and
 $k$ deletions from $B$ in
expected time 
$O(n \log^2 n  + 
m \log^4 n + 
k \log^4 n)$, assuming an oblivious adversary.
We call this data structure the 
\emph{disk revealing data structure} (\RDS{}) and it can be found in \cref{thm:disk-reveal}.

The \RDS{} plays a crucial part in developing 
decremental connectivity structures for disk graphs
of bounded radius ratio and for general disk graphs.
For both cases, we obtain data structures with 
$O(\log n/\log \log n)$ query time. 
The total expected time is 
$O(n\log \Psi \log^{4} n )$ 
for bounded radius ratio (\cref{thm:bounded:deletion}) and 
$O(n\log^{5}n )$ 
 for the general case (\cref{thm:unbounded:deletion}), assuming an oblivious adversary for both data structures.
Our contributions and the current state of the art 
are summarized in  \cref{tab:intro:runningtimes}.
A simplified version of the decremental data structure for general disk graphs also implies a static connectivity data structure with \(O(n\log^2n)\) preprocessing and \(O(1)\) query time (\cref{lem:static:general}).

\section{Preliminaries}%
\label{sec:preliminiaries}
After formally defining our geometric setting,
we briefly recall some basic notions and 
structures that will be relevant throughout the paper.

\subsection{Problem setting}
Let $S\subset \R^2$ be a set of $n$ \emph{sites} in the plane.
Every $s \in S$ has an \emph{associated radius} $r_s \in \R$, $r_s > 0$,
and it defines a closed disk $D_s$ with center $s$ and radius $r_s$.
The \emph{disk graph} $\cD(S)$ of $S$ is the graph with 
vertex set $S$ and an undirected edge $st$ if and only if 
$\| st\| \leq r_s+r_t$.
Here, $\| st\|$ denotes the Euclidean distance between $s$ and
$t$.
Equivalently, there is an undirected edge $st$ if and only if 
the disks $D_s$ and $D_t$ intersect.
Two sites $s$, $t$ are 
\emph{connected} in $\cD(S)$ if and only if $\cD(S)$ contains a path 
between $s$ and $t$.

We consider three types of disk graphs: 
\emph{unit disk graphs}, \emph{disk graphs of bounded radius ratio}, and 
\emph{general disk graphs}. 
In unit disk graphs, we have $r_s = 1$, for all $s \in S$. 
In disk graphs of bounded radius ratio, we require that 
$r_s \in [1, \Psi]$, for all $s \in S$, where $\Psi$ is some fixed
parameter.
The time bounds now may depend on $\Psi$.
Note that $\Psi$ in turn may depend on $n$, i.e., 
be $n^2, 2^n$, or even larger. 
For simplicity, we assume that $\Psi$ is known to 
the data structure in advance.
However, we believe that even if $\Psi$ is not previously known, 
most of our approaches 
can be adapted to achieve the same running times with 
regard to the maximum radius ratio over the life span of the data structure.
In general disk graphs, there are no additional
restrictions on the radii, and the running times depend
only on $n$.

In all three scenarios, the site set $S$ is \emph{dynamic},
i.e., it may change over time.
Our goal is to maintain $S$ while allowing 
\emph{connectivity queries} of the following form: 
\emph{given two sites $s, t \in S$, are they connected in
the current disk graph $\cD(S)$}?
Depending on the exact nature of how $S$ can change, we 
distinguish between the \emph{incremental}, the \emph{decremental}, 
and the \emph{fully dynamic} setting. 
In the incremental setting, we start with $S = \emptyset$, 
and we can only add new sites to $S$. 
In the decremental setting, we start with a given set 
$S$ of $n$ sites that may be subject to preprocessing, 
and we allow only deletions from $S$.
In the fully dynamic case, we start
with $S  = \emptyset$, and sites can
be inserted and deleted arbitrarily.
In all three settings, the updates can 
be interleaved with the queries in any
possible way.

\subsection{Data structures for edge updates}
We rely on existing data structures
that support connectivity queries and
edge updates on general graphs. 
In the incremental case, the problem reduces to
disjoint-set-union~\cite{cormen_introduction_2009}.  
The construction seems to be folklore (see e.g.\@ Holm 
et al.~\cite{holm_poly-logarithmic_2001}, 
Reif~\cite{reif_topological_1987}),
we state it in the following theorem, and provide 
a proof for completeness.

\begin{theorem}
\label{thm:generalinc}
Starting from the empty graph, there 
is a deterministic data structure for incremental
dynamic connectivity such that an isolated vertex can be added in 
$O(1)$ time, an edge between two existing vertices can be 
added in $O(\alpha(n))$ amortized time, and a connectivity query takes 
$O(\alpha(n))$ amortized time, where $n$ is the total 
number of inserted vertices so far and $\alpha(n)$ denotes
the inverse Ackermann function.
\end{theorem}

\begin{proof}
This is a direct application of a disjoint-set-union 
structure with the operations \textsc{Make-Set}, \textsc{Union}, 
and \textsc{Find}~\cite{cormen_introduction_2009}.
The vertices of $G$ correspond to the elements of the sets. 
An insertion of an isolated vertex becomes a 
\textsc{Make-Set} operation and takes $O(1)$ time.
To determine if two vertices are connected, we perform a 
\textsc{Find} operation on the corresponding elements, 
and we return \emph{yes} if they lie in the same set.
This has amortized running time $O(\alpha(n))$.
To insert an edge $uv$, we 
first find the sets that contain $u$ and $v$, and 
then we perform a \textsc{Union} operation on these sets.
This takes $O(\alpha(n))$ amortized time. 
\end{proof}

For the fully dynamic case,  we use 
the following result of Holm, Lichtenberg,
and Thorup~\cite{holm_poly-logarithmic_2001}.
We refer to this data structure as the \emph{HLT-structure}.
\begin{theorem}[Holm et al.~{\cite[Theorem~3]{holm_poly-logarithmic_2001}}]
\label{thm:dynamicspanningtree}
Let $G$ be a graph with $n$ vertices and initially empty  
edge set. There is a deterministic 
fully dynamic data structure so that edge updates in $G$ take
amortized time $O(\log^2 n)$
and connectivity queries take worst-case time 
$O(\log n/\log \log n)$. 
\end{theorem}

\cref{thm:dynamicspanningtree} assumes that $n$ is
fixed. However, 
we can easily support vertex insertions and deletions
within the same amortized 
time bounds, by rebuilding the data structure
whenever the number of vertices changes by a factor of $2$.
Thorup presented a variant of \cref{thm:dynamicspanningtree} 
that uses $O(m)$ space, where $m$ 
is the number of edges that are currently stored in the
data structure~\cite{thorup_near-optimal_2000}.

It should also be mentioned that both 
\cref{thm:generalinc} and \cref{thm:dynamicspanningtree} 
can easily be extended to not only answer connectivity queries 
but also to maintain the number of connected components.
As both data structures require some form of explicit operation 
when merging or splitting a connected component, a counter for the number 
of current components can be maintained and returned in $O(1)$ time.
This directly implies that all except one of our disk connectivity 
data structures can be extended to also support those queries in $O(1)$ 
time.
This is not directly clear for the data structure of 
\cref{thm:bounded-radius-ratio-psi}, as this data 
structure does not maintain the full connectivity 
internally for an improved update time.

\subsection{The hierarchical grid and Quadtrees}
\label{sec:hier-grid}
\paragraph*{The hierarchical grid}
To structure our set of sites, we will make frequent use
of the hierarchical grid and quadtrees.
Let $i \geq 0$. A \emph{cell $\sigma$ of level $i$} is
an axis-parallel square with diameter $2^i$. On the boundary,
the cell $\sigma$ contains the right edge without its bottom 
endpoint and the upper edge
without its left endpoint,
but none of the remaining points.\footnote{We emphasize that
the diameter of $\sigma$ is defined with respect to the
closure, i.e., $\sigma$ has diameter $2^i$, but any pair of
points that lie in $\sigma$ have distance strictly less than
$2^i$.}
The \emph{grid} $\G_i$ is the set of all cells
of level $i$ such that $\G_i$ contains the cell
with the origin as the upper right corner; the cells in $\G_i$ cover
the plane; and the closures of any two cells in $\G_i$  
are either disjoint or share exactly one corner or
a complete boundary edge.
Note that every point in the
plane lies in exactly one cell of $\G_i$.
The 
\emph{hierarchical grid} $\G$ is then defined as 
$\G = \bigcup_{i = 0}^{\infty} \G_i$.
For any cell $\sigma \in \G$, we denote by 
$|\sigma|$ its diameter and by $a(\sigma)$ its center.
We say that 
$\G_i$ is the \emph{grid at level $i$} in $\G$. 
We assume that given a point $p \in \R^2$ and a level $i \geq 0$,
we can find (the coordinates of) the cell of $\G_i$
that contains $p$ in $O(1)$ time. 
Furthermore, for a cell $\sigma \in \G_i$ and odd $k \in \mathbb{N}$,
the \emph{$(k\times k)$-neighborhood} of $\sigma$,
$N_{k\times k}(\sigma)$,
is the set of
$k^2$ cells in $\G_i$ that contains $\sigma$ and all cells that
can be reached from $\sigma$ by crossing at most $(k-1)/2$ vertical
and at most $(k-1)/2$ horizontal cell boundaries.
See \cref{fig:hiergrid} for an illustration.

\begin{figure}
\begin{center}
\includegraphics{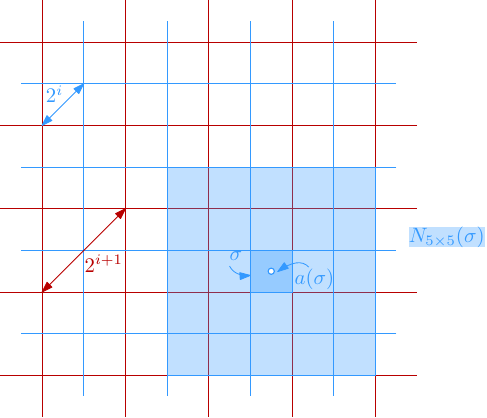}
\end{center}
\caption{Two levels of the hierarchical grid, a cell
$\sigma$ at level $i$, its center $a(\sigma)$, 
and the $(5 \times 5)$-neighborhood
of $\sigma$.}
\label{fig:hiergrid}
\end{figure}

\paragraph*{Quadtrees}
Let $\mathcal{C} \subset \G$ be a finite set of cells.
The \emph{quadtree} $\mathcal{T}$ for $\mathcal{C}$ is a rooted 
$4$-ary 
tree whose nodes are cells from $\G$.
The root of $\mathcal{T}$ is the smallest cell $\rho$
in $\G$  that
contains all the cells of $\mathcal{C}$. 
Starting from $\rho$, we expand the nodes of $\mathcal{T}$ as follows:
if a cell $\sigma$ in $\mathcal{T}$ with level $i \geq 1$,
properly contains at least one cell of $\mathcal{C}$, then
$\sigma$ obtains four children, namely
the cells $\tau_j$ with  level $i-1$ and $\tau_j \subseteq \sigma$.
If a cell $\sigma$ in $\mathcal{T}$ does not properly contain 
a cell of $\mathcal{C}$, 
it is not expanded any further, and it becomes a leaf of $\mathcal{T}$.
Typically, we do not explicitly 
distinguish between a cell $\sigma$ and its associated 
node in $\mathcal{T}$.
A quadtree $\mathcal{T}$ on a given set of $n$ cells can be constructed in 
$O(n\log |\rho| )$ time, where $\rho$ is the root 
of $\mathcal{T}$.\footnote{Usually, 
quadtrees are defined for planar point sets and not
for cells in a hierarchical grid~\cite{berg_computational_2008}. 
However, for our purposes,
it is more convenient to work at the cell level, since sites that are close
together form cliques in $\cD(S)$ and can thus be represented by 
a single cell
that contains them.}

\subsection{Maximal bichromatic matchings}
We make frequent use of a
data structure that dynamically maintains a 
\emph{maximal bichromatic matching} (\MBM) between
two sets of disks.
Let $R \subseteq S$ and $B \subseteq S$ be two disjoint non-empty sets
of sites, and let $(R \times B) \cap \cD(S)$ be the bipartite graph on 
$R$ and $B$ that consists of all edges of $\cD(S)$ 
with one vertex in $R$ and one vertex in $B$.
An \MBM between $R$ and $B$ is a
maximal set of vertex-disjoint edges in $(R \times B) \cap \cD(S)$.

We show how to maintain a dynamic \MBM as sites are
inserted or deleted in $R$ and in $B$. For this, we need
a dynamic data structure for \emph{disk unions}:
let $T \subseteq S$ be a set of sites. A \emph{disk union structure}
dynamically maintains $T$ as sites are inserted or deleted,
while supporting the following query: given a site $s \in S$,
report an arbitrary site $t \in T$ such that $D_s$ and $D_t$
intersect, or indicate that no such $t$ exists.
Given a disk union structure, we can maintain an
\MBM with only constant overhead:

\begin{lemma}[Kaplan et al.~{\cite[Lemma~9.10]{kaplan_dynamic_2020}}]
    \label{lem:dg-mbm}
    Let $R,B \subseteq S$ be two disjoint sets with a total of at most $n$ 
    sites. Suppose we have a dynamic data structure for 
    disk unions in $R$ and $B$  
    that has update time $U(m)$, query time $Q(m)$, and 
    space requirement $S(m)$, for $m$ elements.
    Then, there exists a dynamic data structure that maintains 
    an \MBM for $R$ and $B$ with 
    $O(U(n) + Q(n) + \log n)$ update time,
    using $O(S(n) + n)$ space.
\end{lemma}

\begin{proof}
    We store \MBM $M$ as a red-black tree~\cite{cormen_introduction_2009}.
    Each edge of the matching is stored twice, using as keys both 
    of its vertices, and the search order is any sensible order.
    We use 
    two disk union data structures $D_R$ and $D_B$ 
    for the sites from $R$ and $B$ that currently do not appear
    as an endpoint in $M$.
    The invariant is that $M$ contains a matching between $R$ and $B$,
    and that the disks in $D_R$ and $D_B$ are pairwise disjoint.
    This implies  that $M$ is maximal.

    To insert a new site $r$ into $R$, we 
    perform a query with $r$ in $D_B$.
    If this query reports a site  $b \in B$ such that $D_r$ and $D_b$
    intersect, we add the edge $rb$ to $M$, and
    we delete $b$ from $D_B$.
    Otherwise, we insert $r$ into $D_R$.
    To delete a site $r$ from $R$, there are two
    cases: first, if $r$ is unmatched,  we simply remove $r$ from $D_R$.
    Second, if  $r$ is matched, say to the site $b \in B$, 
    we proceed as follows:
    we remove the edge $rb$ from $M$ and we query $D_R$ with $b$,
    looking for a new partner for $b$ in $R$. If we find a new partner $r'$
    for $b$, we delete $r'$ from $D_R$, and we add the edge $r'b$ to $M$. 
    Otherwise,
    we insert $b$ into $D_B$.  Updates to $B$ are analogous.
    Thus, an update to the \MBM  requires $O(1)$ insertions, deletions, or
    queries to $D_R$ or $D_B$ and $O(1)$ queries to $M$, and the lemma follows.
\end{proof}

We use two ways to implement the dynamic structure
for disk unions, based on two different underlying data
structures. For the general case, 
we  rely on an \emph{additively-weighted nearest neighbor 
data structure} (\AWNN).
An \AWNN{} stores a set 
$P \subset \R^2$ of
$n$ points in the plane, such that every point $p \in P$
is assigned a \emph{weight} $w_p \in \R$. The \AWNN{}
supports \emph{additively-weighted nearest neighbor queries}:
given a point $q\in \R^2$, find the point $p \in P$ that minimizes
the additively weighted Euclidean distance $\| pq \| + w_p$ to $q$.
Kaplan et al.~\cite{kaplan_dynamic_2020}, with the improved
construction of shallow cuttings by 
Liu~\cite{Liu20}, 
show the following result:
\begin{lemma}[Kaplan et al.~{\cite[Theorem 8.3, Section~9]{kaplan_dynamic_2020}, Liu~\cite[Corollary~4.3]{Liu20}}]\label{lem:prelims:dynamicNN}
There is a fully dynamic \AWNN{} data structure that 
allows insertions in $O(\log^2 n)$ amortized 
expected time and deletions in $O(\log^{4} n)$
amortized expected time. 
Furthermore, a query takes $O(\log^2 n)$ worst case time.
The data structure requires $O(n \log n)$ space.
\end{lemma}

\cref{lem:prelims:dynamicNN} can be used directly to implement a dynamic
disk union structure: we assign to each site $t \in T$ the
weight $w_t = -r_t$. To check if a disk $D_s$ intersects a disk in $T$,
we query the \AWNN structure with $s$. Let $t \in T$ be 
the resulting additively-weighted
nearest neighbor. We check if the disks $D_s$ and $D_t$
intersect. If so, we return $t$, otherwise, we report that no such
intersection exists. Combining \cref{lem:dg-mbm,lem:prelims:dynamicNN}, 
we get:
\begin{lemma}
    \label{lem:mbm-general}
    Let $R,B \subseteq S$ be two disjoint sets with a total of at most $n$ 
    sites. 
    Then, there exists a dynamic data structure that maintains 
    an \MBM for $R$ and $B$ with 
    $O(\log^{4} n)$ amortized expected update time,
    using $O(n \log n)$ space.
\end{lemma}

Next, we consider the more restricted case that 
$S$ contains only unit disks, and that the sites
in $R$ and $B$ are separated by a vertical or horizontal
line $\ell$ that is known in advance. In this case, we can
implement a disk union structure using a \emph{dynamic lower envelope}
(\DLE) structure for pseudolines.
Indeed, suppose we have 
a dynamic set $T$ of sites with $r_t = 1$, for all $t \in T$. 
Assume further that we are given a vertical
or horizontal line $\ell$  such that all query sites $s$ 
have $r_s = 1$ and are separated from $T$ by $\ell$.
We rotate and translate everything such that $\ell$ is the $x$-axis and all 
sites in $T$ have 
positive $y$-coordinate.
We consider the set $U_T$ of disks with radius $2$ and centers
in $T$ (see Figure~\ref{fig:lower-envelope}).
\begin{figure}
\centering
\includegraphics{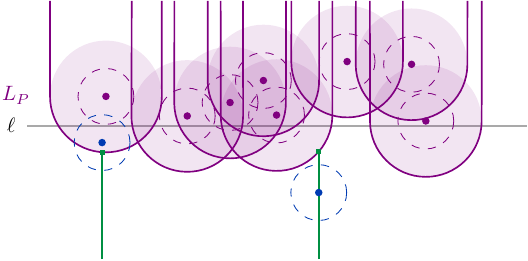}
\caption{The set $L_P$ induced by $P$. The unit disks are 
drawn dashed. If a site $b \in B$ lies above the lower 
envelope, the unit disks intersect.}
\label{fig:lower-envelope}
\end{figure}
Then, a site in $s$ that lies below $\ell$ and has $r_s = 1$ 
intersects some disk in $T$ if 
and only if it is contained in the union of the disks in $U_T$. To detect 
this, we maintain the lower envelope of $U_T$,
with a \DLE-structure for pseudolines. 

This is defined as follows:
let $L$ be a set of \emph{pseudolines} in the plane, i.e., 
each element of $L$ is a simple continuous $x$-monotone 
curve, and any two distinct curves in $L$ cross in exactly 
one point (and do not meet anywhere else). 
The \emph{lower envelope} of $L$ is the pointwise 
minimum of the graphs of the curves in $L$.
Combinatorially, it consists of a sequence of 
\emph{pseudoline segments}, such that each pseudoline
contributes at most one segment. 
The goal is to maintain a representation of this lower envelope
efficiently as pseudolines are inserted and deleted in $L$.
At the same time, the structure should support
\emph{vertical ray shooting queries}: given a query $q \in \R$,
report the pseudoline(s) for the segment(s) of the current
lower envelope that intersect the vertical line with $x$-coordinate
$q$.
Overmars and van Leeuwen show how to maintain the \textsf{DLE}
of a set of lines with update time $O(\log^2 n)$ such that 
\emph{vertical ray shooting} queries take 
$O(\log n)$ time~\cite{overmars_maintenance_1981}.
This was extended to pseudolines by
Agarwal et al.~\cite{agarwal_dynamic_2019}:
 
\begin{lemma}[Agarwal et al.~\cite{agarwal_dynamic_2019}]
\label{lem:lowerenvelope}
We can maintain the \textsf{DLE}
of a set of pseudolines with $O(\log^2 n)$
worst-case update time and $O(\log n)$ worst-case time
for vertical ray shooting queries.
Here, $n$ denotes the maximum number of pseudolines 
in the structure. The data structure uses $O(n)$ space.
\end{lemma}

Now, consider 
the following set $L_T$ of pseudolines: for each disk $D$ of $U_T$, take the 
arc that defines the lower part of the boundary of $D$ and extend 
both ends upward to $\infty$, with a very 
high
slope. 
By maintaining $L_T$, we can
implement a disk union structure for our special case with
worst-case update time $O(\log^2 n)$ and worst-case query time $O(\log n)$.
Combining with \cref{lem:dg-mbm,lem:lowerenvelope},
we get:
\begin{lemma}
    \label{lem:mbm-unit}
    Suppose that $r_s = 1$, for all sites $s \in S$.
    Let $R,B \subseteq S$ be two disjoint sets with a total of at most $n$ 
    sites, such that there is a known vertical or horizontal
    line that separates $R$ and $B$.
    Then, there exists a dynamic data structure that maintains 
    an \MBM for $R$ and $B$ with 
    $O(\log^{2} n)$ worst-case update time,
    using $O(n)$  space.
\end{lemma}

\subparagraph*{Remark.}
By now, there are several improvements
of the structure of  Overmars and van 
Leeuwen~\cite{chan_dynamic_2001,kaplan_faster_2001,brodal_dynamic_2002}, 
culminating in an optimal $O(\log n)$ amortized update time 
due to Brodal and Jacob~\cite{brodal_dynamic_2002}.
There is solid evidence~\cite{hofer_fast_2017} 
that the improvement of 
Kaplan et al.~\cite{kaplan_faster_2001} carries over to the pseudoline 
setting, giving a better $O(\log n\log \log n)$ amortized update 
time with $O(\log n)$ query time. 
However, there is no formal presentation of these arguments yet,
so we refrain from formally stating the result.
The structure by Brodal and 
Jacob~\cite{brodal_dynamic_2002} is very involved,
and we were not able to verify if it carries over to the
setting of pseudolines. 
This poses an interesting 
challenge for further investigation. 

\subsection{Computational Model}
\label{sec:compmodel}

Our algorithms use quadtrees and hierarchical grids.
Even though this is a common practice in the computational geometry
literature (see, e.g., Har-Peled's book~\cite{har-peled_geometric_2011}
for many examples of such algorithms), we should be aware
that this creates some issues: in order to use hierarchical grids
of arbitrary resolution, we typically need to determine which
grid cell of a certain diameter contains a given point.
To do this efficiently, we need to extend the standard
computational model in computational geometry, the real 
RAM~\cite{PreparataS85}, by an additional \emph{floor function} 
that rounds a given
real number down to the next integer~\cite{har-peled_geometric_2011}. 
Even though this floor
function is natural---and implemented in many real-world computers---it 
is problematic in the context of the real RAM: the floor function
together with the unbounded precision of the real RAM provide us with
a very powerful model of computation 
that can even solve PSPACE-complete problems
in polynomial time (see, e.g., Mitchell and Mulzer~\cite{MulzerMi17} and the references
therein for further discussion of this issue).
Thus, when using this model, we should be careful to use the floor function
in a ``reasonable'' way. Typically, it is considered reasonable
to stipulate an operation that allows us to find the cell of a given
level of the hierarchical grid that contains a given input point in constant 
time\cite{har-peled_geometric_2011}.
In the main part of this paper, we will follow this approach, 
because it will lead to a clearer  description of the algorithms.
Note that this is mainly an issue for insertion operations in the
semi-dynamic incremental and the fully dynamic setting, because if the
points are given in advance (as in the decremental setting), 
we have time to preprocess them to find
the associated grid cells. Furthermore, for the query operations, we 
can store the associated grid cell for each point in the data structure
with its satellite data, so that we do not need to determine this cell
with the help of the floor function during the query.

However, it typically turns out that the use of the floor-function is
not strictly necessary. The price we need to pay for this is that
the grid cells are no longer perfectly aligned,  which makes the algorithms
a bit more messy. However, these issues can typically be dealt with
no or little overhead (see, e.g.,~\cite{BuchinLMM11,LofflerM12}
for examples). This is also true for our algorithms, and we give the
details in the appendix, in order to satisfy the curious reader while
keeping the main text free from additional complications.

\subsection{The Role of Randomness}

Most of our algorithms use randomness, and some of the 
running time guarantees hold in expectation only. We would like
to emphasize that the expectation is over the randomness in the algorithms
only, and they hold in the worst case over all possible query sequences.
However, sometimes, in particular in the context of the
disk revealing structure, we will need the \emph{obliviousness assumption}.
This assumption says that the query sequence must be \emph{oblivious} to
the random choices of the algorithm and may not adapt to the state of
the data structure. This is a typical assumption in the analysis of
data structures and is used, e.g., in the analysis of hashing-based
structures. However, we should mention that in recent years a lot of
effort has been put into designing randomized data structures that
do not rely on the obliviousness assumption and that it is an intriguing
open problem to obtain a disk revealing structure that does not need
this assumption.

\section{Unit disk graphs}%
\label{sec:unit}

We first consider the case of unit disk graphs.
As mentioned in the introduction, this problem
was already addressed by Chan et 
al.~\cite{chan_dynamic_2011}.
They observed how to combine several known
results into a data structure
for connectivity queries in
fully dynamic unit disk graphs with amortized update time 
$O(\log^{6}n)$ and query time $O(\log n/\log \log n)$.
We briefly sketch their approach (see \cref{fig:approach}):
\begin{figure}
\centering
\includegraphics{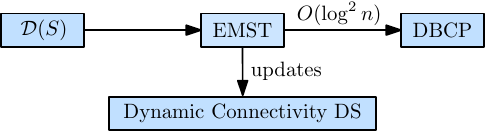}
\caption{A solution for unit disks with $O(\log^{6} n)$ update time.}
\label{fig:approach}
\end{figure}
let $T$ be the Euclidean minimum spanning tree 
(\textsf{EMST}) of $S$. If we 
remove from $T$ all 
edges with length more than $2$, the resulting forest 
$F$ is spanning for $\cD(S)$. Thus, to maintain the 
components of $\cD(S)$, 
it suffices to maintain the components of $F$. 
For this, we use an HLT-structure $\mathcal{H}$.
The geometric properties of the \textsf{EMST} ensure that 
inserting or deleting a site from $S$ modifies 
$O(1)$ edges in $T$. Thus, if we can efficiently 
find the set $E$ of edges in $T$ that change
during an update, we can maintain the components of $F$ through
$O(1)$ updates to $\mathcal{H}$, by considering all edges in $E$ 
of length at most $2$.
To find $E$, we need to 
dynamically maintain the \textsf{EMST} $T$ of $S$.
For this, there is a technique 
of Agarwal et al.~that reduces the problem to 
several instances of the \emph{dynamic bichromatic closest 
pair problem} (\textsf{DBCP}),
with an overhead of $O(\log^2 n)$ in the update 
time~\cite{agarwal_euclidean_1991}. 
Recently, Chan~\cite{Chan20a} showed how to solve the \textsf{DBCP}-problem
with a modified 
dynamic nearest neighbor (\textsf{DNN}) structure, 
in $O(\log^4 n)$ amortized update time. These two results
together allow us to maintain the \textsf{EMST} dynamically
with an amortized update time of $O(\log^6 n)$.\footnote{
At the time of Chan et al.'s original paper~\cite{chan_dynamic_2011},
Chan's improvement~\cite{Chan20a} was not available
yet. Instead, Chan et al.~\cite{chan_dynamic_2011} use a result by 
Eppstein that reduces the 
\textsf{DBCP} problem to 
several instances of the \textsf{DNN} problem
for points in the 
plane~\cite{eppstein_dynamic_1995}, 
with another $O(\log^2 n)$ factor as
overhead in the update time. Using Chan's original
\textsf{DNN} structure~\cite{chan_dynamic_2010} with 
amortized 
expected update time $O(\log^6 n)$, this leads to a total amortized expected update time of 
$O(\log^{10} n)$.} We can use $\mathcal{H}$ for 
queries in $O(\log n/\log \log n)$ time. 

To improve over this result, we replace the \textsf{EMST} by 
a simpler graph that still captures the connectivity of $\cD(S)$.
We also replace the \textsf{DNN} structure 
by a suitable maximal bichromatic matching (\MBM) structure that
is based on dynamic lower envelopes (\cref{lem:mbm-unit}).
These two changes improve the amortized update time
to
$O(\log^2 n)$, without affecting the query time.
 \cref{fig:approach2} shows the structure of our technique.

\begin{figure}
 \centering
 \includegraphics[page=2]{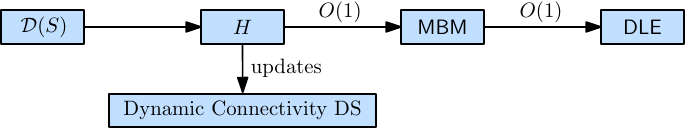}
 \caption{The structure of our data structure for unit disks.}
 \label{fig:approach2}
 \end{figure}

First, we define a \emph{proxy graph} $H$ that 
represents the connectivity of $\cD(S)$.
This proxy graph groups close unit disks by a grid.
It is similar to other grid-based proxy graphs used to solve  problems related to unit disk graphs, e.g.\@ the construction of hop spanners~\cite{catusse_planar_2010}.
For a current set $S \subset \R^2$ of sites, the  vertices 
of $H$ are those cells of the level-1 grid $\G_1$ 
(cf.~\cref{sec:hier-grid}) that contain at least one site
of $S$, i.e.,
the cells $\sigma  \in \G_1$
with $\sigma \cap S \neq \emptyset$.
We call these cells \emph{non-empty}, and we 
say that a site $s \in S$ is \emph{assigned} 
to the cell $\sigma \in \G_1$ that contains it.
We let $S(\sigma)$ denote the sites that are
assigned to $\sigma$.
Two distinct non-empty cells $\sigma$, $\tau$ are connected 
by an edge in $H$ 
if and only if there is an edge $st\in \cD(S)$ with $s \in S(\sigma)$ 
and $t \in S(\tau)$.
The following lemma states that the proxy graph $H$ 
is sparse and that it represents the connectivity in $\cD(S)$:

\begin{lemma}
\label{lem:gridreachability}
The proxy graph $H$ has at most $n$ vertices, each with degree $O(1)$.
Two sites $s, t\in S$ 
are connected in $\cD(S)$ if and only if their assigned cells 
$\sigma$ and $\tau$ are connected in $H$.
\end{lemma}

\begin{proof}
Since every vertex of $H$ has at least one site of $S$ assigned to it,
and since every site is assigned to exactly one vertex,
there are at most $n$ vertices in $H$.
We say that two cells $\sigma$ and $\tau$ in $\G_1$ 
are \emph{neighboring} 
if $\tau \in N_{5\times 5}(\sigma)$.
Then, a cell $\sigma$ can be adjacent in $H$ only to the $O(1)$ 
neighboring cells of $\sigma$, as the distance to all other cells 
is larger than $2$. 

Next, we show
that (i) for every edge $st$ in $\cD(S)$, the assigned
cells $\sigma$ of $s$  and $\tau$ of $t$  are connected in $H$;
and  (ii) for every edge $\sigma\tau$ in $H$, all sites in $S(\sigma)$
can reach all sites in $S(\tau)$ in $\cD(S)$. This immediately 
implies the claim about the connectivity, because we can then map
paths in $\cD(S)$ to paths in $H$, and vice versa.
To prove (i), let $st$ be an edge of $\cD(S)$, and
consider the assigned cells
$\sigma, \tau$ with $s \in S(\sigma)$ and $t \in S(\tau)$. 
If $\sigma = \tau$, there is nothing to show.
If $\sigma \neq \tau$, the definition of $H$ immediately implies that 
the edge $\sigma\tau$ exists in $H$, and hence $\sigma$
and $\tau$ are connected.
For (ii),
we first 
note that for every vertex  
$\sigma$ in $H$,
the sites in $S(\sigma)$
induce a clique in $\cD(S)$.
Indeed, for any $s, t \in S(\sigma)$, 
we have $\| st\| < |\sigma| = 2$, 
so the unit disks $D_s$ and $D_{t}$
intersect, and $st$ is an edge in $\cD(S)$.
Now, let
$\sigma\tau$ be an edge  in $H$. 
By definition of $H$, there is at least one pair $s, t \in S$ 
with $s \in S(\sigma)$, and  $t \in S(\tau)$ such that $st$ is an
edge of $\cD(S)$.
Then, as the sites in $S(\sigma)$ and the sites in $S(\tau)$ each
form a clique in $\cD(S)$, all sites in $S(\sigma)$ are connected 
to $s$, 
and  
all sites in $S(\tau)$ are connected to $t$, so every
site in $S(\sigma)$ can reach every site in $S(\tau)$
by at most three steps in  $\cD(S)$.
The claim follows.
\end{proof}

In our connectivity structure, we maintain an HLT-structure $\cH$ 
for $H$.
To determine the connectivity between two sites $s$ and $t$, 
we first identify the cells $\sigma$ and $\tau$ in $\G_1$
to which $s$ and $t$ are assigned. This requires $O(1)$ time,
because as mentioned in~\cref{sec:compmodel},
we can store the assigned cell for $u$ in its satellite data 
during the insertion of $u$. 
Then, we query $\cH$ with the vertices $\sigma$ and $\tau$ of $H$, 
which, by \cref{lem:gridreachability}, yields the correct answer.
When a site $s$ is inserted into or deleted from $S$, 
only the edges incident to the assigned cell $\sigma$ of $s$ are affected.
By \cref{lem:gridreachability}, there are only $O(1)$ such edges.
Thus, once the set $E$ of these edges is determined, 
by \cref{thm:dynamicspanningtree}, we can update
$\mathcal{H}$ in time $O(\log^2 n)$.

We describe how to find the edges $E$ of $H$ that change when we 
update $S$.
For every pair $\sigma,\tau$ of neighboring cells in $\G_1$ 
where at least one cell is non-empty,
we maintain a maximal bichromatic matching (\MBM)
$M_{\{\sigma,\tau\}}$ for $R = S(\sigma)$ and $B = S(\tau)$,
as in \cref{lem:mbm-unit} (note that the special
requirements of the lemma apply here). 
By definition, 
$\sigma\tau$ is an edge of  $H$ if and only if 
$M_{\{\sigma,\tau\}}$ is not empty. 
When inserting or deleting a site $s$ from $S$, we proceed as 
follows: let $\sigma \in \G_1$ be the cell assigned to $s$. We go 
through all cells $\tau \in N_{5\times 5}(\sigma)$, and we update 
$M_{\{\sigma,\tau\}}$ by inserting or deleting $s$ from the relevant set
(creating or removing the underlying \MBM structure if necessary). 
If the matching $M_{\{\sigma,\tau\}}$ becomes non-empty during an insertion or 
empty during a deletion, we add the edge $\sigma\tau$ to $E$ and mark it 
for insertion or remove $\sigma\tau$ from E and mark it for deletion, respectively.
By \cref{lem:mbm-unit}, these updates to the \MBM-structures take
$O(\log^2 n)$ time.
Putting everything together, we obtain the main result of this section:

\begin{theorem}
\label{thm:dynamicUDG}
There is a dynamic connectivity structure for
unit disk graphs 
such that
the insertion or deletion of a site takes amortized time 
$O\left(\log^2 n\right)$
and  a connectivity query takes worst-case time
$O(\log n/\log \log n)$,
where $n$ is the maximum number of sites at any time. 
The data structure requires $ O(n)$ space.
\end{theorem}
\begin{proof}
The main part of the theorem follows from the discussion so far and
\cref{lem:gridreachability,lem:mbm-unit}.
For the space bound, note that the \MBM{}s have size linear in the number
of involved sites, and every site lies
in a constant number of cells and \MBM{}s.
Also, the total number of edges in $H$ is  $O(n)$, 
so the HLT-structure requires linear space as well.
\end{proof}

\subparagraph*{Remark.}
Following the remark regarding \cref{lem:mbm-unit}, there is evidence that the update time of \cref{thm:dynamicUDG} can be improved.
Using a better \MBM{} data structure and a faster data structure for maintaining the proxy graph~\cite{huang_fully_2017,wulff-nilsen_faster_2013} would lead to an improved update time (with possibly worse query and space bounds, depending on the data structures chosen).

\section{Polynomial dependence on \texorpdfstring{\(\Psi\)}{psi}}%
\label{sec:poly-dependence}

We extend our structure from \cref{thm:dynamicUDG}
to intersection graphs of arbitrary disks.
Now, the running times will depend
polynomially on the radius ratio $\Psi$. The general approach is unchanged, 
but the varying radii
of the disks introduce new issues.

Again, we use a grid to structure the sites, but
instead of the single grid $\G_1$, we rely on 
the first $\lfloor \log \Psi \rfloor + 1$ 
levels of the hierarchical grid $\G$ (see \cref{sec:hier-grid}). 
Each site $s$ is assigned 
to a grid level that is determined by the associated radius $r_s$.
As in the other sections, we assume without loss of generality that all those cells are part of the same globally aligned grids.
In \cref{sec:app:quadforest} we give the details on why this assumption can be made without negatively impacting the running time.

Since the disks have different sizes,
we can no longer use \cref{lem:mbm-unit} to maintain
the maximal bichromatic matchings (\MBM{}s) between
neighboring grid cells. Instead, we use
the more complex structure from \cref{lem:mbm-general}.
This increases the overhead for updating the \MBM
for each pair of neighboring cells.\footnote{Actually,
instead of the \MBM{}s, we could also use the dynamic
bichromatic closest pair (\textsf{DBCP}) structure for general
distance functions of Kaplan et al.~\cite{kaplan_dynamic_2020}.
It achieves the same time bounds, but it is slightly more complicated.
Hence, we prefer to stay with \MBM{}s in our presentation.}
Furthermore, we will see that a disk can now intersect disks 
from $\Theta(\Psi^2)$ other cells, generalizing the $O(1)$-bound
from the unit disk case.
Thus, the degree of the proxy graph, and hence the number
of edges that need to be modified in a single update, depends on $\Psi$.
To address this latter problem---at least partially---we describe
in \cref{sub:dg-limit-insertions-to-mbm} how to refine
the definition of the proxy graph 
so that fewer edges need to be 
modified in a single update.
This will reduce the dependence on $\Psi$ in the update time
from quadratic to linear. 
The query procedure becomes slightly more complicated,
but the asymptotic running time remains unchanged.

We note that the approach in \cref{sub:dg-adapting-the-unit-disk-case} 
is similar to the method of Kaplan et 
al.~\cite[Theorem~9.11]{kaplan_dynamic_2020} that 
achieves the same time and space bounds.
However, our implementation uses a hierarchical
grid instead of a single fine grid. This will be 
crucial for the improvement in 
\cref{sub:dg-limit-insertions-to-mbm},
so we first describe the details of the modified approach.

\subsection{Extending the unit disk case}
\label{sub:dg-adapting-the-unit-disk-case}
As for unit disk graphs, we define a proxy graph $H$ 
for $\cD(S)$ that is
based on grid cells and the intersections between the disks.
The vertices of $H$ are 
cells from the first $\lfloor \log \Psi \rfloor + 1$
levels of the hierarchical grid $\G$.
More precisely, a site $s \in S$ is \emph{assigned} to the 
cell $\sigma \in \G$ with $s \in \sigma$ 
and $|\sigma| \leq r_s < 2 |\sigma|$.\footnote{Note 
that this differs from \cref{sec:unit}, where 
sites with radius $1$ are assigned to cells of diameter $2$.}
We define the \emph{level} of $s$ as the level of the
cell $\sigma$ that $s$ was assigned to.
As before in \cref{sec:unit}, we denote the set of sites assigned to a cell $\sigma$ by 
$S(\sigma)$.
It is still the case that all sites in $S(\sigma)$
form a clique in $\cD(S)$, since for $s, t \in S(\sigma)$,
we have $\|st\| \leq |\sigma| \leq r_s + r_t$.
The vertices of $H$ are all cells $\sigma$ of $\G$ 
with $S(\sigma) \neq \emptyset$.

As in \cref{sec:unit}, we connect two cells 
$\sigma, \tau \in \G$ by an edge in $H$
if and only if there are  assigned sites
$s \in S(\sigma)$ 
and $t \in S(\tau)$ such that
$st$ is an edge in $\cD(S)$.
Note that $\sigma$ and $\tau$ do not have to be on the same level.
We call a pair of cells $\sigma$, $\tau$ \emph{neighboring} 
if and only if it is possible that they could become
adjacent in $H$, i.e., if some disks in $S(\sigma)$
and $S(\tau)$ could intersect. This is the case if
the distance between $\sigma$ and $\tau$ is less than 
$2 | \sigma | + 2 | \tau |$, see \cref{fig:neighborhood-in-grids-quadtree}.
Since we are now dealing with cells from $\lfloor \log \Psi \rfloor + 1$
levels, 
the degree in $H$ depends on $\Psi$.

\begin{lemma}
    \label{lem:dg-graph-properties}
    The proxy graph $H$ has at most $n$ vertices.
    Let $\sigma$ be a vertex of $H$, and let  $\mathcal{N}(\sigma)
    \subseteq \G$ be the neighboring cells of $\sigma$.
    Then, we have $|\mathcal{N}(\sigma)| = O(\Psi^2)$, so $H$
    has maximum
    degree $O(\Psi^2)$.
    Two sites of $S$
    are connected in $\cD(S)$ if and only if their assigned 
    cells are connected in $H$.
\end{lemma}
\begin{proof}
    As in \cref{lem:gridreachability}, the bound on the
    vertices follows from the facts that
    every vertex in $H$ has at least one site assigned to it, 
    and that every site is assigned to exactly one vertex.
    Now, let $\sigma$ be a vertex of $H$, and 
    let $\mathcal{N}(\sigma)$ the set of neighboring cells for
    $\sigma$.
    Recall that all sites are assigned 
    to the lowest $\lfloor \log \Psi \rfloor + 1$ levels of $\G$.
    Let $\ell \in \{0, \dots,  \lfloor \log \Psi \rfloor\}$
    be such that $\sigma \in \G_{\ell}$.
    First, fix an $\ell'$ with  
    $\ell \leq \ell' \leq \lfloor \log \Psi \rfloor$, and let 
    $\sigma' \supseteq \sigma$ be 
    the cell at level $\ell'$ that contains $\sigma$.
    As can be seen by simple volume considerations, all cells in
    $\mathcal{N}(\sigma) \cap \G_{\ell'}$ lie in the neighborhood 
    $N_{13\times 13}(\sigma')$, so there
    are at most $13^2$ of them.
    Next, consider $\ell'$ with $0 \leq \ell' < \ell$. 
    All cells
    of  $\mathcal{N}(\sigma) \cap \G_{\ell'}$
    are contained in
    the region $N_{13 \times 13}(\sigma)$.
    Thus, since every cell of $\G_\ell$ contains exactly
    $4^{\ell - \ell'}$ cells of $\G_{\ell'}$, we have
    $|\mathcal{N}(\sigma) \cap \G_{\ell'}| \leq 13^2 \cdot 4^{\ell - \ell'}$.
    Altogether, $|\mathcal{N}(\sigma)|$ is at most
    \[
        13^2 \cdot \sum_{i=0}^{\lfloor \log \Psi \rfloor} 4^i 
	= O(4^{\log \Psi}) = O(\Psi^2).
    \]
    Now, the claim on the maximum degree of $H$ is
    immediate. The claim on the connectivity is shown verbatim as in
    the proof of \cref{lem:gridreachability}.
\end{proof}

\begin{figure}
    \centering
    \includegraphics{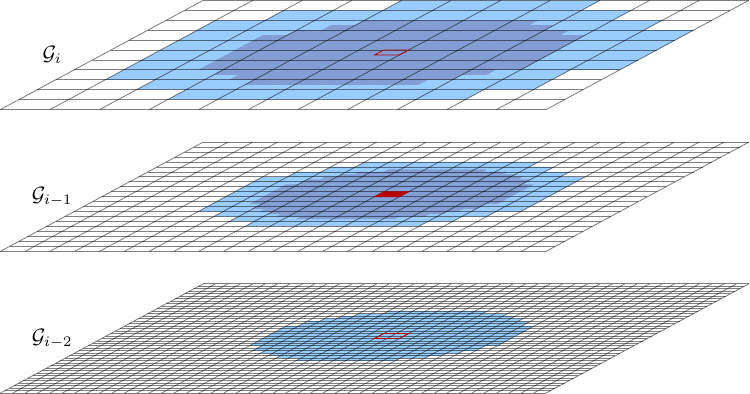}
    \caption{
        The neighboring cells of the red cell in $\G_{i-1}$.
        The area of the neighboring cells in one level beneath is 
	colored in a darker shade.
    }
    \label{fig:neighborhood-in-grids-quadtree}
\end{figure}

We now describe the details of our data structure. The main 
ingredient is a \emph{quadforest} $\mathcal{F}$.
The quadforest helps us find the relevant edges of the proxy graph
$H$ that need to be updated.
Let $\mathcal{C} = \bigcup_{\sigma \in H} \{\sigma\} \cup \mathcal{N}(\sigma)$
be the set of all vertices $\sigma$ in $H$ together with
their neighboring cells in $\G$.

Let $\mathcal{C}'$ be the set of all cells
in $\G_{\lfloor \log \Psi \rfloor}$ 
that contain
at least one cell from $\mathcal{C}$.
For each cell $\rho \in \mathcal{C}'$, we construct
a \emph{quadtree} $\mathcal{T}_{\rho}$ (cf.~\cref{sec:hier-grid}): 
we make  $\rho$ the
root of $\mathcal{T}_{\rho}$. If $\rho$ properly contains
a cell of $\mathcal{C}$, we take the four cells of $\G$
that lie in $\rho$ in the level below, 
and we add each such subcell as a child to $\rho$
in $\mathcal{T}_{\rho}$. We recurse on all children of $\rho$
that still properly contain a cell from $\mathcal{C}$,
until $\mathcal{T}_{\rho}$ is complete.
The resulting set of trees $\mathcal{T}_\rho$, for all
$\rho \in \mathcal{C}'$, 
constitutes our quadforest
$\mathcal{F}$. 
Since $\mathcal{C}$ contains the full neighborhood for every vertex in $H$, the 
quadforest $\mathcal{F}$ has $O(\Psi^2 n)$ nodes, by
\cref{lem:dg-graph-properties}. The number of quadtrees in $\mathcal{F}$ 
is $O(n)$, since the roots for the neighboring cells 
of a vertex $\sigma \in H$ must be in the 
$(13 \times 13)$-neighborhood of the root for $\sigma$.
We store the $O(n)$ quadtree root cells of $\mathcal{F}$ in a 
red-black tree~\cite{cormen_introduction_2009} 
that allows us to locate a root 
in $O(\log n)$ time when given the coordinates of its
lower-left corner. 
We have an \MBM-structure of \cref{lem:mbm-general} for every pair of neighboring cells
$\sigma$, $\tau$ such that $S(\sigma) \neq \emptyset$ or
$S(\tau) \neq \emptyset$, containing the sites from
$S(\sigma)$ and $S(\tau)$. 
The proxy graph $H$ is represented by an 
HLT-structure $\cH$
(cf.~\cref{thm:dynamicspanningtree}). In $\cH$, we  store a vertex for every 
cell $\sigma$ with $S(\sigma) \neq \emptyset$ and
 edge $\sigma\tau$ for every pair of neighboring cells
$\sigma$ and $\tau$ whose \MBM is nonempty.

For every cell $\sigma$ in $\mathcal{F}$,
we store the set
$S(\sigma)$ of assigned sites (possibly empty).
Additionally, we store for every cell $\sigma$
a red-black tree containing all \MBM-structures that involve $\sigma$.
These red-black trees use as index the center of the other
cell $\tau$ that defines the \MBM-structure together with $\sigma$, allowing retrievals and updates in $O(\log n)$ time given $\tau$.
The center of a cell is unique, even across different cell sizes.
The total space for $\mathcal{F}$, the red-black tree of the  roots, and
$\cH$ is $O(\Psi^2 n)$, so the space is dominated by the total size 
of the
\MBM{}-structures.
 By \cref{lem:mbm-general}, this is 
bounded by $O(\Psi^2 n \log n)$,
since every site appears in at most $O(\Psi^2)$ \MBM{}s.
See \cref{fig:approach-psi-squared} for an overview of the structure.

\begin{figure}
    \centering
    \includegraphics{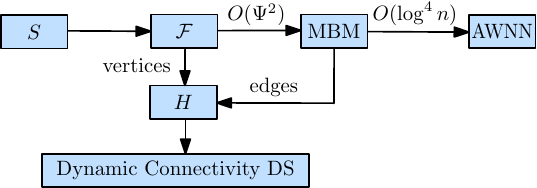}
    \caption{The structure of our initial data structure for general disks.}
    \label{fig:approach-psi-squared}
\end{figure} 

To perform a connectivity query for
two sites $s$ and $t$,
we first obtain the cells $\sigma$ and $\tau$ with 
$s \in S(\sigma)$ and $t \in S(\tau)$.
This takes $O(1)$ time, because this information can
be stored with a site when it is inserted into the structure.
Then, we
query the HLT-structure $\mathcal{H}$ with 
$\sigma$ and $\tau$, in $O(\log n / \log\log n)$ time. 
By \cref{lem:dg-graph-properties}, this yields the correct result.

To insert a site $s$ into $S$, 
we determine the cell $\sigma \in \G$ to which 
$s$ is assigned, as well as the set of neighboring
cells $\mathcal{N}(\sigma)$ of $\sigma$.
We locate all cells of $\{ \sigma \} \cup \mathcal{N}(\sigma)$ in 
$\mathcal{F}$, creating new quadtree nodes if necessary. This takes
$O(\Psi^2 + \log n)$ time,
by \cref{lem:dg-graph-properties} and the time required to obtain the quadtree roots.
We add $s$ to $S(\sigma)$, and we insert
$s$ into the \MBM{} structures for $\sigma$ and every neighboring cell 
$\tau \in \mathcal{N}(\sigma)$, creating new \MBM-structures if necessary. 
For every \MBM that becomes non-empty
we insert an edge into the HLT-structure $\mathcal{H}$.
Since $|\mathcal{N}(\sigma)| = O(\Psi^2)$,
by Lemmas~\ref{lem:mbm-general} and \ref{lem:dg-graph-properties},
updating the \MBM{}s needs
amortized expected time $O(\Psi^2 \log^{4} n)$ (this
also includes the time for inserting the new \MBM{}s into their respective
red-black trees),
and the update in $\mathcal{H}$ needs $O(\Psi^2 \log^2 n)$ amortized time by \cref{thm:dynamicspanningtree}.
Thus, the total amortized expected time for an insertion is
$O(\Psi^2 \log^{4} n)$.

Similarly, to delete a site $s \in S$, we locate the cell
$\sigma$ with $s \in S(\sigma)$ together
with the set of neighboring cells $\mathcal{N}(\sigma)$ of $\sigma$
in $\mathcal{F}$.
This takes  $O(\Psi^2 + \log n)$ 
time, by \cref{lem:dg-graph-properties} and the time required to obtain the quadtree roots. Then, we 
delete $s$ from $S(\sigma)$ and from all \MBM{}s for
$\sigma$ and a neighboring cell $\tau$.
If now $S(\sigma) = \emptyset$,
we also delete for all neighboring cells $\tau$ with $S(\tau) = \emptyset$ the \MBM{} for $\sigma$ and $\tau$.
Since $\sigma$ has $O(\Psi^2)$ neighboring cells,
all this needs
amortized expected time $O(\Psi^2 \log^{4} n)$
(including the time for deleting the \MBM{}s from their
respective red-black trees).
Afterwards, for all $\tau \in \mathcal{N}(\sigma)$ 
whose \MBM with $\sigma$ has become empty or was deleted through the deletion
of $s$, we delete the edge $\sigma\tau$ from the  HLT-structure $\mathcal{H}$.
Since there are 
$O(\Psi^2)$ such edges,
this takes $O(\Psi^2 \log^2 n)$ amortized time.
Finally, we
delete from 
$\mathcal{F}$ all cells $\tau \in \{ \sigma  \} \cup \mathcal{N}(\sigma)$ that
have $S(\tau) = \emptyset$ and that no longer occur in the
neighborhood of any vertex of $H$ with an assigned site.
The last condition corresponds to all cells of $\{ \sigma \} \cup \mathcal{N}(\sigma)$ with an empty \MBM{} red-black tree.
Hence, they can be obtained in $O(\Psi^2 + \log n)$ time.
Overall, the deletion time is dominated by the updates in
the \MBM{}s, and it is bounded by
$O(\Psi^2 \log^{4} n)$.
We obtain the following theorem:

\begin{theorem}
    \label{thm:dg-dynamicDG}
    There is a fully dynamic connectivity structure for
    disk graphs of bounded radius ratio $\Psi$
    such that
    an update takes amortized expected time 
    $O(\Psi^2 \log^{4} n)$
    and  a connectivity query takes worst-case time
    $O(\log n/\log \log n)$,
    where $n$ is the maximum number of sites at any time.
    The data structure requires $O(\Psi^2 n \log n)$ space.
\end{theorem}

\subsection{Improving the dependence on \texorpdfstring{$\Psi$}{the
radius ratio}}
\label{sub:dg-limit-insertions-to-mbm}

To improve over \cref{thm:dg-dynamicDG}, we show how to
reduce the degree of the proxy graph $H$
from $O(\Psi^2)$ to $O(\Psi)$.
The intuition is that it suffices
to focus on disks that are not contained in any other disk of $S$ to maintain
the connected components of $\cD(S)$. 
We call these disks \emph{(inclusion) maximal}.
Then, 
we need to consider only edges between disks that intersect
\emph{properly}, i.e.,  their boundaries intersect. 
If we want to perform a connectivity query
between two sites $s$ and $t$, we must find appropriate
maximal disks in $S$ that contain $D_s$ and $D_t$.
See \cref{fig:dg-disks-path-ignore-fully-contained} for examples depicting this intuition.

The following definition formalizes the notion of containment between 
disks.
The definition is based on grids, since we must take into account not only 
the disks that are currently present in the structure, but also
the disks that could be inserted in the future.

\begin{definition}
    \label{def:fully-contained}
    Let $D$ be a disk and $\sigma \in \G$ a grid cell with
    level at most $\lfloor \log \Psi \rfloor$.
    We say that $\sigma$ is \emph{Minkowski covered} by $D$ if and only if 
    every possible disk that can be assigned to 
    $\sigma$ fully lies in $D$.
    We call $\sigma$ \emph{maximal} 
    if and only if there is no larger cell $\tau \supset \sigma$ in
    $\G$ that
    is Minkowski covered by $D$.
\end{definition}

Equivalently, \cref{def:fully-contained} states  that
$\sigma$ is Minkowski covered by $D$ if and only if  
(i) $\sigma$ has level at most $\lfloor \log \Psi \rfloor$ and
(ii) the Minkowski sum of $\sigma$ with an
open disk of radius $2|\sigma|$ is contained in $D$. 
By definition, 
if $\sigma$ is Minkowski covered
by $D$, then all smaller cells $\tau \in \G$ 
with $\tau \subset \sigma$ 
are also Minkowski covered by $D$.
The following lemma bounds the number of cells of different types,
for a given disk.
See \cref{fig:fully-contained-in-grids-quadtree} for 
an illustration.

\begin{figure}
    \centering
    \colorlet{colorfullytopmost}{mypurple}
    \colorlet{colorfully}{mygreen}
    \colorlet{colorboundary}{mylightblue}
    \includegraphics{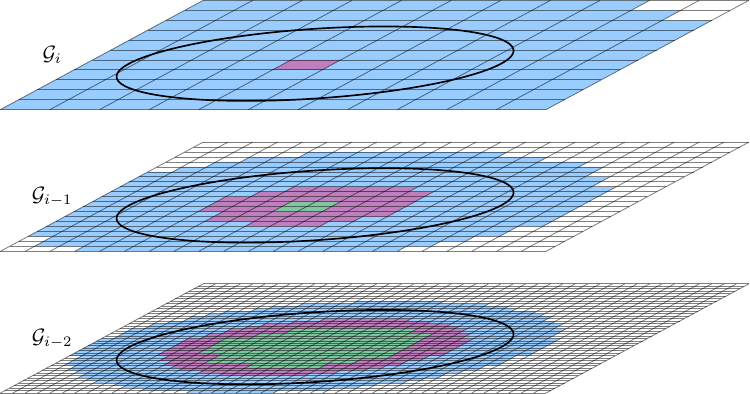}
    \caption{
        The types of cells that 
	require checking when updating the black disk in the 
	data structure of \cref{thm:dg-dynamicDG}:
        {\sffamily\bfseries\color{colorboundary}\pgfsetfillopacity{0.5}\rule{0.73em}{0.73em}\pgfsetfillopacity{1}}~$\mathcal{N}_1$:~not Minkowski covered \hfill
        {\sffamily\bfseries\color{colorfullytopmost}\pgfsetfillopacity{0.5}\rule{0.73em}{0.73em}\pgfsetfillopacity{1}}~$\mathcal{N}_2$:~maximal~Minkowski~covered \hfill\\
        {\sffamily\bfseries\color{colorfully}\pgfsetfillopacity{0.5}\rule{0.73em}{0.73em}\pgfsetfillopacity{1}}~$\mathcal{N}_3$:~Minkowski~covered, not maximal
    }
    \label{fig:fully-contained-in-grids-quadtree}
\end{figure}

\begin{lemma}
    \label[lemma]{lem:dg-linear-cells-intersecting-disks-edge}
    Let $s \in S$ be a site, and let $\mathcal{N}(s)$ be the set of cells of
    $\G$ that may have an assigned disk that intersects $D_s$.
    Write $\mathcal{N}(s)$ as the disjoint union 
    $\mathcal{N}(s) = \mathcal{N}_1(s) \mathrel{\dot\cup} \mathcal{N}_2(s) \mathrel{\dot\cup} \mathcal{N}_3(s)$,
    such that (i) $\mathcal{N}_1(s)$ consists of the cells that are not 
    Minkowski covered by 
    $D_s$, (ii) $\mathcal{N}_2(s)$ consists of  the cells that are maximal Minkowski covered
    by $D_s$, and (iii) $\mathcal{N}_3(s)$ consists of  the cells that are Minkowski covered
    by $D_s$, but not maximal with this property.
    Then, we have $|\mathcal{N}_1(s) \cup \mathcal{N}_2(s)| = O(\Psi)$
    and $|\mathcal{N}_3(s)| = O(\Psi^2)$.
\end{lemma}

\begin{proof}
    First, we show that $|\mathcal{N}_1(s)| = O(\Psi)$.
    Let $\ell \in \{0, \dots, \lfloor \log \Psi \rfloor\}$ be the level of $s$,
    and let $\sigma \in \G_\ell$ be the cell with $s \in S(\sigma)$.
    By definition, all cells of $\mathcal{N}_1(s)$ have level at most
    $\lfloor \log \Psi \rfloor$.
    We bound 
    for each level $\ell' \in \{0, \dots,  \lfloor \log \Psi \rfloor\}$
    the size of $\mathcal{N}_1(s) \cap \G_{\ell'}$.
    For $\ell' \geq \ell$, there are $O(1)$ cells
    in $\G_{\ell'}$ that are neighboring to $\sigma$, so
    we have $|\mathcal{N}_1(s) \cap \G_{\ell'}| = O(1)$.
    For $\ell' < \ell$, we observe that any cell in $\mathcal{N}_1(s) \cap 
    \G_{\ell'}$
    must intersect the annulus $A_s$ that is centered at the boundary of  
    $D_s$ and has width 
    $\Theta(2^{\ell'})$,
    see \cref{fig:fully-contained-in-grids-quadtree}.
    The area of $A_s$ is
    $\Theta(r_s \cdot 2^{\ell'}) = \Theta(2^{\ell + \ell'})$.
    Since the cells at level $\ell'$ have pairwise disjoint
    interiors and area
    $\Theta(2^{2\ell'})$,
    a simple volume argument shows
    that $|\mathcal{N}_1(s) \cap \G_{\ell'}| = O(2^{\ell - \ell'})$.
    Adding over all $\ell'$, we get 
    \begin{align*}
    |\mathcal{N}_1(s)| = 
    \sum_{\ell' = 0}^{\lfloor \log \Psi \rfloor} 
    |\mathcal{N}_1(s) \cap G_{\ell'}|
    &=
    \sum_{\ell' = 0}^{\ell - 1} |\mathcal{N}_1(s) \cap G_{\ell'}|
    +
  \sum_{\ell' = \ell}^{\lfloor \log \Psi \rfloor} 
    |\mathcal{N}_1(s) \cap G_{\ell'}|\\
    &=
    \sum_{\ell' = 0}^{\ell - 1}  O\Big(2^{\ell - \ell'}\Big)
    +
  \sum_{\ell' = \ell}^{\lfloor \log \Psi \rfloor} 
     O(1)\\
    &= O(2^\ell) + O(\log \Psi)
    = O(\Psi).
    \end{align*}
    Next, we bound $|\mathcal{N}_1(s) \cup \mathcal{N}_2(s)|$. 
    For this, we note that every cell in $\mathcal{N}_2(s)$ either
    (i) has level $\lfloor \log \Psi \rfloor$ (and there 
    are $O(1)$ such cells that are Minkowski covered
    by $D_s$), or (ii)
    is one of the four cells that are directly contained
    in a cell $\sigma \in \mathcal{N}_1(s)$ and have diameter
    $|\sigma|/2$.
    Hence, we get
    $|\mathcal{N}_1(s) \cup \mathcal{N}_2(s)| = 
    O(|\mathcal{N}_1(s)|) = O(\Psi)$.
    The bound on $|\mathcal{N}_3(s)|$ follows from 
    \cref{lem:dg-graph-properties}.
\end{proof}

Now, 
we show how to reduce the degree
of the proxy graph. 
Let $H$ be the proxy graph from
\cref{sub:dg-adapting-the-unit-disk-case}, 
and
let $H'$ be a subgraph of $H$ that is defined as follows:
as in $H$, the vertices of $H'$ are all the grid cells $\sigma \in \G$
with $S(\sigma) \neq \emptyset$.
Let $\sigma, \tau$ be two distinct vertices of $H'$.
Then, there is an edge between $\sigma$ and $\tau$
in $H'$ if and only if there are
sites $s \in S(\sigma)$ and $t \in S(\tau)$ such that
(i) $D_s$ and $D_t$ intersect; 
(ii) $D_s$ does \emph{not}
Minkowski cover $\tau$; and (iii) $D_t$ does \emph{not} Minkowski cover $\sigma$.
From the definition, it is immediate that $H'$ is a subgraph of $H$,
and that there is an edge between $\sigma$ and $\tau$
if and only there are sites $s \in S(\sigma)$ and $t \in S(\tau)$
such that (i) $D_s$ and $D_t$ intersect; (ii) $\tau \in \mathcal{N}_1(s)$;
and (iii) $\sigma \in \mathcal{N}_1(t)$.

Let $\sigma$ be a vertex in $H'$. In order to implement
connectivity queries in the sparsified graph $H'$, 
we define the \emph{proxy vertex}
$\sigma'$ of $\sigma$ in $H'$ as follows: if there is no site $s \in S$
such that
$D_s$ Minkowski covers $\sigma$, then $\sigma' = \sigma$. Otherwise,
let $\overline{\sigma}$ be the maximal cell such that
(i) $\overline{\sigma}$  contains
$\sigma$ and (ii) there is a site $s \in S$ such 
that $\overline{\sigma}$ is Minkowski covered by $D_s$.
Let $t \in S$ be a site of maximum radius such that 
$D_t$ Minkowski covers 
$\overline{\sigma}$,
and let $\rho$ be the cell that $t$ is assigned to, i.e.,
$\rho$ is the cell with
$t \in S(\rho)$.
Then, the proxy vertex $\sigma'$ for $\sigma$ in $H'$
is defined as  $\rho$.
If there are multiple choices for $t$, we fix an arbitrary one.
See \cref{fig:proxy-cell} for an illustration.
The next lemma shows that $H'$ preserves the connectivity between
proxy vertices.

\begin{figure}
    \centering
    \includegraphics{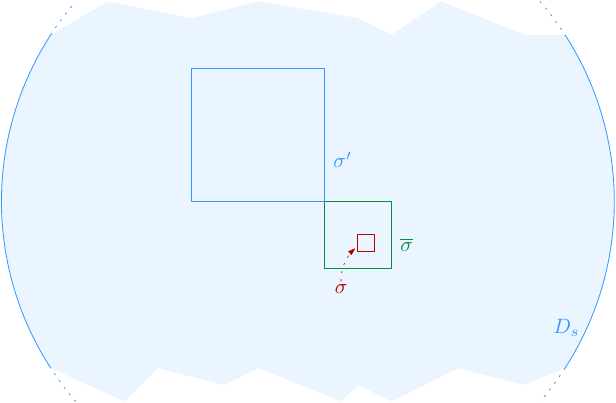}
    \caption{Obtaining the proxy vertex for $\sigma$. We assume $\overline{\sigma}$ is the largest cell which (i) contains $\sigma$ and (ii) is Minkowski covered by a disk. If $D_s$ is the largest disk doing so, and we have $s \in S(\sigma')$,   then $\sigma'$ is the proxy vertex for $\sigma$. This means in particular that for each $v \in S(\sigma)$ the disk $D_s$ properly contains $D_v$.}
    \label{fig:proxy-cell}
\end{figure}

\begin{lemma}
    \label[lemma]{lem:dg-gnore-fully-connected}
    Let $s, t \in S$ be two sites, and let $\sigma, \tau$ be
    the cells with $s \in S(\sigma)$ and $t \in S(\tau)$.
    Let $\sigma'$ and $\tau'$ be the proxy vertices for 
    $\sigma$ and $\tau$ in $H'$.
    Then, $\sigma'$ and $\tau'$ are  connected in $H'$ 
    if and only if $s$ and $t$ are connected in $\cD(S)$.
\end{lemma}

\begin{proof}
    First, suppose that $\sigma'$ and $\tau'$ are connected
    in $H'$. 
    The cells
    $\sigma'$ and $\tau'$ are 
    connected in $H$,
    since $H'$ is a subgraph of $H$.
    From the definition of proxy vertices
    it is immediate that $\sigma'$ is adjacent in $H$ to $\sigma$ and 
    that $\tau'$ is adjacent in $H$ to $\tau$.
    Thus, 
    $\sigma$ and $\tau$ are 
    connected in $H$,
    and \cref{lem:dg-graph-properties} shows that
    $s$ and $t$ are connected in $\cD(S)$.

    Next, suppose that 
    $s$ and $t$ are connected in $\cD(S)$.
    We call a disk $D_u$, $u \in S$, 
    \emph{(inclusion) maximal} if there is no site $v \in S$ with
    $u \neq v$ and $D_u \subset D_v$.
    Our strategy is to consider a path $\pi$ of 
    \emph{(inclusion) maximal} disks 
    that connects $s$ and $t$ in  $\cD(S)$, and to show
    that $\pi$ induces a path $\pi'$ between $\sigma'$ and $\tau'$ in 
    $H'$. 
    To construct $\pi$,
    we first find a maximal disk $D_{s'}$ 
    that \emph{represents} $D_s$, as follows:
    if the proxy vertex $\sigma'$ is different from $\sigma$, we let $D_{s'}$
    be the disk used to define $\sigma'$,  i.e.,
    the disk of maximum radius that Minkowski
    covers $\overline{\sigma}$, the maximal Minkowski covered
    cell that contains $\sigma$.  Then, the disk $D_{s'}$ is indeed maximal,
    since any disk that properly contains $D_{s'}$ would have larger
    radius and would Minkowski cover $\overline{\sigma}$, contradicting the choice of 
    $D_{s'}$. 
    If the proxy vertex $\sigma'$ is $\sigma$ itself, we let 
    $D_{s'}$, $s' \in S$, be an arbitrary
    maximal disk that contains $D_s$. Possibly, this may be $D_s$ itself, but it also may be that $s' \not \in S(\sigma)$.
    In the same way, we obtain a maximal disk $D_{t'}$ that represents $D_t$.

    Now, we claim that there is a path $\pi$ in $\cD(S)$ 
    between $s'$ and $t'$ 
    that uses only maximal disks: 
    by our choice of $s'$ and $t'$, we know that $s'$ is in the
    same connected component as $s$, and $t'$ is in the same connected
    component as $t$. Since we assumed that $s$ and $t$ are connected
    in $\cD(S)$, this also holds for $s'$ and $t'$. Consider a 
    path in $\cD(S)$ between $s'$ and $t'$, and replace 
    every disk along this path by a maximal disk that contains it. 
    The resulting path
    $\pi$ 
    has the required property
    (possibly after making shortcuts between duplicate disks). 
    See \cref{fig:dg-disks-path-ignore-fully-contained} for an
    illustration.

    Consider the sequence $\pi'$ of cells in $H'$ that
    we obtain by replacing every site $u$ in $\pi$ by the vertex
    $\sigma_u$ of $H'$ with $u \in S(\sigma_u)$, and by 
    taking shortcuts between duplicate cells. Then,
    $\pi'$ is a path in $H'$, since the definition
    of $H'$ implies that the assigned cells for
    two intersecting maximal disks of $S$ must be either identical or  
    adjacent in $H'$.
    Furthermore, the first cell of $\pi'$ is either the
    proxy vertex $\sigma'$ or adjacent in $H'$ to $\sigma'$.
    The latter may happen if $\sigma' = \sigma$ and 
    $S$ has  a disk  that contains $D_s$,
    but no disk that Minkowski covers $\sigma$, see \cref{fig:maximal-path-replacement-adjacent}.
    Similarly, the last cell of $\pi'$ is either $\tau'$ or a cell
    that is adjacent to $\tau'$ in $H'$.
    Thus, by possibly adding $\sigma'$ and/or $\tau'$ to $\pi'$, we
    obtain a path $\pi'$ in $H'$ that connects $\sigma'$ and $\tau'$.
\end{proof}

\begin{figure}
    \centering
    \begin{subfigure}[t]{0.45\textwidth}
        \begin{center}
            \includegraphics{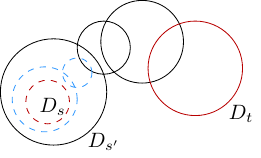}
        \end{center}
        \caption{When $D_s$ is connected to $D_t$, then the inclusion maximal disk $D_{s'} \supseteq D_s$, which represents $D_s$, is also connected to $D_t$.}
            \label{fig:dg-disks-path-ignore-fully-contained:b}
    \end{subfigure}
    \hfill
    \begin{subfigure}[t]{0.45\textwidth}
        \begin{center}
            \includegraphics{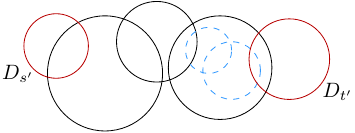}
        \end{center}
        \caption{A path between the two red disks 
	may consist of the black disks as intermediates.
        Note that any non-maximal disk is not required to form a path and can be safely ignored.}
        \label{fig:dg-disks-path-ignore-fully-contained:a}
    \end{subfigure}
    \caption{The two main ingredients for the construction of path $\pi$ in \cref{lem:dg-gnore-fully-connected}.}
    \label{fig:dg-disks-path-ignore-fully-contained}
\end{figure}

\begin{figure}
    \centering
    \includegraphics{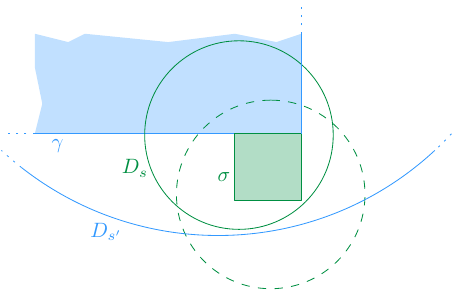}
    \caption{In the proof of \cref{lem:dg-gnore-fully-connected} the first cell of $\pi'$ might not be the proxy vertex of $\sigma$. This occurs if $\sigma$ is its own proxy cell, there exist a maximal disk $D_{s'}$ containing $D_s$, $s' \not \in S(\sigma)$, and $D_{s'}$ does not Minkowski cover $\sigma$. In this example, we have $s \in S(\sigma)$, $s' \in S(\gamma)$, and the dashed disk would be assigned to $\sigma$, hence $D_{s'}$ does not Minkowski cover $\sigma$.}
    \label{fig:maximal-path-replacement-adjacent}
\end{figure}

\cref{lem:dg-gnore-fully-connected} implies that it is enough to 
maintain the sparsified proxy graph $H'$ in an HLT-structure
$\mathcal{H}$. To perform updates and
queries efficiently, we need to do two things:
first, 
we need to determine
which edges in $H'$ are affected by an update in $S$;
second, 
in order to perform
a connectivity query between two sites $s$ and $t$, we  not
only need to determine the assigned cells $\sigma$ and $\tau$ for
$s$ and $t$, but we must also find the proxy vertices 
$\sigma'$ for $\sigma$ and $\tau'$ for $\tau$ in $H'$.
We 
again maintain a  \emph{quadforest} $\mathcal{F}$:
for each $s \in S$, let $\sigma_s$ be the cell
with $s \in S(\sigma_s)$,
and let $\mathcal{C} = \bigcup_{s \in S} \{\sigma_s\} \cup \mathcal{N}_1(s)
\cup \mathcal{N}_2(s)$, with the
notation from 
\cref{lem:dg-linear-cells-intersecting-disks-edge}.\footnote{Note that
unlike in~\cref{sub:dg-adapting-the-unit-disk-case}, the neighborhoods
depend on the site $s$ and not just the cell $\sigma$ that $s$
is assigned to.}

As in~\cref{sub:dg-adapting-the-unit-disk-case},
we use $\mathcal{C}$ to define $\mathcal{F}$.
Let $\mathcal{C}'$ be the set of all cells in 
$\G_{\lfloor \log \Psi \rfloor}$ that contain
at least one cell from $\mathcal{C}$. 
For each
$\rho \in \mathcal{C}'$, let  $\mathcal{T}_\rho$ be the quadtree 
for the cells in $\mathcal{C}$ that lie in $\rho$ 
(cf.~\cref{sub:dg-adapting-the-unit-disk-case}).
Let $\mathcal{F}$
be the quadforest that contains all quadtrees $\mathcal{T}_\rho$,
$\rho \in \mathcal{C}'$.
By \cref{lem:dg-linear-cells-intersecting-disks-edge}, the quadforest
$\mathcal{F}$ contains $O(\Psi n)$ nodes, and as in 
\cref{sub:dg-adapting-the-unit-disk-case}, we see
that is has $O(n)$ roots.
We store the  
root cells of $\mathcal{F}$ in a red-black tree~\cite{cormen_introduction_2009}, indexed by the lower-left corner 
of the cells. We have an \MBM-structure for every pair of cells
$\sigma$, $\tau$ with $S(\sigma) \neq \emptyset$ and 
$\tau \in \bigcup_{ s \in S(\sigma) } \mathcal{N}_1(s)$ or with
$S(\tau) \neq \emptyset$ and 
$\sigma \in \bigcup_{t \in S(\tau)} \mathcal{N}_1(t)$.
The \MBM-structure for $\sigma$ and $\tau$ contains all
sites $s \in S(\sigma)$ with $\tau \in \mathcal{N}_1(s)$ and
all sites $t \in S(\tau)$ with $\sigma \in \mathcal{N}_1(t)$.
The proxy graph $H$ is represented by an 
HLT-structure $\cH$
(cf.~\cref{thm:dynamicspanningtree}). In $\cH$, we  store a vertex for every 
cell $\sigma$ with $S(\sigma) \neq \emptyset$
and edge $\sigma\tau$ for every pair of cells
$\sigma$ and $\tau$ whose \MBM is nonempty.

For 
every cell $\sigma$ in $\mathcal{F}$,
we store several data structures.
First, 
we maintain the
set $S(\sigma)$ of assigned sites for $\sigma$ (possibly empty).
Second, we store the set 
$\mathcal{C}(\sigma)$ of all sites $s \in S$ 
such that $\sigma$ is maximal Minkowski covered by $D_s$,
i.e., such that $\sigma \in \mathcal{N}_2(s)$.
The set $\mathcal{C}(\sigma)$ is represented as a red-black tree, 
where the ordering is determined by the associated radius $r_s$.
In addition, a pointer to the maximum element is maintained, allowing its retrieval in $O(1)$ time.
Third, we have the  \MBM{}s that involve $\sigma$.
The \MBM{}s are organized in a red-black tree, indexed by the lower-left corner of a neighboring
vertex $\tau$. The total space for $\mathcal{F}$, the red-black tree 
of roots, the $\mathcal{C}(\sigma)$ structures and $\mathcal{H}$ is 
$O(\Psi n)$, so the space is dominated by the total size
of the \MBM-structures. This is bounded by $O(\Psi n \log n)$, 
since every site appears in at most $O(\Psi)$ \MBM{}s.

Let $q = \lceil \log \log \Psi / 2 \rceil$.
In the quadforest $\mathcal{F}$, we designate the cells with level 
$i \cdot q$, for $0 \leq i \leq \lfloor \log \Psi \rfloor/q$
as \emph{special}.
In $\mathcal{F}$, we maintain a pointer from 
every 
special cell $\rho$ to the lowest special ancestor of $\rho$,
lying $q$ levels above. Additionally, the special cell $\rho$ 
stores a red-black tree $\overline{\mathcal{C}}(\rho)$ 
that contains the elements of the red-black trees
$\mathcal{C}(\rho')$ for all ancestors $\rho'$ of $\rho$ that
are between $\rho$ and the next special level above it.
Since the elements from a red-black tree $\mathcal{C}(\rho')$ appear
in at most $4^q = O(\log  \Psi)$ special red-black trees, the total
additional space is $O( \Psi n \log \Psi)$.
It follows that the whole structure needs 
$O(\Psi n (\log n + \log \Psi))$ space.
See \cref{fig:approach-psi} for an overview of the data structure.

\begin{figure}
    \centering
    \includegraphics{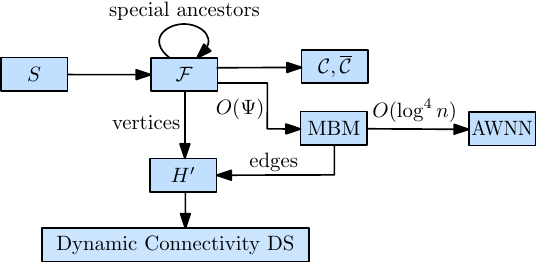}
    \caption{The structure of our data structure for general disks.}
    \label{fig:approach-psi}
\end{figure} 

Using $\mathcal{F}$ we can perform a connectivity query for two sites
$s$ and $t$. Let $\sigma$ and
$\tau$ be the cells with $s \in S(\sigma)$ and $t \in S(\tau)$.
As in~\cref{sub:dg-adapting-the-unit-disk-case}, we can store these 
cells with the sites
$s$ and $t$ and retrieve them in $O(1)$ time.
We describe how to find the proxy vertex $\sigma'$
for $\sigma$, the
procedure for $\tau$ is analogous: 
first, we obtain the maximal cell
$\overline{\sigma} \supseteq \sigma$ that contains $\sigma$
and is Minkowski covered by a disk from $S$ (if it exists).
For this, we ascend from $\sigma$ in 
its respective quadtree for at most $q$ steps,
until we reach the closest special level above $\sigma$. From
there, we follow the pointers that connect the special
levels, until the root.
For each cell $\rho$
along the way, we check if the red-black tree $\mathcal{C}(\rho)$
(for a non-special cell) or the red-black tree 
$\overline{\mathcal{C}}(\rho)$ (for a special cell)
is non-empty, in $O(1)$ time. 
Let $\rho'$ be the largest cell with this property.
If $\rho'$ is non-special, we set $\overline{\sigma} = \rho'$.
If $\rho$ is special, we ascend from 
$\rho'$ towards the root until the next special level, to find the largest
cell $\overline{\sigma}$ along the way that has 
$\mathcal{C}(\overline{\sigma})$ non-empty.
At the end, we have found the largest ancestor 
$\overline{\sigma}$ of $\sigma$ that is Minkowski covered by a 
disk $D_s$, $s \in S$,
if it exists. Now, if $\overline{\sigma}$ does not
exist, we set $\sigma' = \sigma$. Otherwise, let $u$ be the site of
maximum radius that is stored in the red-black tree 
$\mathcal{C}(\overline{\sigma})$ and set $\sigma' = \sigma_u$,
where $\sigma_u$ is the cell with $u \in S(\sigma_u)$.
Once $\overline{\sigma}$ is known, it takes $O(1)$ time to find
$\sigma'$. Thus, the total time to find $\sigma'$ (and $\tau'$) is 
$O(\log \Psi / \log \log \Psi)$.
Once $\sigma'$ and $\tau'$ are determined, we use them
to query the connectivity structure $\mathcal{H}$,
in $O(\log n/\log \log n)$ time, and we
return the result. By 
\cref{lem:dg-gnore-fully-connected}, this gives the
 correct answer. 
The overall running time  is 
$O(\log n/\log \log n + \log \Psi/\log\log \Psi)$.

To insert a site $s$ into $S$, we determine the cell $\sigma \in \G$
to which $s$ is assigned, as well as the sets $\mathcal{N}_1(s)$ and
$\mathcal{N}_2(s)$ of neighboring cells.
We locate all cells of 
$\{\sigma \} \cup \mathcal{N}_1(s) \cup \mathcal{N}_2(s)$ in $\mathcal{F}$,
creating new quadtree nodes if necessary. This takes
$O(\Psi + \log n)$ time, by 
\cref{lem:dg-linear-cells-intersecting-disks-edge} and the time to obtain the quadtree roots.
We add $s$ to $S(\sigma)$, and we insert $s$ into the \MBM-structures
for $\sigma$ and every neighboring cell $\tau \in \mathcal{N}_1(s)$,
creating new \MBM-structures if necessary. Finally, 
for $\tau \in \mathcal{N}_2(s)$,
we insert $s$ into
the red-black tree $\mathcal{C}(\tau)$
and into the special red-black trees $\overline{\mathcal{C}}(\tau')$ for all
special descendants of $\tau$ that are at most $q$ levels below it, creating cells as necessary.
For every \MBM{} that becomes non-empty, we insert a corresponding 
edge into the
HLT-structure $\cH$. Since 
$|\mathcal{N}_1(s) \cup \mathcal{N}_2(s)| = O(\Psi)$, by 
\cref{lem:dg-linear-cells-intersecting-disks-edge}, updating the \MBM{}s takes
$O(\Psi \log^4 n)$ amortized expected time,
while updating the red-black trees $\mathcal{C}(\tau)$ and the edges in $\mathcal{H}$
takes $O(\Psi \log^2 n)$  amortized time.
The special red-black trees
can be updated in  $O(\Psi \log \Psi \log n)$ time, since $s$ needs to be inserted into
 $O(\Psi \log \Psi)$ of them.
The overall running time  is 
$O(\Psi (\log^4 n + \log n\log \Psi))$.

To delete a site $s$ from $S$, we locate
in $\mathcal{F}$ the cell $\sigma$ with $s \in S(\sigma)$
as well as the sets $\mathcal{N}_1(s)$ and
$\mathcal{N}_2(s)$ of neighboring cells.
This takes
$O(\Psi + \log n)$ time, by 
\cref{lem:dg-linear-cells-intersecting-disks-edge}.
We remove $s$ from $S(\sigma)$ and from  the \MBM-structures
for $\sigma$ and every cell $\tau \in \mathcal{N}_1(s)$.
For 
$\tau \in \mathcal{N}_2(s)$,
we delete $s$ from
the red-black tree $\mathcal{C}(\tau)$
as well as the special red-black trees $\overline{\mathcal{C}}(\tau')$ for all
special descendants of $\tau$ that are at most $q$ levels below it.
For every \MBM{} that becomes empty, we delete the corresponding 
edge from the
HLT-structure $\cH$. Finally,
we delete from $\mathcal{F}$ all cells $\tau$ that
have $S(\tau) = \emptyset$, for which no site $t$ with $\tau \in \mathcal{N}_1(t) \cup \mathcal{N}_2(t)$ exists, and which also don't have a special cell $\tau'$ with non-empty red-black tree $\overline{\mathcal{C}}(\tau')$ as descendant.
Since 
$|\mathcal{N}_1(s) \cup \mathcal{N}_2(s)| = O(\Psi)$, by 
\cref{lem:dg-linear-cells-intersecting-disks-edge}, updating the \MBM{}s takes
$O(\Psi \log^4 n)$ amortized expected time,
while updating the red-black trees $\mathcal{C}(\tau)$ and the edges in $\mathcal{H}$
takes $O(\Psi \log n)$  amortized time. The special red-black trees
can be updated in  $O(\Psi \log \Psi \log n)$ time, since $O(\Psi \log \Psi)$ of them contain $s$.
The overall running time  is 
$O(\Psi (\log^4 n + \log n\log \Psi))$.

\begin{theorem}
    \label{thm:bounded-radius-ratio-psi}
    There is a data structure for 
    dynamic disk connectivity with amortized expected update time
    $O(\Psi \log^{4} n)$
    and query time $O(\log n/\log \log n)$.
    It needs $O(\Psi n \log n)$ 
    space.
\end{theorem}
\begin{proof}
  This is immediate by the discussion above and the fact that
  we can assume $\Psi = O(n^3)$, since otherwise the 
  theorem follows from the trivial algorithm that
  completely recomputes all connected components after
  every update.
\end{proof}

\section{Disk revealing data structure}%
\label{sec:reveal}

In \cref{sec:bounded}, we will  construct data structures 
for the semi-dynamic setting.
Before that, we address the following problem:
we are given two dynamic sets $R$ and $B$ of
disks, such that the disks in $R$ 
can be both inserted and deleted, while the
disks in $B$
can only be deleted.
We would like to maintain $R$ and $B$ in
a data structure such that whenever we delete
a disk $b$ from $B$,
we receive a list of all the disks in the current set $R$ 
that intersect the disk $b$ but no other disk
from the remaining set $B \setminus \{b\}$.
We call such disks
\emph{revealed} by the deletion of $b$;
see~\Cref{fig:example-disk-reveal}.
Using this \emph{disk revealing}
data structure, we will be able 
to obtain efficiently the affected edges 
during an update in the decremental setting, as described 
in \cref{sec:bounded}.
\begin{figure}
    \centering
    \includegraphics{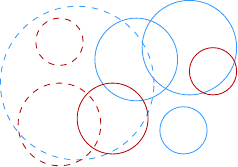}
    \caption{When removing a blue disk $b$, we want to obtain all 
    red disks that intersect $b$ but no other blue disk. 
    For example, after removing the dashed blue disk, the dashed red 
    disks are revealed and need to be reported.
    }
    \label{fig:example-disk-reveal}
\end{figure}

We construct the disk revealing structure 
in \cref{sec:sampling-planes,sec:sampling-disks}, and state our
final result in \cref{thm:disk-reveal}.
The central idea is to represent the  intersections 
between each disk $r \in R$ and the disks
in $B$ sparsely by assigning $r$ to one disk $b \in B$
that intersects it. If $b$ gets deleted, we either
report $r$ as revealed, or we determine a new 
disk of $B$ that intersects $r$.
To ensure that the assignments are not updated too often (at least
in expectation), we choose the assigned disk of $r$ at  random and assume an oblivious adversary.
Thus, given $r$, we must be able to randomly sample a disk 
among all disks in $B$ that intersect it.

Our main task in this section is 
to obtain a dynamic data structure for disks
that allows for \emph{random sample queries}:
given a query disk $D$, report a random disk among all disks in the
current set that intersect $D$.
We build upon the dynamic lower envelope data structure 
by Kaplan et al.~\cite{kaplan_dynamic_2020},
which we  briefly review 
in \cref{par:a-deeper-dive-kaplan}.
We proceed in two steps:
first, we describe a simpler data structure for 
sampling a random disk that contains a given query \emph{point}. This 
builds on the data structure of Kaplan et al.\@ for 
planes~\cite[Section~7]{kaplan_dynamic_2020}.
Afterwards, we extend the result to the data structure of 
continuous bivariate functions of constant description 
complexity~\cite[Section~8]{kaplan_dynamic_2020}.
This works essentially in the same way, but the details
are slightly more intricate. As a result, we get a
structure for sampling a random disk intersecting a given disk,
from which our disk revealing structure is readily derived.

\subsection{A deeper dive into Kaplan et al.}
\label{par:a-deeper-dive-kaplan}

We describe the details of the data structure
by Kaplan et al.~\cite{kaplan_dynamic_2020}.
Essentially, they explain how to
maintain the lower envelope of a set of $xy$-monotone surfaces
in $\R^3$ under insertions and deletions, while supporting 
vertical ray shooting queries~\cite[Section~8]{kaplan_dynamic_2020}.
Their structure is an extension of a structure by 
Chan~\cite{chan_dynamic_2010} that applies to the special case of
planes. 

The main ingredient are
\emph{vertical $k$-shallow $(1/r)$-cuttings}:
let $\mathcal{A}(H)$ be the arrangement of a set of 
planes $H$ in $\R^3$.
The \emph{$k$-level $L_k$} of $\mathcal{A}(H)$ 
is the closure of all points in $\bigcup_{h \in H} h$ with $k$ 
planes of $H$ strictly below them.
Then, $L_{\leq k}(H)$ (or $L_{\leq k}$ if the planes 
are clear from the context) is the union of the levels  $L_0, \dots, L_k$.
A \emph{vertical $k$-shallow $(1/r)$-cutting} for $H$ is a set $\Lambda$ 
of pairwise openly disjoint prisms such that (i) the union of the
prisms in
$\Lambda$ covers $L_{\leq k}(H)$; (ii) the interior of each 
$\tau \in \Lambda$ is intersected by at most $|H|/r$ planes of 
$H$; and (iii) every prism is \emph{vertical} (i.e., it consists of a 
three-dimensional triangle and all points vertically below it, where
some or all vertices the triangle may lie at infinity.). 
The \emph{conflict list} $\CL(\tau)$ of a prism $\tau \in \Lambda$ 
is the set of all planes in $H$ that cross the interior of $\tau$.
The \emph{size} of $\Lambda$ is the number of prisms
in $\Lambda$.
Given $H$, 
 a vertical 
$\Theta(|H|/r)$-shallow $1/r$-cutting for $H$ of size $O(r)$ 
can be found in time $O(|H| \log r)$,
as shown by Chan and Tsakalidis~\cite{chan_optimal_2016}.

The data structure of Kaplan et al.\@ 
consists of $O(\log n)$ static \emph{substructures} of 
exponentially decreasing sizes, where $n$ 
is the current number of planes in $H$.
A substructure $\Xi$ is constructed for some
initial set $H' \subseteq H$ of $n'$ planes, 
It consists of a 
hierarchy $\Lambda_{m +1}, \dots, \Lambda_0$ 
of $m + 2$ vertical shallow cuttings, with $m = O(\log n')$.
These shallow cuttings are obtained as follows:
let $\alpha > 0$ and $k_0 \in \mathbb{N}$ be appropriate constants, 
and let $k_j = 2^j k_0$, for $j > 0$.
The cutting $\Lambda_{m+1}$ consists of a single prism 
covering all of $\R^3$ which has $H'$ as its conflict list.
Let $H_{m+1} = H'$ 
and $n_{m + 1} = |H_{m + 1}|$.
The remaining cuttings are constructed iteratively:
set $H_m = H_{m + 1}$ and $n_m = |H_m|$.
For $j = m, \dots, 0$, the cutting $\Lambda_j$ is obtained as a 
vertical $k_j$-shallow $(\alpha k_j/n_{j})$-cutting for
$H_{j}$. 
The prisms of $\Lambda_j$ cover $L_{\leq k_j}(H_{j})$,
and each conflict list contains at most $\alpha k_j$ planes of $H_j$.
Additionally,
the set $Q$ of all planes  in $H_{j}$ 
that lie in ``too many'' conflict lists of the cuttings
constructed so far is determined.
All planes in $Q$ are removed from the
conflict lists in $\Lambda_j$ (but not from the higher cuttings, they are just marked there). 
For the next round $j - 1$, 
set $H_{j - 1} = H_j \setminus Q$ and $n_{j - 1} = |H_{j - 1}|$,
and the construction continues until $j = 0$. 
We say that the planes in $H_{-1}$
are \emph{stored} in $\Xi$, and the planes in $H' \setminus
H_{-1}$ are \emph{pruned} in $\Xi$. One can show that
$|H_{-1}| = \Theta(|H'|)$
(see the paper~\cite{kaplan_dynamic_2020} for details). 
The substructure $\Xi$
needs $O(n' \log n')$ space,
and Chan showed how to build it in $O(n' \log n')$ 
time~\cite{Chan20a}.

The whole data structure is now obtained as follows:
first, a substructure $\Xi$ with $H' = H$ is constructed.
Since not all planes from $H$ may be stored in $\Xi$,
let $H'' =  H' \setminus H_{-1}$ 
be the set of pruned planes in $\Xi$. 
Then, a substructure
$\Xi'$ is constructed for $H''$. This 
is repeated until there are no more pruned planes at the end of an iteration.
At the end, every plane from $H$ is stored
in exactly one substructure, and one can show that
$O(\log n)$ substructures are needed until the pruning stops.

Now, we describe how to perform a vertical ray shooting query for
$q = (x, y)$. Let $\ell_q$ be the vertical line  through $q$.
The shallow cutting $\Lambda_0$ of each of the $O(\log n)$ 
substructures is searched 
for the prism that intersects $\ell_q$. This takes $O(\log n)$ time
in each $\Lambda_0$, using a suitable point location structure (e.g.\@ Edelsbrunner et al.~\cite{edelsbrunner_optimal_1986}).
Then, the conflict list of each such prism is scanned 
for the lowest plane that intersects $\ell_q$,
in $O(k_0) = O(1)$ time per prism. To answer the query,
the lowest plane overall is returned. It follows that 
the total time for a vertical
ray shooting query is $O(\log^2 n)$.

To insert a new plane a standard 
Bentley-Saxe-method~\cite{bentley_decomposable_1980} is used,
rearranging smaller substructures into 
bigger ones to accommodate the new element.
Since the Bentley-Saxe-method adds another $\log$-factor
over the construction time of a substructure, 
an insertion needs $O(\log^2 n)$ amortized time.
Deletions are handled lazily,
by marking a plane as \emph{deleted}, while
leaving it in the structures that contain it.
During a query and the rebuilding during an insertion, the deleted planes are simply ignored.
Thus, without further action, it may happen that 
the lowest active (i.e.\@ not deleted, not pruned, and not purged, with purging being explained shortly) plane 
along a given vertical line no longer lies in the relevant prism of 
a lowest cutting $\Lambda_0$. See \cref{fig:lowest-cutting-empty} for an example.
\begin{figure}
    \centering
    \includegraphics{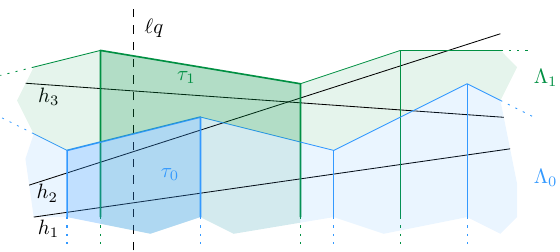}
    \caption{
    Illustration of a situation, where the lowest active plane of a substructure might not be in the relevant prism.
    The ray shooting query is performed with \(\ell_q\). If \(h_1\) and \(h_2\) are not active (e.g.\@ they were deleted), then \(h_3\) has to be returned.
    As only the conflict list of $\tau_0 \in \Lambda_0$ with $\ell_q \cap \tau_0 \neq \emptyset$ is searched for the lowest active hyperplane, $h_3$ is not found.
    This issue is mitigated through purging.}
    \label{fig:lowest-cutting-empty}
\end{figure}
This is avoided as follows: when deleting a 
plane $h$, all prisms $\tau$ in the
whole structure that have $h \in \CL(\tau)$ are obtained.
For each such $\tau$, a counter is incremented.
If this counter reaches the \emph{purging threshold}
$\left|\CL(\tau)\right|/2\alpha$, where $\alpha$ is the constant from
the cutting construction above,
the prism $\tau$ and all active planes 
in $\CL(\tau)$ are marked as \emph{purged}.
When a plane is first marked as purged in the substructure where
it is stored, it is 
reinserted into the data structure, using the standard insertion. 
Naturally, purged planes are ignored
during a query and the rebuilding during an insertion.

We briefly outline the reason why the purging mechanism 
ensures the correctness of the data structure:
suppose we would like to perform a vertical ray shooting
query $q = (x, y)$, and let $\ell_q$ be the vertical line through $q$. 
Let $h$ be the plane
in $H$ that has the lowest intersection with $\ell_q$,
and let $q^* = \ell_q \cap h$.
Furthermore, let $\Xi$ be the substructure that stores $h$, and let
$\Lambda_j$ be the lowest shallow cutting (i.e., with minimum index
$j$) in $\Xi$ such that $\Lambda_j$ has 
a prism $\tau$ that contains $q^*$.
If $j = 0$, the ray shooting query will be answered
correctly. Thus,  suppose that $j > 0$.
We know that $\left|\CL(\tau)\right| \leq \alpha k_j$,  
and that $\Lambda_{j-1}$ 
covers the $k_{j - 1}$-level of $H_{j - 1}$.
Hence, $H_{j - 1}$ has at least $k_{j - 1}$ planes 
that lie vertically below $q^*$;  see~\cref{fig:prisms-tau-i}.
By construction, all these planes are also in $\CL(\tau)$,
and since $q^*$ is the answer to the query, all of them
must have been deleted from the structure.
Thus, at least $k_{j - 1} = k_j/2 \geq \left|\CL(\tau)\right|/2 \alpha$
planes in $\CL(\tau)$  have been deleted, 
so $\tau$ should have been purged already, and $h$ would
no longer be stored in $\Xi$. This gives a contradiction and
we must have $j = 0$.
Using an appropriate potential function, one can show that
deletions require $O(\log^4 n)$ amortized time~\cite{Chan20a}. The
overall space requirement of the structure is $O(n \log n)$.

As mentioned above, this approach can also be extended to arrangements of 
bivariate functions.
For the shallow cuttings, the cells are
now vertical pseudo-prisms, 
where the ceiling is a pseudo-trapezoid 
that lies on a single surface~\cite[Section~3]{kaplan_dynamic_2020}.
One can then use the adapted shallow cuttings
as a black box in their data structure for planes to obtain the
extended result~\cite[Theorem 8.3, Theorem 8.4]{kaplan_dynamic_2020}.
Recently, Liu~\cite[Corollary~4.3]{Liu20} presented an improved construction
for the shallow cuttings that leads to fewer $\log$-factors.
The analysis is similar to the one sketched above, 
one simply needs to change a few
parameters.
Under appropriate, reasonable, conditions on the bivariate
functions, 
the data structure allows insertions in 
$O(\log^2 n)$ 
amortized expected time, deletions in $O(\log^4 n)$
amortized expected time, and a query requires $O(\log^2 n )$ time.
The data structure requires $O( n \log n )$ space.

\subsection{Sampling planes}
\label{sec:sampling-planes}

We first construct a data structure for the 
problem of sampling a random disk containing a given point, 
from a dynamic set of disks.
Using standard lifting~\cite{agarwal_range_1994, yao_yao_1985}, 
this is the same as sampling a 
random plane in $\R^3$ that does not lie above a given
three-dimensional query point, from a dynamic set of
planes. For this, we extend the data structure 
by Kaplan et al.

Recall from \cref{par:a-deeper-dive-kaplan} that 
their data structure uses $O(\log n)$  substructures
$\Xi$, each of which uses a logarithmic number of shallow 
cuttings $\Lambda_j$, each with an associated set of planes $H_j$.
Upon deletion, a plane is only marked as deleted in the respective 
substructures, and it remains stored there.
As explained in \cref{par:a-deeper-dive-kaplan},
the main idea behind the data structure lies in
the \emph{purging mechanism} which clears out
the conflict list of a prism once too many planes
in it are deleted. This ensures that for any query point $q$
that lies in a non-purged prism $\tau$, but not in the 
shallow cutting below it,
at least one plane in $\CL(\tau)$ below $q$ is not deleted.
We can extend this property of the structure to implement
our random sampling query.
As described in \cref{par:a-deeper-dive-kaplan},
in the original structure, the conflict list $\CL(\tau)$
in a prism $\tau$ is purged
once $\left|\CL(\tau)\right|/2\alpha$ planes in $\CL(\tau)$ have been
marked as deleted. To ensure that for any query point
$q$ that lies in $\tau$, but not in the shallow cutting below
it, a constant \emph{fraction} of the non-deleted planes in $\CL(\tau)$
lie below $q$, we need to
lower the purging threshold to $f' = 1/4 \alpha$.\footnote{Note
that the non-deleted planes below $q$ could still be marked as pruned,
and hence cannot be used for the sampling. However,
we will see below that this difficulty can be easily
dealt with.}
The essential properties of the data structure remain unchanged.

\begin{lemma}
    \label[lemma]{lem:adjusting-f}
    Suppose we change the purging threshold     
    in the Kaplan et al.\@ data structure for 
    planes~{\cite[Section~7]{kaplan_dynamic_2020}} 
   from $f = 1/2\alpha$ to $f' = 1/4\alpha$. 
   Then, this affects neither the correctness nor 
   the asymptotic time and space bounds of the structure.
\end{lemma}

The correctness argument in 
Kaplan et~al.~\cite[Lemma~7.6]{kaplan_dynamic_2020} (see also 
 \cref{par:a-deeper-dive-kaplan}) 
is unchanged: the prisms are just purged earlier.
The running time analysis~\cite[Lemma~7.7]{kaplan_dynamic_2020} 
requires only an adjustment of constants.\footnote{In the proof 
of Kaplan et al.\@, the constant $b'$ 
has to be chosen larger than originally 
(e.g., $b' \geq 8 \alpha b''$), 
as purging a prism $\tau$ releases at least
$(b'/4 \alpha - b'') \left| \CL(\tau) \right| \log N$ 
credits when changing $f$ to $f'$.} 
The asymptotic space bound is also unaffected.

\begin{lemma}
    \label[lemma]{lem:f-prime-sampling}
    Suppose we use the purging threshold 
    $f' = 1/4 \alpha$.
    Let $q \in \R^3$ be a query point, and 
    $\ell_q$ the downward vertical ray from $q$.
    Let $\Lambda_j$ and $\Lambda_{j-1}$ be two
    consecutive shallow cuttings in the same
    substructure, and suppose that $\ell_q$ intersects
    the prism $\tau_j \in \Lambda_j$ and the prism
    $\tau_{j - 1} \in \Lambda_{j - 1}$.
    Assume further that $q$ lies in $\tau_j \setminus \tau_{j - 1}$,
    and that $\tau_j$ has not been purged.
    Then, $\ell_q$ intersects at least 
    $\left|\CL(\tau_j)\right|/4\alpha$ planes from $\CL(\tau_j)$
    that are not marked as deleted.
\end{lemma}

\begin{proof}
   Since $\Lambda_{j - 1}$ covers the  
   $k_{j - 1}$-level of 
   $H_{j - 1}$, it follows that $\ell_q$ intersects
   at least $k_{j - 1}$ planes from $\CL(\tau_{j - 1})$.
   By construction, all these planes are also present in $\CL(\tau_j)$, 
   see~\cref{fig:prisms-tau-i}.
   Furthermore, we have $\left|\CL(\tau_j)\right| \leq \alpha k_j = 2\alpha k_{j - 1}$,
   as $\Lambda_j$ is an $(\alpha k_j/n_j)$-cutting for $H_j$.
   Thus, if $\tau_j$ has not been purged, 
   the downward vertical ray $\ell_q$ intersects
   more than 
   \[
     k_{j-1} - \frac{\left|\CL(\tau_j)\right|}{4\alpha}
        \geq 
     \frac{\left|\CL(\tau_j)\right|}{2\alpha} - \frac{\left|\CL(\tau_j)\right|}{4\alpha}
        = \frac{\left|\CL(\tau_j)\right|}{4 \alpha}
    \]
    non-deleted planes from $\CL(\tau_{j})$.
\end{proof}

\begin{figure}
  \centering
  \includegraphics{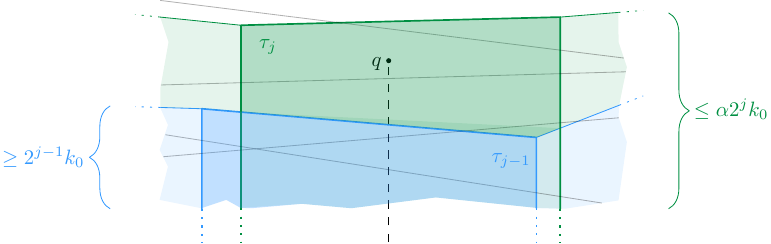}
  \caption{The situation after walking through the shallow cuttings $\Lambda_j$ for a query point $q \in \R^3$ from $j = 0$ upwards. 
  Here, $q$ lies in $\tau_j \in \Lambda_j$, but not in $\tau_{j - 1} \in \Lambda_{j - 1}$.
  Additionally, the number of planes of $H_j$ intersecting 
  the downward vertical ray from $q$ is indicated.}
  \label{fig:prisms-tau-i}
\end{figure}

To sample a plane not above a given point 
$q \in \R^3$, we go through the substructures $\Xi$.
In each $\Xi$, we walk through the shallow cuttings 
$\Lambda_j$,  from $j = 0$ upwards.
In each step, we locate the vertical prism 
that intersects the downward vertical ray from $q$ 
in $O(\log n)$ time with a suitable point location structure, 
as it is done in $\Lambda_0$ in the original data structure.
If $q$ lies \emph{inside} this prism, we stop; 
otherwise we continue with $\Lambda_{j + 1}$. 
Finally we
 apply
 \cref{lem:f-prime-sampling} to the prism we obtain,
see~\cref{fig:prisms-tau-i}.

\begin{theorem}
    \label[theorem]{thm:sample-plane}
    There is a data structure that maintains the lower envelope of 
    a dynamic set of planes in $\R^3$, such that an insertion 
    takes $O(\log^2 n)$ amortized time, a deletion takes 
    $O(\log^4 n)$ amortized time, 
    vertical ray shooting queries take $O(\log^2 n)$ worst-case  time, 
    and random sampling queries take
    $O(\log^3 n)$ expected time.
    Here, $n$ is the number of planes when the operation is performed.
    The structure requires $O(n \log n)$ space.
\end{theorem}

\begin{proof}
    We construct the data structure 
    by Kaplan et al.~\cite[Section~7]{kaplan_dynamic_2020} with 
    $f' = 1/4 \alpha$  as the purging threshold,  and with Chan's improved
    algorithm for hierarchies of shallow cuttings~\cite{Chan20a} and
    without the space optimization 
    (which we omitted in \cref{par:a-deeper-dive-kaplan}).
    According to \cref{lem:adjusting-f}, the 
    correctness and asymptotic time and space bounds are unchanged.
    We construct suitable point location structures~\cite{edelsbrunner_optimal_1986} for 
    all substructures and all shallow cuttings.
    This requires $O(n \log n)$ storage, and the  running 
    time required to construct these structures is subsumed in the time required to construct
    the vertical shallow cuttings.

    To sample a plane that lies on or below a query point $q \in \R^3$
    we proceed  as follows:
    for each substructure $\Xi$, we obtain the 
    first shallow cutting $\Lambda_j$ and the corresponding prism $\tau$ containing 
     $q$, as described above.
    If $\tau$ was purged, or if $\tau$ is in $\Lambda_0$ and contains
    no non-deleted plane that intersects the downward vertical
    ray from $q$ (which we can detect in $O(1)$ time) 
    we omit $\tau$, otherwise we add it to the set
    $\Delta$ of \emph{relevant prisms}.
    This step needs $O(\log^3 n)$ time in total.

    Now, we use rejection sampling as follows:
    we sample a random relevant prism $\tau$ from $\Delta$,
    where the sampling probability of a prism 
    $\rho \in \Delta$ is set to 
    $\left|\CL(\rho)\right|/\sum_{\rho' \in \Delta} \left|\CL(\rho')\right|$.
    Then, we choose a plane $h$ from $\CL(\tau)$, uniformly at random.
    If $h$ intersects the downward vertical ray from $q$ and 
    is not deleted, not purged, and not pruned in $\tau$, we return $h$.
    Otherwise, we repeat.

    By \cref{lem:f-prime-sampling}, with probability $\Omega(1)$
    during one sampling attempt, we obtain an $h$ that is not deleted
    and intersects the downward vertical ray from $q$.
    Furthermore, the data structure has the property that for every 
    non-deleted plane $h'$, there is exactly one
    substructure in which $h'$ is neither purged
    nor pruned. Hence, each sampling attempt
    returns a valid plane with probability $\Omega(1/\log n)$. 
    Thus, we expect $O(\log n)$ attempts until the sample succeeds.   
\end{proof}

Using standard lifting, we get the following corollary:

\begin{corollary}
    We can implement a data structure
    that maintains a dynamic set of disks in the plane
    and allows to sample a random disk from the current set that
    contains a given query point with the bounds of 
    \cref{thm:sample-plane}.
\end{corollary}

\subsection{Sampling disks}
\label{sec:sampling-disks}

To solve the more general problem of
sampling a random disk  that intersects a
given query disk, we adapt the more general data
structure of Kaplan et al.\@ for continuous bivariate 
functions of constant description 
complexity~\cite[Section~8]{kaplan_dynamic_2020}.
This structure works in the same way as the structure for planes,
the only difference being that the shallow cutting construction
is more general (see also~Liu~\cite[Corollary~4.3]{Liu20} for the improved log-factors).
Thus, applying our technique from \cref{sec:sampling-planes} in the more
general setting, we get the following result:

\begin{theorem}
    \label[theorem]{thm:le-sampling-functions}
    Let $\mathcal{F}$ be a 
    set of totally defined continuous bivariate functions 
    of constant description complexity,
    such that for every finite subset $F \subseteq \mathcal{F}$,
    the lower envelope of  $F$ has linear size.
    Then, a dynamic subset of $\mathcal{F}$ can be maintained while
    supporting vertical ray shooting queries and random sampling
    queries with the following guarantees,  where $n$ 
    is the number of functions in $F$ when the operation is performed:
    \begin{itemize}
        \item inserting a function takes 
	$O(\log^2 n)$ amortized expected time,
        \item deleting a function takes 
	$O(\log^4 n)$ amortized expected time, 
        \item a vertical ray shooting query takes $O(\log^2 n)$ time, and 
        \item sampling a random function that intersects a given downward vertical ray
	takes $O(\log^3 n)$ expected time.
    \end{itemize}
    The data structure requires $O(n \log n)$ space.
\end{theorem}

For disks, \cref{thm:le-sampling-functions} yields  the
following:

\begin{corollary}
    \label[corollary]{cor:sampling-disks}
    We can implement a data structure
    that maintains a dynamic set of disks in the plane
    and can sample a random disk from the current set that
    intersects a given query disk with the bounds of 
    \cref{thm:le-sampling-functions}.
\end{corollary}

\begin{proof}
    Let $D$ be a disk with center $c_d$ and radius $r_d$.
    We can represent the distance of any point $p \in \R^2$ 
    from $D$ as an additively weighted Euclidean distance
    function with $\delta(p, D) = \| p c_d\| - r_d$.
    The distance functions for a set of disks form a lower envelope of 
    linear complexity~\cite{sharir_intersection_1985}.
    Hence, we can apply \cref{thm:le-sampling-functions} 
    to maintain the distance functions $p \mapsto \delta(p, D)$ for
    a dynamic set of disks $D$.
    A disk intersecting a given query disk $Q$ with center
    $c_q$ and radius $r_q$
    can then be found by sampling a random function that intersects
    the downward vertical ray from the point $((c_q)_x, (c_q)_y, r_q)$, 
    as every such function satisfies $\delta(c_q, D) = 
    \| c_q c_d\Vert - r_d \le r_q$.
\end{proof}

We can now describe the revealing data structure:

\begin{theorem}[Revealing data structure (RDS)]
    \label{thm:disk-reveal}
    Let $B$ be a set of $n$ disks in the plane,
    and let $R$ be initially empty.
    We can preprocess $B$ into a data structure such 
    that the following operations are possible:
    (i) insert a disk into $R$; (ii) delete a disk from $R$; and
    (iii) delete a disk from $B$ and report all disks from
    $R$ that are revealed by the deletion.
    It takes $O(\log^3 n)$ expected time to insert a disk
    into $R$, and $O(1)$ worst-case time to delete
    a disk from $R$.
    Preprocessing $B$ and performing $k$ deletions in $B$
    requires $O(n \log^2 n + 
    k \log^4 n + m \log^4 n)$ expected time and 
    $O(n \log n + m)$ space, where $m$ 
    is the total size of $R$ and the $k$ deletions are assumed to be oblivious of the internal random choices of the data structure.
\end{theorem}

\begin{proof}
   We store $B$ in the data structure of \cref{cor:sampling-disks}.
   It takes  $O(n \log^2 n)$ expected time to
   insert all disks from $B$ into the initially empty structure, and 
   the space is $O(n \log n)$.

   To insert a new disk $r$ into $R$,
   we pick a random disk $b \in B$ that intersects $r$, and we
   \emph{assign} $r$ to $b$. 
   By \cref{cor:sampling-disks}, this takes $O(\log^3 n)$
   expected time. To delete a disk $r$ from $R$, we just remove
   it from $R$ without further action, in $O(1)$ time.
   To delete a disk $b$ from $B$, we proceed as follows:
   we find the set $R_b \subseteq R$ of disks from $R$ that
   are assigned to $b$, and for each $r \in R$, we try to sample a 
   new disk $b_r \in B$ that intersects it. If there is no such
   disk, we report $r$ as \emph{revealed}. Otherwise,
   we reassign $r$ to $b_r$. Finally, we delete
   $b$ from the structure.
   By \cref{cor:sampling-disks}, this takes amortized expected time 
   $O(\log^4 n + |R_b| \log^3 n)$.

   To bound the total time for the deletions, we fix a sequence
   of $k$ deletions
   $b_1, b_2, \dots, b_k$ of $B$, which is assumed to be oblivious of the random assignments through sampling.
   Then, the total deletion time is
   \begin{align*}
     \sum_{i = 1}^{k} \left(\log^4 n + |R_{b_{i}}| \log^3 n\right) &=
     k \log^4 n + \left(\sum_{i = 1}^{k}  |R_{b_{i}}|\right) \log^3 n\\
     &=
     k \log^4 n + \left(\sum_{r \in R}  I_r \right) \log^3 n,
   \end{align*}
   where $I_r$ denotes the number of times that the blue disk a red disk $r \in R$ is currently assigned to is deleted. 
   To bound the expected value of $I_r$, 
   consider the set $B_r$
   of all disks in the original set $B$ that intersect $r$.
   Due to the assumption of obliviousness, for the $i$-th disk $b_i$ in the deletion sequence
   the probability that $r$ is assigned to $b_i$ is at most
   $1/(|B_r| - i + 1)$, because
   this is an upper bound on the probability that $b_i$
   was chosen when $r$ was assigned the previous time to a blue disk.
   Hence, the expected value of $I_r$ is at most 
   \[
    \mathbf{E}[I_r] \leq 
    \sum_{i = 1}^{|B_r|} \frac{1}{|B_r| - i + 1} = O(\log |B_r|) =
   O(\log n),
   \]
   and the total expected time for $k$ deletions from $B$ is 
   $O(k \log^4 n + m \log^4 n)$.
\end{proof}

We now state a variant of \cref{thm:disk-reveal} that will be useful
for the decremental connectivity data structures 
in \cref{subsec:bounded:deletion,subsec:unbounded:deletion}.

\begin{corollary}
  \label{lem:reveal:cor}
  Let $R$ and $B$ be two sets of disks in $\R^2$ with 
 $|R| + |B| = n$.
  We can preprocess $R \cup B$ into a data structure 
  that supports deletions, while detecting 
  all newly revealed disks of $R$ after each deletion.
  Preprocessing needs 
  $O\left(|B|\log^2 n + |R|\log^3 n\right)$ expected time
  and the resulting data structure needs $O(|B| \log |B| + |R|)$ space.
  Deleting $k$ disks from $B$ and any number of disks from $R$ 
  needs $O\left(k\log^4 n + |R|\log^4 n\right)$
  expected time, where the $k$ deletions are assumed to be oblivious of the internal random choices of the data structure.
\end{corollary}

\section{Logarithmic dependence on \texorpdfstring{\(\Psi\)}{psi}}%
\label{sec:bounded}
We turn to the semi-dynamic setting, and
we show how to reduce the dependency on $\Psi$ 
from linear to logarithmic.
For both the incremental and the decremental scenario, 
we use the same proxy graph $H$ 
to represent the connectivity in $\cD(S)$.
The proxy graph is described in \cref{subsec:bounded:proxy}.
In \cref{subsec:bounded:deletion,subsec:bounded:insertion}
we present the data structures that are based on $H$.

Throughout this section, we assume that we work with a quadforest on a globally aligned grid.
Refer to \cref{sec:app:quadforest} for a details description of how the data structure can be adapted to work on the realRAM without global alignment.
\subsection{The proxy graph}\label{subsec:bounded:proxy}

To define the proxy graph $H$, we use the set $S$ of sites,
together with a set $\mathcal{A}$ of planar regions, to be described below.
The proxy graph $H$ has one \emph{site-vertex} for every
site in $S$ and one \emph{region-vertex} for  every region
in $\mathcal{A}$.
The regions in $\mathcal{A}$ are derived from a
quadforest for $S$, and every region $A \in \mathcal{A}$
has two 
associated sets of sites, $S_1(A)$ and $S_2(A)$.
The first set $S_1(A) \subseteq S$ 
has the property that (i) all sites of $S_1(A)$ lie in $A$;
(ii) every site $s \in S_1(A)$ has a
radius $r_s$ ``comparable'' to the diameter of $A$; and (iii)
the induced disk graph $\cD(S_1(A))$ 
is a clique.
A site $s$ can lie
in more than one set $S_1(A)$.
The second set $S_2(A)\subseteq S$  contains all sites 
$s \in S$ that (i) lie in the quadtree
cell that is ``associated'' with the region $A$;
(ii) have a  ``small'' radius relative to the diameter
of $A$; and (iii)
are adjacent in $\cD(S)$ 
to at least one site in $S_1(A)$.

The proxy graph $H$ is bipartite, with all edges going between the
site-vertices for $S$ and the 
region-vertices for $\mathcal{A}$. 
More precisely, a region-vertex
$A \in \mathcal{A}$ is connected to all site-vertices
in $S_1(A) \cup S_2(A)$. 
The edges between $A$
and $S_1(A)$ constitute a sparse representation of the
clique $\cD(S_1(A))$ in $\cD(S)$. 

The edges between $A$ and sites 
in $S_2(A)$ allow us to represent all edges in $\cD(S)$ between sites in $S_1(A)$ and sites in \(S_2(A)\) by a path of length at most two in $H$. 
This representation does not change the connectivity between the site
vertices in $H$, as compared to $\cD(S)$, since $\cD(S_1(A))$ is a clique. 
We will ensure that every edge in $\cD(S)$ is represented in $H$
in this manner.
Furthermore, 
the
number of regions in $\mathcal{A}$, as well as the total 
size of the associated sets $S_1(A)$ and
$S_2(A)$ will be ``small'', resulting in a sparse proxy graph $H$.

We now describe the details of the construction. 
The graph $H$ has vertex set $S \cup \mathcal{A}$,
where $S$ are the sites and $\mathcal{A}$ is a set of \emph{regions}.
For
every site $s \in S$, let $\sigma_s \in \G$ be the grid cell of the hierarchical grid
 with $s \in \sigma_s$ and 
$|\sigma_s| \leq r_s < 2|\sigma_s|$.
Let $N(s) = N_{15 \times 15}(\sigma_s)$ be the $15 \times 15$ neighborhood of $\sigma$ on the same level of the hierarchical grid.
Then, we let $\mathfrak{N} = (\bigcup_{s \in S} \{ \sigma_s \} \cup N(s))$, and
we construct the
quadforest $\mathcal{F}$ for $\mathfrak{N}$, as described
in \cref{sub:dg-adapting-the-unit-disk-case}.
Recall that $\mathcal{F}$ contains quadtrees that cover the
lowest $\lfloor \log \Psi \rfloor + 1$ levels of the hierarchical grid $\G$.
For every cell $\sigma$ of $\mathcal{F}$,
we define three kinds of regions:
the \emph{outer regions}, the \emph{middle regions}, and the 
\emph{inner region}.

To describe these regions, we first define for
every $d \in \mathbb{N}$
a set $\mathcal{C}_d$ of $d$ \emph{cones}
that have
(i) opening angle $2\pi/d$;
(ii) the origin as their apex; and
(iii) pairwise disjoint interiors. 
It follows that the cones in $\mathcal{C}_d$
cover the plane.
For a cell $\sigma\in \mathcal{F}$, we denote by 
$\mathcal{C}_d(\sigma)$ the set of 
all translated copies of the cones in
$\mathcal{C}_d$ whose apex has been moved to the
center
$a(\sigma)$ of $\sigma$, as shown in \cref{fig:cones}.
\begin{figure}
\centering
\includegraphics{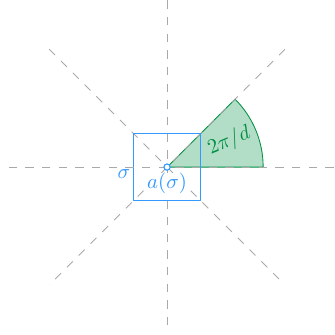}
\caption{The cones $\mathcal{C}_d$ with angle $2\pi/d$ and apex at 
the center $a(\sigma)$ of a cell $\sigma$.}
\label{fig:cones}
\end{figure}
To define the \emph{outer regions} for a cell  $\sigma \in \mathcal{F}$, 
we let $d_1  \in \mathbb{N}$  be a parameter 
to be fixed below, 
and we intersect the cones
from $\mathcal{C}_{d_1}(\sigma)$
with the annulus that is centered at $a(\sigma)$ with
inner radius $\frac{5}{2}|\sigma|$
and outer radius $\frac{9}{2}|\sigma|$. 
The \emph{middle regions} of $\sigma$ are defined similarly:
we let $d_2 \in \mathbb{N}$ be a parameter to be fixed below,
and we intersect all cones from  $\mathcal{C}_{d_2}(\sigma)$ 
with the annulus around $a(\sigma)$ that has 
inner radius $|\sigma|$ and outer
radius $\frac{5}{2}|\sigma|$.
The \emph{inner region} for $\sigma$ is the disk
with center $a(\sigma)$ and radius $|\sigma|$.
See \cref{fig:edgedef} for an illustration.
We let $\mathcal{A}_{\mathcal{F}}$ be
the set of all inner, middle, and outer regions for all
cells in $\mathcal{F}$.
\begin{figure}
\centering
\includegraphics{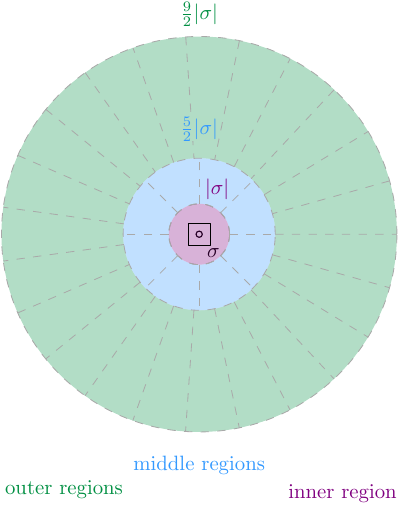}
\caption{The regions defined by a cell $\sigma$.}
\label{fig:edgedef}
\end{figure}
\begin{figure}
\centering
\includegraphics{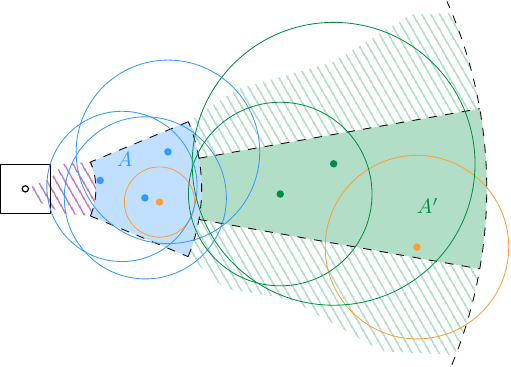}
\caption{The set $S_1(A)$ is marked blue. 
The orange site in $A$ is not in the set because its radius is too small.
The orange site in $A'$ is not in $S_1(A')$: 
even though its radius is in the correct range, it does 
not touch or intersect the inner boundary. 
}
\label{fig:conn:s1def}
\end{figure}
With every region $A \in \mathcal{A}_\mathcal{F}$,
we associate a set of sites $S_1(A) \subseteq S$, as follows:
\begin{itemize}
\item if $A$ is an outer region for a cell $\sigma$ of $\mathcal{F}$, then
the set $S_1(A)$ contains all sites $t$ such that
(i) $t \in A$; (ii)
 $|\sigma| \leq r_t < 2|\sigma|$; and
(iii) $\| a(\sigma)t\| \leq r_t + \frac{5}{2}|\sigma|$.
This means that the disk $D_t$ has (i) its
center in $A$; (ii) a radius comparable
to $|\sigma|$; and  (iii) a nonempty intersection 
with the inner boundary of the annulus used to define $A$.
\item if $A$ is a middle region or the inner region for a cell
$\sigma$ of $\mathcal{F}$, then $S_1(A)$ contains all sites $t$ 
such that
(i) $t \in A$; and
(ii) $|\sigma| \leq r_t < 2|\sigma|$.
That is, the disk $D_t$ has (i) its center in $A$
and (ii) a radius comparable
to $|\sigma|$.
\end{itemize}
We define $\mathcal{A} \subseteq \mathcal{A}_\mathcal{F}$ as the 
set of regions $A$ with $S_1(A) \neq \emptyset$.
For every $A \in \mathcal{A}$, 
we define $S_2(A) \subseteq S$  as the set of all sites $s$
that
(i) lie in $\sigma$;
(ii) have radius $r_s < 2|\sigma|$; and
(iii) are adjacent in $\cD(S)$ to at least one site in $S_1(A)$, see  \cref{fig:conn:s2def}.

\begin{figure}
\centering
\includegraphics{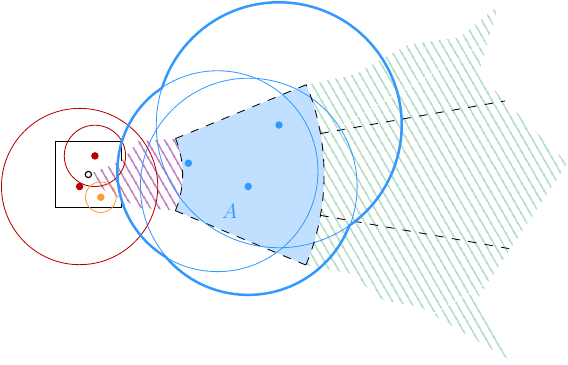}
\caption{The red sites in $\sigma$ are in $S_2(A)$. 
The radius of the orange site is in the correct range, 
but it does not intersect a site in $S_1(A)$ (marked blue).}
\label{fig:conn:s2def}
\end{figure}

We add an edge $sA$ in $H$ between a site $s$ and a region $A$ 
if and only if $s\in S_1(A) \cup S_2(A)$.
Note that the sets $S_1(A)$ and $S_2(A)$ are not 
necessarily disjoint, as for a center region $A$
defined by a cell $\sigma$, a site $s$ with $s \in \sigma$ and 
$|\sigma|\leq r_s < 2|\sigma|$ 
will satisfy the conditions for both  $S_1(A)$ and $S_2(A)$. 
However, this will  adversely affect neither the 
preprocessing time nor the correctness.
The following structural lemma helps us  to show that $H$
accurately represents the connectivity in
$\cD(S)$ as well as to bound 
the size of $H$ and
the preprocessing time in the decremental setting.

\begin{lemma}\label{lem:bounded:neighborhood}
Let $t \in S$ be a site.
Then
\begin{enumerate}
\item all cells that define a region $A$ with $t\in S_1(A)$
 are in $N(t)$; and
\label{item:lem:bouned:neighborhood:const}
\item if $s \in S$ is a site with $r_s \leq r_t$ such 
that $st$ is an edge in $\cD(S)$,
then there is a cell $\sigma \in N(t)$ 
such that $s \in \sigma$ and such that $\sigma$ 
defines a region $A$ with $t \in S_1(A)$.
\label{item:lem:bounded:neighborhood:cell}
\end{enumerate}
\end{lemma}

\begin{figure}
\centering
\includegraphics{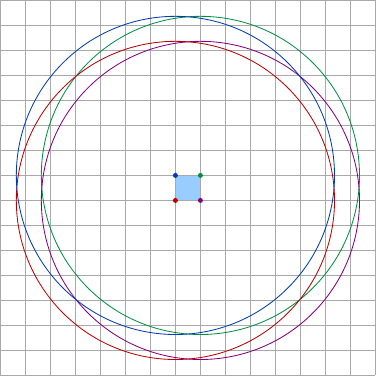}
\caption{The disk $D(t, \frac{9}{2}|\sigma_t|)$ is contained 
in $N_{15\times 15}(\sigma_t)$.}
\label{fig:neighborhood_contains}
\end{figure}

\begin{proof}
To prove the first claim, we observe that by definition, all 
cells $\sigma$ that define a region $A$ with $t \in S_1(A)$ have
$|\sigma| = |\sigma_t|$
and centers $a(\sigma)$
in the disk $D = D\left(t,\frac{9}{2}|\sigma_t|\right)$. 
The disk $D$ is contained in $N(t)$
(see \cref{fig:neighborhood_contains}), and the claim follows.

We now proceed to the second claim. Since $st$ is an edge in $\cD(S)$,
we have $\| st\| \leq r_s + r_t$.
By the definition of $\sigma_t$ and since $r_s \leq r_t$, 
we have $r_s \leq r_t <  2 |\sigma_t|$. 
Now, let $\sigma \in \G$ be the cell with $|\sigma| = |\sigma_t|$ 
and $s\in \sigma$. 
Then,  the triangle inequality shows 
\begin{equation}
\label{equ:triang}
\| a(\sigma) t\| \leq 
\frac{1}{2}|\sigma|+ r_s + r_t
< \frac{5}{2}|\sigma| + r_t
< \frac{9}{2}|\sigma|.
\end{equation}
Thus, $a(\sigma)$ lies in the disk
$D=D\left(t,\frac{9}{2}|\sigma|\right)$,
and $\sigma \in N(t)$,
By symmetry, we also have
$t\in D\left(a(\sigma),\frac{9}{2}|\sigma|\right)$,
and hence $t \in S_1(A)$. Note that \(A\) can be an inner, middle or outer region. 
The inequalities in (\ref{equ:triang}) also show
that the $S_1$-condition (iii) for an outer region is fulfilled.
\end{proof}

Before we argue that $H$ 
accurately represents the connectivity of $\cD(S)$, we 
show that the associated sites of a region in $\mathcal{A}$ 
form a clique in $\cD(S)$.

\begin{lemma}\label{lem:bounded:clique1}
Suppose that $d_1 \geq 23$ and $d_2 \geq 8$.
Then, for any region $A \in \mathcal{A}$, 
the sites in $S_1(A)$ induce a clique in $\cD(S)$.
\end{lemma}

\begin{proof}
Let $\sigma$ be the cell that defines $A$.
We distinguish three cases.

\subparagraph*{Case 1: $A$ is an outer region.}
Let $t \in S_1(A)$, and let $C \in \mathcal{C}_{d_1}(\sigma)$ be the cone
that was used to define $A$.
Consider the line segments $\ell_1$ and $\ell_2$ that 
that go through $t$, lie in $C$, and are perpendicular to the upper and
the lower boundary of $C$, respectively.
We denote the endpoints of $\ell_1$ by $p_1$ and $p_2$, and 
the endpoints of $\ell_2$ by $q_1$ and $q_2$, where $p_1$ and $q_1$ 
are the endpoints at the right angles; 
see \cref{fig:cliquec1}.

Now, let $Z_t$ be the convex hull of $p_1$, $p_2$, $q_1$, and $q_2$.
We claim that $Z_t \subset D_t$.
Indeed, 
let $\alpha_1$, $\alpha_2$ be the angles at $a(\sigma)$ defined by the line 
segment $\overline{a(\sigma)t}$ and the two boundaries of $C$.
By our choice of $C$, we have $\alpha_1 + \alpha_2 = 2\pi/d_1$.
\begin{figure}
\center
\includegraphics{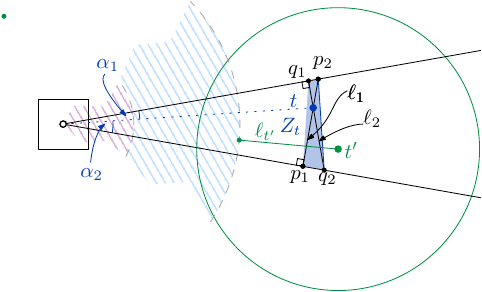}
\caption{The line segment $\ell_{t'}$ intersects $p_2q_2$.}
\label{fig:cliquec1}
\end{figure}
Basic trigonometry yields
\begin{align*}
\| a(\sigma)p_1 \| &= \cos(\alpha_2) \cdot \| a(\sigma) t\|,&
\| p_1p_2 \| &= \| a(\sigma)p_1\| \cdot \tan\left(2\pi/d_1\right),\\
\| a(\sigma)q_1 \| &= \cos(\alpha_1) \cdot \| a(\sigma) t\|, \text{ and }&
\| q_1q_2 \| &= \| a(\sigma)q_1\| \cdot \tan\left(2\pi/d_1\right).
\end{align*}
Hence,
\begin{align*}
\| \ell_1\| &= \| p_1p_2\|
= \cos(\alpha_2) \cdot \| a(\sigma) t\| \cdot \tan\left(2\pi/d_1\right) \\
\| \ell_2\| &= \| q_1q_2\|
= \cos(\alpha_1) \cdot \| a(\sigma) t\| \cdot \tan\left(2\pi/d_1\right)
\end{align*}
By $|\cos(\alpha_1)|, |\cos(\alpha_2)| \leq 1$, 
the properties $\|a(\sigma)t\| \leq r_t + (5/2)|\sigma|$ 
and $|\sigma| \leq r_t$ in the definition of
$S_1(A)$, and our choice of $d_1$, we get
\begin{align*}
\| \ell_1 \| , \| \ell_2\| & \leq \| a(\sigma) t\| 
\tan\left(\frac{2\pi}{d_1}\right)
\leq \left(r_t + \frac{5}{2}|\sigma|\right)
\cdot \tan\left(\frac{2\pi}{d_1}\right)
\leq \frac{7}{2} r_t \cdot \tan\left(\frac{2\pi}{d_1}\right)
\leq r_t.
\end{align*}
Since $t$ lies on $\ell_1$ and $\ell_2$, it follows that 
$Z_t \subset D_t$, as claimed.

Next, let
$\ell_t = a(\sigma)t \setminus 
D\left(a(\sigma),\frac{5}{2}|\sigma|\right)$ be the
part of the line segment between $t$ and $a(\sigma)$ that 
lies outside of $D\left(a(\sigma), \frac{5}{2}|\sigma|\right)$.
By the 
property
$\| a(\sigma)t \|  \leq r_t + \frac{5}{2}|\sigma|$ in
the definition of $S_1(A)$, we have
that $\ell_t \subset D_t$.

Consider two sites $t, t' \in S_1(A)$.
We argue that $D_t$ and $D_{t'}$ intersect
and hence $t$ and $t'$ are adjacent in $\cD(S)$.
Indeed, if $t' \in Z_t \subset D_t$
or $t \in Z_{t'} \subset D_{t'}$,
this is immediate.
Thus, suppose that this does not hold, and
assume without loss of generality that $\|a(\sigma)t\|
\leq \|a(\sigma)t'\|$.
As $\| a(\sigma) t\| \geq \frac{5}{2}|\sigma|$, the 
line segment $p_2q_2$ lies completely outside of 
$D\left(a(\sigma), \frac{5}{2}|\sigma|\right)$,
so $t'$ is separated in $C$ from 
$a(\sigma)$ by the convex hull $Z_t$.
Then, $\ell_{t'} \subset D_{t'}$ intersects $p_2q_2 \subset D_t$, 
as claimed.

\subparagraph*{Case 2: $A$ is a middle region.}
The law of cosines and our choice of $d_2$ yield 
\[
\diam(A)  \leq |\sigma|\cdot \sqrt{1 + 
\frac{25}{4} - 
\frac{10}{2}\cos\left(\frac{2\pi}{d_2}\right)} 
\leq 2|\sigma|.
\]
Now, let $t, t' \in S_1(A)$.
By the properties in the definition of $S_1(A)$, 
we have $t,t'\in A$ and  $r_t, r_{t'} \geq |\sigma|$. 
Thus, we get
\[
\|tt'\| \leq \diam(A) \leq 2|\sigma| \leq r_t + r_{t'},
\]
and $tt'$ is an edge in $\cD(S)$.

\subparagraph*{Case 3: $A$ is an inner region.}
Let
$t, t' \in S_1(A)$. Since $t$ and $t'$ both lie in the
disk 
$D\left(a(\sigma), |\sigma|\right)$ of diameter $2|\sigma|$, 
and since 
 $r_t, r_{t'} \geq |\sigma|$, 
we again get that $tt'$ is an edge in $\cD(S)$.
\end{proof}

We can now
show that $H$ accurately represents the connectivity of $\cD(S)$.

\begin{lemma}\label{lem:bounded:conn}
Two sites are connected in $H$ if and only if they are 
connected in $\cD(S)$.
\end{lemma}

\begin{proof}
Let $s, t \in S$.
First, we show that if $s$ and $t$ are connected in $H$, 
they are also connected in $\cD(S)$.
The path between $s$ and $t$  in $H$ alternates between vertices 
in $S$ and vertices in $\mathcal{A}$.
Thus, it suffices to show that if two sites $u$ and $u$'
are adjacent in $H$ to the same region $A \in \mathcal{A}$,
they are connected in $\cD(S)$. 
This follows directly from \cref{lem:bounded:clique1}:
if $u, u' \in S_1(A)$, they are part 
of a clique, and hence adjacent.
If $u \in S_2(A)$, then
there is at least one site
in $S_1(A)$ whose disk intersects $D_u$, and
hence $u$ is connected to all sites in the clique $S_1(A)$. 
Thus, if $u' \in S_1(A)$, we are done, and
if $u' \in S_2(A)$, the same argument shows that $u'$ must  also 
be connected to
all sites in $S_1(A)$, and hence $u$ and $u'$ are connected through
$S_1(A)$.

Now, we consider two sites that are connected in $\cD(S)$,
and we show that they are also connected in $H$.
It suffices to show that if $s$ and $t$ are adjacent 
in $\cD(S)$, they are connected in $H$.
Assume without loss of generality that $r_s \leq r_t$, and let $\sigma$ 
be the cell in $N(t)$ with $s\in\sigma$.
The cell $\sigma$ exists by 
\cref{lem:bounded:neighborhood},
it belongs to $\mathcal{F}$, since $\sigma$
lies in the first $\lfloor \log \Psi \rfloor + 1$ levels
of $\G$, and we have $\sigma_s \subseteq \sigma$.
From \cref{lem:bounded:neighborhood} also follows, that $t \in S_1(A)$ for some 
$A\in \mathcal{A}_\mathcal{F}$ defined by $\sigma$. 
As the regions with non-empty sets $S_1(A)$ are in $\mathcal{A}$,
the edge $tA$ exists in $H$.
Now we argue that $s\in S_2(A)$, and thus the edge $As$ 
also exists in $H$.
This follows by straightforward checking of 
the properties from the definition of  $S_2(A)$:
(i)
by the choice of $\sigma$, we have $s\in \sigma$;
(ii)
by the definition of $N(t)$ and the
assumption $r_s \leq r_t$, we have
$r_s < 2|\sigma|$; and (iii)
we already know that $t \in S_1(A)$, and 
$s$ and $t$ are adjacent in $\cD(S)$.
Thus, $s$ and $t$ are connected in $H$ through $A$, 
and the claim follows.
\end{proof}

Finally, 
we show that the size of $H$ depends
only on $n$ and $\Psi$, and not on the number of edges in $\cD(S)$
or the diameter of $S$.
We first bound the total size of the sets $S_1(A)$  
and $S_2(A)$.

\begin{lemma}\label{lem:bounded:S1S2size}
We have	$\sum_{A\in \mathcal{A}} |S_1(A)| = O(n)$ and 
$\sum_{A \in \mathcal{A}} |S_2(A)| = O(n\log \Psi)$.
\end{lemma}

\begin{proof}
First, we bound the total size of the sets $S_1(A)$.
Fix a site $t \in S$.
By 
\cref{lem:bounded:neighborhood}, 
the site $t$ can lie only in regions $A$ that are defined by cells in $N(t)$,
and there are $O(1)$ such cells. Thus, $t$ lies in at most $O(1)$ sets
$S_1(A)$, and since $t$ was
an arbitrary site, we get $\sum_{A\in \mathcal{A}} |S_1(A)| = O(n)$.
Next, we focus on the total size of all sets $S_2(A)$.
Fix a site $s \in S$.
A necessary condition for 
$s \in S_2(A)$ is that $s$ lies in the
cell that defines $A$. There
are at most $\lfloor \log \Psi \rfloor + 1$  
cells containing \(s\)
in $\mathcal{F}$, and every such cell defines $O(1)$ regions
$A$.
Since $s$ was arbitrary, it follows that 
$\sum_{A\in \mathcal{A}} |S_2(A)| = O(n\log \Psi)$.
\end{proof}

\begin{corollary}\label{lem:bounded:size}
The proxy graph $H$ has $O(n)$ vertices and $O(n\log\Psi)$ edges.
\end{corollary}

\begin{proof}
There are at most $\sum_{A\in\mathcal{A}} |S_1(A)|$ non-empty regions
$A$, so 
 $|\mathcal{A}| = O(n)$. 
This gives the bound on the number of vertices.
The number of edges is at most
$\sum_{A\in \mathcal{A}} |S_1(A)| + |S_2(A)| = O(n\log \Psi)$.
\end{proof}

\subsection{Decremental data structure}\label{subsec:bounded:deletion}

We use the proxy graph $H$ from
\cref{subsec:bounded:proxy} to build 
a data structure that allows interleaved
deletions and connectivity queries in a disk graph. 
The data structure has several components: 
we store a quadforest that contains the cells defining $\mathcal{A}$, and
for every $A \in \mathcal{A}$, 
we store the sets $S_1(A)$ and $S_2(A)$.
For each region $A\in \mathcal{A}$, we store a revealing data structure (\RDS) 
as in \cref{lem:reveal:cor} with $B = S_1(A)$ and
$R = S_2(A)$.
Finally, we store the proxy graph $H$ in an
HLT-structure 
\(\mathcal{H}\)~\cite{holm_poly-logarithmic_2001}.
See \cref{fig:bounded:dec} for an illustration.

\begin{figure}
\centering
\includegraphics{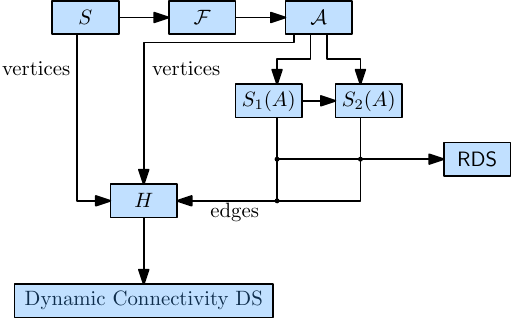}
\caption{The structure of the decremental data structure.}
\label{fig:bounded:dec}
\end{figure}

As usual, the connectivity queries are answered through $\mathcal{H}$. 
To delete a site $s$,
we first remove from $\mathcal{H}$ all  edges incident to $s$.
Then, we go through all regions $A$ with $s\in S_1(A)$.
We remove $s$ from $S_1(A)$ and from the \RDS of $A$. Let $U$ be the set of 
revealed sites from $S_2(A)$ that are reported by the \RDS.
We delete every $u\in U$ from $S_2(A)$ and from the corresponding \RDS.
Additionally, we delete the edges $uA$ 
from $\mathcal{H}$ for all $u\in U$ that are not also in $S_1(A)$.
Next, for each  region $A$ with  $s\in S_2(A)$,
we simply remove $s$ from the set $S_2(A)$ and the associated \RDS.
First, we analyze the preprocessing time.

\begin{lemma}\label{lem:bounded:preprocessing}
Given a set $S$ of $n$ sites, we can construct the data 
structure described
 above in $O\left(n\log \Psi \log^3 n\right)$ 
 time.
The data structure requires $O(n (\log n + \log \Psi))$ space. 
\end{lemma}

\begin{proof}
For each $s \in S$, we identify the cell $\sigma_s$ of $\G$
that has $s \in \sigma_s$ and  $|\sigma_s| \le r_s < 2 |\sigma_s|$.
Let $\mathfrak{N} = \bigcup_{s \in S} (\{ \sigma_s \} \cup N(s))$.\footnote{Here, in the decremental setting, $\mathfrak{N} = \{ \sigma_s \mid s \in S\}$ would already be sufficient, as the omitted cells would never have a nonempty $S_2(A)$ for their regions $A$.}
We build the quadforest $\mathcal{F}$
for $\mathfrak{N}$, together with a red-black tree~\cite{cormen_introduction_2009} that contains the roots, as described in 
\cref{sub:dg-adapting-the-unit-disk-case}.
As in \cref{sub:dg-adapting-the-unit-disk-case},
the cells of $\mathcal{F}$ lie on the first 
$\lfloor \log \Psi \rfloor + 1$ levels of $\G$,
and it takes 
$O(n(\log \Psi + \log n))$ steps to construct it.

Now we find the sets $S_1(A)$ and $S_2(A)$.
Fix a site $t \in S$.
By \cref{lem:bounded:neighborhood}, we can
find all sets $S_1(A)$ that contain $t$ by looking
at the regions defined by cells in $N(t)$.
Thus, we find $N(t)$, and for every $\sigma \in N(t)$,  we 
iterate over all  regions $A$ defined by  $\sigma$,
identifying those with $t \in S_1(A)$.
Identifying the cells in $N(t)$ takes  $O(\log n + \log \Psi)$ time per
cell and also overall, since $N(t)$ has constant size.
Thus, the sets $S_1(A)$ can be constructed 
in total time $O(n(\log n + \log \Psi))$

To find the sets $S_2(A)$, we build an  
\emph{additively
weighted Voronoi diagram} for each set $S_1(A)$,
where a site $t\in S_1(A)$ has weight $-r_t$.
Such a diagram can be constructed in
$O\left(|S_1(A)|\log n\right)$ time, and 
additively weighted nearest-neighbor queries
in $S_1(A)$ then take  $O(\log n)$ 
time~\cite{berg_computational_2008,fortune_sweepline_1987,sharir_intersection_1985}.
By \cref{lem:bounded:S1S2size}, we
need $O(\sum_{A\in \mathcal{A}} |S_1(A)| \log n) = O(n\log n)$ 
time to compute
all the diagrams. For a site $s\in S$, let $\pi_s$ be the path in 
$\mathcal{F}$
from the root to the cell $\sigma_s$.
For each cell along $\pi_s$, we query all additively weighted Voronoi
diagrams for its regions $A$ with $s$.
If $s$ intersects the reported nearest 
neighbor, we add $s$ to $S_2(A)$.
A site $s$ is used for $O(\log \Psi)$ queries, 
for an overall of $O(n\log n\log \Psi)$ time to find all sets $S_2(A)$.

The edges of $H$ are determined by the sets $S_1(A)$ and $S_2(A)$.
We insert the edges 
into an initially empty HLT-structure, so we obtain the
 connectivity data structure $\mathcal{H}$ in overall 
 $O\left(n\log \Psi \log^2 n\right)$ time.
 Following \cref{lem:bounded:size}, $\mathcal{H}$ requires $O(n \log \Psi)$ space.

For every region $A \in \mathcal{A}$,
we build  the \RDS with 
$B = S_1(A)$ and $R = S_2(A)$ in 
$O\left(|S_1(A)|\log^2 n + |S_2(A)|\log^3 n \right)$ 
expected time,
using  \cref{lem:reveal:cor}.
Each of the resulting \RDS requires $O(|S_1(A)| \log |S_1(A)| + |S_2(A)|)$ space.
Summing over all regions and using 
\cref{lem:bounded:S1S2size}, we get a total expected time of
\begin{align*}
&~~\phantom{=} O\left(\sum_{A\in\mathcal{A}} |S_1(A)|\log^2 n +|S_2(A)|\log^3 n \right)\\
&= O\left(\left(\sum_{A\in \mathcal{A}} |S_1(A)|\right)\cdot \log^2 n + \left(\sum_{A\in \mathcal{A}}|S_2(A)|\right) \cdot \log^3 n \right)\\
&=O\left(n\log^2 n + n \log\Psi \log^3 n \right)
= O(n \log \Psi \log^3 n),
\end{align*}
and this dominates the preprocessing time.

Similarly, summing over all regions and using 
\cref{lem:bounded:S1S2size} yields a total space usage for the revealing data structures of
\begin{align*}
    &~~\phantom{=} O\left(\sum_{A\in\mathcal{A}} |S_1(A)| \log |S_1(A)| + |S_2(A)| \right)\\
    &= O\left(\left( \sum_{A\in\mathcal{A}} |S_1(A)| \right) \log n + \left( \sum_{A\in\mathcal{A}} |S_2(A)| \right) \right) \\
    &= O(n (\log n + \log \Psi)),
\end{align*}
which dominates the space usage for the regions and their associated sets, the quadforest $\mathcal{F}$, the HLT-structure $\mathcal{H}$, and also the  temporarily constructed additively weighted Voronoi diagrams.

\end{proof}

Now we show that the data structure correctly and efficiently handles
queries and deletions.

\begin{theorem}\label{thm:bounded:deletion} The data structure
handles deletions of sites in overall expected time
$O\left(
n\log\Psi\log^4n  
\right)$, assuming the deletions are oblivious of the internal random choices of the data structure. 
Furthermore, it requires
$O(\log n/\log\log n)$ time to answer connectivity queries correctly and requires $O(n (\log n + \log \Psi))$ space.
\end{theorem}

\begin{proof} 
We first show that the answers given by our data structure are indeed correct.
Over the lifetime of the data structure, we
maintain the invariant that the sets $S_1(A)$ and $S_2(A)$ always contain
the sites as defined in \cref{subsec:bounded:proxy}, the graph stored in
$\mathcal{H}$ is the proxy graph $H$, and each \RDS associated 
with a region
$A$ contains the sets $S_1(A)$ and $S_2(A)$.
Assuming that this invariant
holds, \cref{lem:bounded:conn} implies that the connectivity
queries are answered correctly.

To show that the invariant is maintained when deleting a
site $s$, we first note that removing $s$ 
from a set $S_2(A)$ only leads to the deletion of a single 
edge in $H$.
As we make sure to mirror the removal from $S_2(A)$ in $\mathcal{H}$
and the \RDS, removing $s$ from all sets $S_2(A)$
that contain it maintains the second half of the invariant.

Now, let $A$ be a region such that 
$s$ lies in $S_1(A)$.
Then, for all sites $t$ in the corresponding set $S_2(A)$,
it is necessary that $t$ intersects at least one site in $S_1(A)$.
Furthermore, there is a---possibly empty---set $U'$ 
of sites in $S_2(A)$ that only intersect $s$. 
So to maintain the invariant, we have to delete $U'$ from 
$S_2(A)$ and the associated \RDS.
As $U'$ contains exactly the sites reported in the set $U$ 
returned by the \RDS, the sites from  $U'$ are removed by construction, 
and the invariant on $S_2(A)$ and the \RDS is maintained.
Since we do not delete the edges that were present because a site was in 
$S_1(A) \cap S_2(A)$, the graph stored in $\mathcal{H}$ 
is still the proxy graph $H$, and the invariant is maintained.

By these observations, applied to $S_1(A)$ and $S_2(A)$ for all regions $A$ during a deletion, and by the fact that a deletion removes all edges incident to the deleted site, the invariant holds.

Now, we analyze the running time.
Queries are performed to $\mathcal{H}$ take
$O(\log n/\log \log n)$ time.
Every edge is removed exactly once from $\mathcal{H}$, for a total of
$O\left(n\log\Psi\log^2n\right)$ time.
Finding the regions $A$ whose sets $S_1(A)$ and $S_2(A)$ have to be 
updated during a deletion takes again $O(\log n+\log \Psi)$ time, 
by similar argument as in the proof of \cref{lem:bounded:preprocessing}.
The  running time is dominated by the deletions from the \RDS.
By \cref{lem:reveal:cor}, the \RDS associated with a single region 
adds
$O\left(|S_1(A)| \log^2 n + k_A \log^{4} n + 
|S_2(A)|\log^4 n\right)$ expected steps
to the total running time, where $k_A$ is the number 
of sites deleted from $S_1(A)$ and the deletions are assumed to be oblivious of the internal random choices of the \RDS{}. 
Summing over all regions, we have 
$\sum_{A\in\mathcal{A}} k_A = O(k)$, since every 
site is contained in $O(1)$ sets $S_1(A)$.
Furthermore, as we have $\sum_{A\in\mathcal{A}} |S_1(A)| = O(n)$ and
$\sum_{A\in\mathcal{A}} |S_2(A)| = O(n\log\Psi)$, 
an overall running time of 
\[
O\left(n\log^2 n +  k \log^{4}n +
 n\log \Psi \log^4 n \right)
= O\left(n\log \Psi \log^4 n + 
 k \log^{4} n\right)
\]
 for $k$ deletions follows. 
As \(k \in O(n\log\Psi)\) and the
 the space usage is unchanged from \cref{lem:bounded:preprocessing}, the theorem follows.
\end{proof}

\subsection{Incremental data structure}\label{subsec:bounded:insertion}
Next, we describe our incremental connectivity data structure 
for the bounded radius ratio case.
It is also based on the proxy graph $H$ from 
\cref{subsec:bounded:proxy}.
Since we only do insertions, we use the connectivity data structure 
from \cref{thm:generalinc} for $\mathcal{H}$.
This data structure achieves $O\left(1\right)$ amortized time 
for updates and $O\left(\alpha(n)\right)$ amortized time for queries.

To update the edges incident to a region $A$ defined by a 
cell $\sigma$, we use two fully dynamic additively weighted 
nearest neighbor data structures (\AWNN, \cref{lem:prelims:dynamicNN}): 
one for the  set $S_1(A)$ and one for the set $\overline{S_2(A)}$ 
that contains those sites $s \in \sigma$ with radius $r_s < 2|\sigma|$ 
that have not been added to $S_2(A)$ yet.
As before, we maintain a quadforest 
$\mathcal{F}$ of height $\lfloor \log \Psi \rfloor + 1$ 
to navigate the cells.
See \cref{fig:bounded:inc} for an illustration of the data structure.

\begin{figure}
\centering
\includegraphics[page=2]{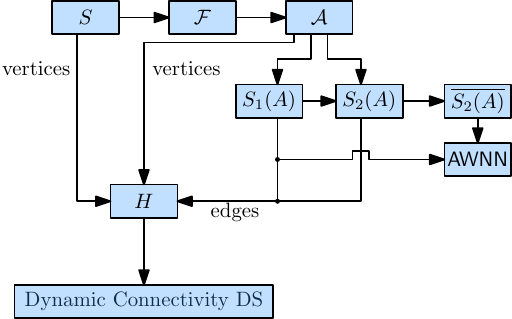}
\caption{The structure of the incremental data structure.}
\label{fig:bounded:inc}
\end{figure}

The data structure works as follows: when inserting a site $s$, 
we determine the cells of the neighborhood $N(s)$, and
add those among then that are not in $\mathcal{F}$ to
 $\mathcal{F}$.
Furthermore, we add the associated region vertices to $H$ 
and also to the dynamic connectivity graph structure $\mathcal{H}$.
Then, we have to connect the site $s$ to the regions.
Hence, we have to identify the sets $S_1(A)$, $S_2(A)$, and 
$\overline{S_2(A)}$ that the site $s$ belongs to, add $s$ to the 
corresponding \AWNN{}s, and insert the edges incident to $s$ 
into $\mathcal{H}$.
After the insertion of $s$ into a set $S_1(A)$, we also query the \AWNN 
of the associated set $\overline{S_2(A)}$ to find possible sites 
that intersect $D_s$ and thus have to be moved to $S_2(A)$, 
do so if required, and add the edges incident to the transferred site to 
$\mathcal{H}$.

\begin{theorem}\label{thm:bounded:insertion} The data structure
    described above correctly answers connectivity queries in $O(\alpha(n))$ 
    amortized time, performs insertions in
    $O(\log\Psi \log^{4} n)$ 
    amortized expected time, and requires $O(n \log \Psi \log n)$ space, where $n$ 
    is the number of sites stored in the data structure.
\end{theorem}

\begin{proof}    
    First, we show correctness, using
    similar invariants as in the proof of \cref{thm:bounded:deletion}.
    For every region $A$,
    \begin{enumerate}
        \item the \AWNN of $S_1(A)$ contains exactly the sites in $S_1(A)$,
	\label{ins:(a)}
        \item the \AWNN $\overline{S_2(A)}$ 
	contains exactly the sites that would lie in $S_2(A)$ 
	if there was a disk in $S_1(A)$ intersecting them,
        \label{ins:(b)}
        \item the sets $S_1(A)$ and $S_2(A)$ 
	represented in our data 
	structure always contain the sites as defined in 
	\cref{subsec:bounded:proxy}; and
        \label{ins:(c)}
        \item the graph stored in $\mathcal{H}$ is the proxy graph 
	$H$.
        \label{ins:(d)}
    \end{enumerate}
    Note that 
    we do not need to store the sets $S_2(A)$ explicitly, since
    no sites are deleted from $S_2(A)$ and since 
    $S_2(A)$ is not used during the insertion procedure. Thus,
    the set $S_2(A)$ is implicitly represented 
    by the edges of $\mathcal{H}$.
    Under the assumption that the invariants hold, 
    \cref{lem:bounded:conn} again implies that connectivity 
    queries are answered 
    correctly.

    Invariant~\ref{ins:(a)} and the first part of 
    Invariant~\ref{ins:(c)} hold by construction.
    When inserting a site $s$, it is added to the sets $S_1(A)$ 
    of all regions $A$ with $A \in N(s)$ that $s$ belongs to.
    By \cref{lem:bounded:neighborhood}, these are the 
    only regions potentially containing $s$.
    When adding $s$ to $S_1(A)$, we also add $s$ 
    to the \AWNN of $S_1(A)$, and insert the edge $sA$  into 
    $\mathcal{H}$. This is exactly the edge we need  in $H$.
    To complete the proof of Invariants~\ref{ins:(c)} and~\ref{ins:(d)}, 
    we first need to prove that Invariant~\ref{ins:(b)} holds.
    
    By definition of $S_2(A)$, for a region $A$ defined by a cell $\sigma$,
    we  only have in $S_2(A)$ sites $s$ with 
    $s \in \sigma$, $r_s < 2|\sigma|$, which satisfy the constraint that $D_s$ 
    intersects at least one site in $S_1(A)$.
    Furthermore, we add $s$  to $\overline{S_2(A)}$, if it fulfills the 
    first two but not the last constraint.
    This guarantees that a site $s$ is potentially present in the sets 
    $S_2(A)$ of the regions defined by any cell on the path from the root to 
    $\sigma_s$, as these are exactly the cells for which the radius 
    constraint of $S_2(A)$ holds.
    We check for all those regions if $s$ intersects any site in $S_1(A)$.
    The site is only inserted into $S_2(A)$ if there is an intersection, 
    otherwise, it is added to the \AWNN of $\overline{S_2(A)}$.
    Furthermore, a site in $\overline{S_2(A)}$ is only transferred from 
    $\overline{S_2(A)}$ to $S_2(A)$ if a newly inserted site to $S_1(A)$ 
    intersects it. This proves Invariant~\ref{ins:(b)}.
    Whenever we assign a site to a region $S_2(A)$, during its insertion or 
    due to a change of the corresponding set $S_1(A)$, we add the edge 
    $sA$ to $\mathcal{H}$, thus representing the edge we need  in $H$.
    This also concludes the proof of Invariants~\ref{ins:(c)} 
    and~\ref{ins:(d)}. The correctness follows.
    
    Now, we consider the finer details and the running time.
    When inserting a site $s$, we first have to identify the
    root of the quadtree in $\mathcal{F}$ that contains $\sigma_s$.
    Then, we can descend into this quadtree to find $\sigma_s$, creating
    new quadtree nodes when necessary, including in the neighborhoods of the cells on the path towards $\sigma_s$.
    As in \cref{sub:dg-adapting-the-unit-disk-case}, this 
    takes $O\left(\log n + \log \Psi \right)$ time, where $n$ 
    is the number of sites at the time of the insertion.
    When introducing new cells into $\mathcal{F}$, 
    we also need their associated regions to be present in the proxy 
    graph $H$.
    Hence, we insert the regions of the new cells as isolated vertices 
    to $\mathcal{H}$,  and we initialize the corresponding \AWNN{}s 
    for $S_1(A)$ and $\overline{S_2(A)}$.
    Subsequently, we have to insert $s$ and its incident edges:
    first, we add $s$ as an isolated vertex to $\mathcal{H}$.
    Then, we insert $s$ into the sets it needs to be contained in and, if applicable, to 
    the associated \AWNN{}s.
    Since we implicitly store $S_1(A)$ and $S_2(A)$, 
    we add the corresponding edges between the regions and $s$ to 
    $\mathcal{H}$ along the way.
    
    Recall that by \cref{lem:bounded:neighborhood},
    the site $s$ can only be in a set $S_1(A)$ defined by a cell in $N(s)$.
    Thus, we iterate through the cells of $N(s)$ and we add $s$ 
    with weight $-r_s$ to the \AWNN{}s of the sets $S_1(A)$ 
    it belongs to.  
    We also introduce the respective edges 
    $sA$ into $H$.
    As $N(s)$ is of constant size, this step takes $O(1)$ 
    amortized time for the insertion of edges in $\mathcal{H}$ and 
    $O(\log^2 n)$ amortized expected time for insertions into the 
    \AWNN{}s.
    As the size of \(N(s)\) is a constant, this step can be implemented without explicit pointers to neighboring cells by traversing the quadtree from the root to each cell in \(N(s)\) without any overhead.
    
    Furthermore, for a region $A$, changes in $S_1(A)$ 
    may affect the set $S_2(A)$.
    Hence, we have to check if the insertion of $s$ into $S_1(A)$ causes
     a site $t$ to move from $\overline{S_2(A)}$ to $S_2(A)$ 
    because $D_s$ intersects $D_t$.
    We identify these sites by successive weighted nearest-neighbor queries 
    with $s$ in the \AWNN of $\overline{S_2(A)}$, stopping when 
    the resulting site does not intersect $D_s$.
    For each site $t$ that intersects $D_s$, we delete $t$ 
    from the \AWNN of $\overline{S_2(A)}$ 
    and insert the edge $tA$ into $\mathcal{H}$.
    The amortized expected running time of $O(\log^{4} n)$ 
    to delete $t$ from $\overline{S_2(A)}$ dominates this step.
    We do not know how many deletions are performed in this step, but a site 
    $t$ can be present in the most $O(\log \Psi )$ sets $\overline{S_2(A)}$, 
    associated to the cells containing $t$ on the path from the root to 
    $\sigma_t$.
    Furthermore, a site $t$ is inserted into $\overline{S_2(A)}$ only once.
    Thus, we can charge those deletions to the previous insertions, 
    which yields an amortized expected running time of 
    $O(\log\Psi \log^{4} n)$ for this step per insertion.
    
    Now, we consider the time needed to insert  $s$ either into $S_2(A)$ 
    or $\overline{S_2(A)}$ for some regions $A$.
    By definition, we know that $s$ has to be added either to $S_2(A)$ 
    or to $\overline{S_2(A)}$ of the regions of the cells 
    along the path from the corresponding root to $\sigma_s$ in $\mathcal{F}$.
    Let $A$ be a region defined by a cell along this path. To 
    determine whether to add $s$ to $S_2(A)$ or to $\overline{S_2(A)}$, 
    we perform a weighted nearest neighbor query  in the \AWNN of $S_1(A)$.
    If we find a site in $S_1(A)$ that intersects $D_s$, 
    we know that $s$ belongs to $S_2(A)$.
    Then, we insert the edge $sA$ to $\mathcal{H}$, 
    if it does not exist yet.
    This takes $O(1)$ amortized time.
    In the other case, $s$ is inserted to the \AWNN of 
    $\overline{S_2(A)}$, in $O(\log^2 n)$
    amortized expected time.
    The assignment to $S_2(A)$ or to $\overline{S_2(A)}$
    concludes the insertion.
    Since each cell defines a constant number of regions, 
    the step requires a constant number of lookup 
    and insertion operations on each level.
    As the time for insertions into the \AWNN of $\overline{S_2(A)}$ 
    dominates, 
    we get an amortized 
    expected  running time of $O(\log\Psi \log^2 n)$.
    Summing over all insertion steps, we achieve an amortized expected 
    running time of $O(\log\Psi \log^{4} n)$ 
    per insertion, due to the dominating running time of the 
    deletions from $\overline{S_2(A)}$.
    
    The query time follows
    directly from the query time in $\mathcal{H}$.
    
    Analyzing the space usage of the data structure is similar to the proof of \cref{lem:bounded:preprocessing}.
    The quadforest $\mathcal{F}$ requires $O(n \log \Psi)$ space, as does $\mathcal{H}$ by  \cref{lem:bounded:size}.
    This space usage is again dominated by the space required for the \AWNN{}s for $S_1(A)$ and $\overline{S_2(A)}$ for all regions $A$.

    First, we have $\sum_{A\in\mathcal{A}} |\overline{S_2(A)}| = O(n \log \Psi)$ for the same reason as in \cref{lem:bounded:S1S2size}.
    Following this and \cref{lem:bounded:S1S2size}, we get a space usage for the \AWNN{}s of
    \begin{align*}
    &~~\phantom{=} O\left(\sum_{A\in\mathcal{A}} |S_1(A)|\log  |S_1(A)| +|\overline{S_2(A)}|\log |\overline{S_2(A)}| \right)\\
    &= O\left(\left(\sum_{A\in\mathcal{A}} |S_1(A)| +|\overline{S_2(A)}| \right)\cdot \log n \right)\\
    &= O(n \log \Psi \log n),
    \end{align*}
    which is also the overall space bound.
    The theorem follows.
\end{proof}

\section{Arbitrary radius ratio}\label{sec:unbounded}
We extend the approach from \cref{sec:bounded} to obtain a decremental data 
structure with a running time that is independent of $\Psi$. 
The $O\left(n \log \Psi\log^4n\right)$ term in \cref{thm:bounded:deletion} 
came from the total size of the sets $S_2(A)$ which, in turn, followed
from the height of the quadtrees in $\mathcal{F}$.
We can get rid of this dependency by using a \emph{compressed quadtree} 
$\Q$ instead of $\mathcal{F}$ (see \cref{sec:more_prelims}).
The height and size of $\Q$  do not depend on the radius ratio 
or the diameter 
of $S$, but only on $n$.
Nonetheless, the height of $\Q$ could still be $\Theta(n)$, 
which is not favorable for our purposes. 
In order to reduce the number of edges in our proxy graph to 
$O\left(n\log^2n\right)$, we use a 
\emph{heavy path decomposition} of $\Q$ (see \cref{sec:more_prelims})
in combination with a \emph{canonical decomposition}  of every heavy path.
The new proxy graph $H$ is described in \cref{subsec:unbounded:proxy}, 
and the decremental connectivity data structure based on $H$ 
can be found in \cref{subsec:unbounded:deletion}.

\subsection{More preliminaries}
\label{sec:more_prelims}

Let $\diam(S) = \max_{s,t\in S} \| st \|$ be the diameter of the
initial site set $S$.
To simplify our arguments, we assume without loss of generality that
$S$ and its associated radii are scaled 
so that all associated radii are at least $1$. 
This allows us to keep working with the hierarchical grid $\G$
from \cref{sec:hier-grid}

\subparagraph*{Compressed quadtrees.}
If we define a quadtree for a set $\mathcal{C}$
of $n$ cells as in \cref{sec:hier-grid},
then it has $O(n)$ leaves and height $O(\log(|\rho|))$, 
where $\rho$ is the smallest cell in $\G$ that contains all cells
of $\mathcal{C}$.
This height can be arbitrarily large, even if $n$ is small.
To avoid this, we need the notion of a 
\emph{compressed quadtree} $\Q$~\cite{har-peled_geometric_2011}.
Let $\T$ be the (uncompressed) quadtree for $\mathcal{C}$.
Then, let $\sigma_1, \dots, \sigma_k$ be a maximal path in
$\T$ towards the leaves, where all $\sigma_i$, $1 \leq i \leq k - 1$,
have only one child that contains (not necessarily proper) cells from $\mathcal{C}$, 
and no $\sigma_i$, $2 \leq i \leq k - 1$ lies in $\mathcal{C}$.
In the compressed quadtree $\Q$, 
this path is replaced by the single edge $\sigma_1 \sigma_k$. 
Then, $\Q$ has $O(n)$ vertices, height $O(n)$,
and it can be constructed in $O(n\log n)$
time~\cite{BuchinLMM11,har-peled_geometric_2011}.

Indeed, the latter
construction algorithm is stated for planar
point sets (and not for cells), but it can be applied by 
using a set of 
$O(n)$ virtual sites, similar to a construction 
of Har-Peled~\cite{har-peled_geometric_2011}:
for each cell $\sigma \in \mathcal{C}$, we add two virtual sites
that lie at the centers of two of the four cells 
that partition $\sigma$, 
see \cref{fig:neighborhoodlevels}.
Now, all cells in $\mathcal{C}$ have at least two children in the 
non-compressed quadtree for the virtual sites, and thus the cells are 
also present in the compressed quadtree for the virtual sites,
as constructed by the traditional algorithms.

The construction of Buchin et al.\ \cite{BuchinLMM11} does not use the floor function at the cost of having cells that are not aligned. In \cref{sec:app:quadtrees} we give the details on how to slightly modify the construction process and the operations on the quadtrees such that we can for our purposes assume to use an aligned compressed quadtree with pointers from each cell to its neighbors at each level of the hierarchical grid.
\begin{figure}
\center
\includegraphics[page=2]{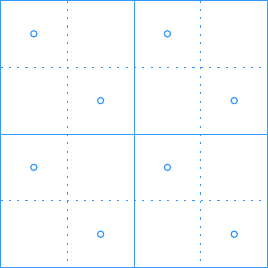}
\caption{Four cells with the virtual sites to ensure that they are present in the compressed quadtree}
\label{fig:neighborhoodlevels}
\end{figure}

\subparagraph*{Heavy paths.}
Let $T$ be a rooted ordered tree (i.e., we have an order
on the children of each interior node of $T$). 
An edge $uv \in T$ is called \emph{heavy} if $v$ 
is the first child of $u$ in the given child-order
that maximizes the total number of nodes in the subtree rooted
at $v$ (among all children of $u$). Otherwise,
the edge $uv$ is \emph{light}.
By definition,
every interior node in $T$ has exactly
one child that is connected by a heavy edge.

A \emph{heavy path} is a maximum path in $T$ that
consists only of heavy edges.
The \emph{heavy path decomposition} of $T$ is 
the set of all the heavy paths in $T$.
The following lemma summarizes a classic result on
the properties of heavy path decompositions.

\begin{lemma}[Sleator and Tarjan~\cite{sleator_data_1983}]\label{lem:heavypath}
Let $T$ be a tree with $n$ vertices.
Then, the following properties hold:
\begin{enumerate}
  \item Every leaf-root path in $T$ contains $O(\log n)$ light 
  edges;
  \label{lem:heavypath:height}
  \item every vertex of $T$ lies on exactly one heavy path; and
  \label{lem:heavypath:path}
  \item the heavy path decomposition of $T$ can be constructed in $O(n)$ time.
  \label{lem:heavypath:constr}
\end{enumerate}
\end{lemma}

\subsection{The proxy graph}\label{subsec:unbounded:proxy}

The general structure of the proxy graph is as 
in \cref{subsec:bounded:proxy}, and we will often refer back to it.
We still have a bipartite graph with $S$ 
on one side and a set of regions vertices on the other side.
The regions are again used to define sets $S_1(A)$ and $S_2(A)$ 
that determine the edges.
However, we adapt the regions $A$ and 
define them based on certain \emph{subpaths} of the compressed quadtree 
$\Q$ instead of single cells.
Furthermore, we relax the condition on the radii in 
the definition of the sets $S_1(A)$.

As usual, for a site $s \in S$, let $\sigma_s$ be the
cell in $\G$ with $s \in \sigma_s$ and 
$|\sigma_s| \leq r_s < 2|\sigma_s|$. 
Let $N(s)$ be the $(15 \times 15)$-neighborhood of $\sigma_s$.
Let $\mathfrak{N} = \bigcup_{s \in S} \{ \sigma_s \} \cup N(s)$, and let $\mathcal{Q}$
be the compressed quadtree for $\mathfrak{N}$.
Now, let $\cR$ be the heavy path decomposition of $\Q$, as
in \cref{lem:heavypath}.
For each heavy path $R \in \cR$, we find a set $\P_R$ 
of \emph{canonical paths} such that every subpath of $R$ 
can be written as the disjoint union of $O(\log n)$ canonical paths. 
Specifically, for each $R \in \cR$, 
we build a \emph{biased} binary search tree $T_R$ 
with the cells of $R$ in the leaves, 
sorted by increasing diameter. 
The weights in the biased binary search 
tree are chosen as described by Sleator and
Tarjan~\cite{sleator_data_1983}: for a node $\sigma$ of $R$, 
let the weight $w_\sigma$ be the number of nodes in $\Q$ that are below
$\sigma$ (including $\sigma$), but not below another node of $R$ below 
$\sigma$. 
Then, the depth of a leaf $\sigma$ in $T_R$ is $O(\log (w_R/w_\sigma))$,
where $w_R$ is the total weight of all leaves in $T_R$.
We associate each vertex $v$ in $T_R$ with the path induced by the 
cells in the subtree rooted at $v$, and we add this path to $\P_R$.
Using this construction, we can write every path in $\Q$ that 
starts at the root as the disjoint union of $O(\log n)$ canonical 
paths, as shown in the following lemma:

\begin{lemma}\label{lem:pathunion}
Let $\sigma$ be a vertex of $\Q$, and let $\pi$ 
be the path from the root of $\Q$ to $\sigma$. 
There exists a set $\P_\pi$ of canonical paths such that:
\begin{enumerate}
\item $|\P_\pi| = O\left(\log n\right)$; and
\item $\pi$ is the disjoint union of the canonical paths in $\P_\pi$.
\end{enumerate}
\end{lemma}

\begin{figure}
\includegraphics[page=2]{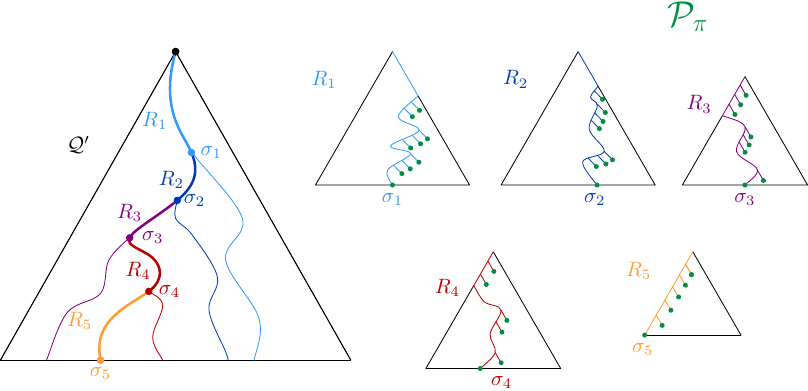}
\caption{Illustration of \cref{lem:pathunion}. 
On the left, we see the decomposition of $\pi$ into $R_1,\dots R_k$. 
On the right, the vertices defining $\mathcal{P}_\pi$ are depicted in green.}
\label{fig:heavypath}
\end{figure}

\begin{proof}
Consider the heavy paths $R_1,\dots, R_k$ encountered along $\pi$.
By \cref{lem:heavypath}, $k = O(\log n)$, and the path $\pi$
can be subdivided into the intersections $\pi \cap R_i$.
Each of these intersections constitutes a subpath of $R_i$ 
whose largest cell is also the largest cell of $R_i$.
Let $\sigma_i$ be the smallest cell of $\pi \cap R_i$. 
Then, the subpath of $R_i$ consists of all cells in $R_i$ with 
diameter at least $|\sigma_i|$.
This subpath can be composed as the disjoint union of all canonical 
paths defined by the right children along the search path of
$\sigma_i$ in the biased binary search tree associated with $R_i$, 
together with the path consisting of $\sigma_i$ only; see \cref{fig:heavypath}.
Summing the depths of the corresponding leaves  of the 
biased binary search trees, we get that  
the overall number of canonical paths for $\pi$
is $O(\log n)$. 
\end{proof}

The vertex set of the proxy graph $H$ again consists of $S$ 
and a set of regions $\mathcal{A}$.
We define $O(1)$ regions for each canonical path $P$ 
in a similar way as in 
\cref{subsec:bounded:proxy}. Let $\sigma$ be the smallest cell and $\tau$ 
the largest cell of $P$. 
The \emph{inner} and \emph{middle regions} of $P$ are defined as 
in \cref{subsec:bounded:proxy}, 
using $\sigma$ as the defining cell. 
More precisely, the inner region for $P$ is the disk with center $a(\sigma)$ 
and radius $|\sigma|$.
The middle regions of $P$ are the $d_2$ regions that 
are defined as the intersection of the cones in $\mathcal{C}_{d_2}(\sigma)$ 
with the annulus of inner radius $|\sigma|$ and outer radius 
$\frac{5}{2}{|\sigma|}$.
For the \emph{outer regions} of $P$, we extend the outer radius of 
the annulus:
they are defined as the intersections of the cones in 
$\mathcal{C}_{d_1}(\sigma)$ 
with the annulus of inner radius $\frac{5}{2}{|\sigma|}$ 
and outer radius $\frac{5}{2}|\sigma| + 2|\tau|$, 
again centered at $a(\sigma)$.
The set $\mathcal{A}$ now contains the regions defined in this way for 
all canonical paths.

Given a region $A\in \mathcal{A}$ for a canonical path $P$ with
smallest cell $\sigma$ and largest cell $\tau$, 
we can now define the sets $S_1(A)$ and $S_2(A)$.
The set $S_1(A)$ is defined similarly to the
analogous set in \cref{subsec:bounded:proxy}, again using $\sigma$ 
as the defining cell for most parts. The difference is that the 
radius range for a site $t$ in a set \(S_1(A)\) is larger, as its upper bound depends on 
the diameter of $\tau$.
The set $S_1(A)$ contains all sites $t$ such that
\begin{enumerate}[label=(\roman*)]
\item $t \in A$; 
\item $|\sigma| \leq r_t < 2|\tau|$; and
\item $\| a(\sigma)t\| \leq r_t + \frac{5}{2}|\sigma|$.
\end{enumerate}
The last condition is only relevant if $A$ is an outer region, 
as it is trivially true for middle and inner regions.

The definition for $S_2(A)$ is also similar 
to \cref{subsec:bounded:proxy}, using canonical
paths instead of cells.
Let $s \in S$ be a site and
$\pi_s$ be the path in $\Q$ from the root to $\sigma_s$.
Let $\P_{\pi_s}$ be the decomposition of $\pi_s$ into canonical paths as 
in \cref{lem:pathunion}.
Let $A$ be a region, defined by a canonical path $P$. 
Then, $s \in S_2(A)$ if
\begin{enumerate}[label=(\roman*)]
\item  $P \in \P_{\pi_s}$ and
\item $s$ is adjacent in $\cD(S)$ to at least one site in $S_1(A)$.
\end{enumerate}
If $\sigma$ is the smallest cell in a canonical path defining a region $A$, 
then every site $s \in S_2(A)$ lies in $\sigma$, has 
$r_s < 2|\sigma|$, and intersects at least one site in $S_1(A)$ 
, satisfying basically the same conditions we had in \cref{subsec:bounded:proxy}. 
However, as the definition here is restricted to those canonical paths in 
$\P_{\pi_s}$, not all sites satisfying the conditions from \cref{subsec:bounded:proxy} are considered for inclusion in $S_2(A)$. 
As we will see below, this suffices to make sure that the proxy graph 
represents the connectivity, while also ensuring that each site $s$ 
lies in few sets $S_2(A)$.

The graph $H$ is now again defined by connecting 
each region $A \in \mathcal{A}$ to all sites in 
$s\in S_1(A) \cup S_2(A)$.
To show that $H$ accurately reflects the connectivity in $D(S)$, 
we need the following corollary of \cref{lem:bounded:clique1}.

\begin{corollary}\label{cor:unbounded:clique}
Suppose that $d_1 \geq 23$ and $d_2 \geq 8$.
Then, for any region $A \in \mathcal{A}$, the 
sites in $S_1(A)$ form a clique in $\cD(S)$.
\end{corollary}

\begin{proof}
Recall that the center of the annuli and disks defining $A$ 
is given by the smallest cell of the associated canonical path.
A close inspection of the proof for \cref{lem:bounded:clique1} 
shows that we only use the lower bound on the radii of the sites in 
$S_1(A)$.
As this lower bound is unchanged, all arguments carry over 
for sites with larger radii.
\end{proof}

\begin{lemma}\label{lem:unbounded:conn}
Two sites are connected in $H$ 
if and only if they are connected in $\cD(S)$.
\end{lemma}

\begin{proof}
Let $s, t \in S$.
If $s$ and $t$ are connected in $H$, 
the same argument as in the proof of \cref{lem:bounded:conn} 
with \cref{cor:unbounded:clique} instead of \cref{lem:bounded:clique1} applies.
The more challenging part is to show that if $s$ and $t$ are connected in 
$\cD(S)$, they are also connected in $H$.
\begin{figure}
\centering
\includegraphics{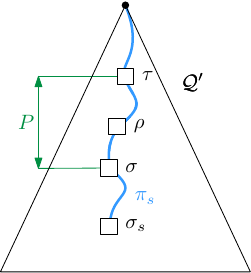}
\caption{The cell $\sigma_s$ is the smallest cell with 
$r_s \leq 2|\sigma_s|$. The canonical path $P$ contains $\rho$.}
\label{fig:unboundconns}
\end{figure}
It suffices to show that if $s$ and $t$ are adjacent in $\cD(S)$, 
they are connected to the same region $A\in \mathcal{A}$.
Refer to \cref{fig:unboundconns} for an illustration 
of the following argument.
Assume without loss of generality that $r_s \leq r_t$,
and consider the neighborhood $N(t)$ of $t$. 
By \cref{lem:bounded:neighborhood}, there is a cell $\rho \in N(t)$ 
that contains $s$.
Consider the path $\pi_s$ in $\Q$ from the root to $\sigma_s$.
Then, $\rho$ belongs to $\pi_s$.
Let $\mathcal{P}_{\pi_s}$ be the decomposition of $\pi_s$ 
into canonical paths as in \cref{lem:pathunion}, and let 
$P \in \mathcal{P}_{\pi_s}$ be 
the canonical path containing $\rho$.
Let $\sigma$ and $\tau$ be the smallest and the largest cell on $P$, 
respectively.
By the definition of $P$, 
we have $\sigma_s \subseteq \sigma \subseteq \rho \subseteq \tau$.
As $st$ is an edge in $\cD(S)$, we have 
$\| st \| \leq r_s + r_t < 2|\sigma| + 2|\tau|$, 
and thus $\| a(\sigma) t\| \leq \frac{5}{2}|\sigma| + 2|\tau|$.
This implies that $t$ lies in a region $A$ defined by $P$,  and thus 
$t\in S_1(A)$.
As $D_s$ intersects $D_t$, it intersects at least one site in $S_1(A)$.
Moreover, $P$ is a canonical path in $\mathcal{P}_{\pi_s}$, and thus $s\in S_2(A)$.
Thus, $s$ and $t$ are both connected to $A$ in $H$,
and the claim follows.
\end{proof}

\begin{lemma}\label{lem:unbounded:size}
  The total size of the binary search trees that define the canonical paths is $O(n)$ and $\sum_{A\in\mathcal{A}}(|S_1(A)| + |S_2(A)|) = 
  O\left(n\log n\right)$.
The proxy graph $H$ has $O(n)$ vertices and 
$O\left(n\log n\right)$ edges. 
\end{lemma}

\begin{proof}
As discussed in \cref{sec:more_prelims}, 
the compressed quadtree $\Q$ consists of $O(n)$ cells. 
Each cell is part of exactly one heavy path, so the total size 
of the binary search trees that define the canonical paths is $O(n)$.
A balanced binary search tree with $m$ leaves has $O(m)$ 
inner vertices, and thus, there is a total of $O(n)$ canonical paths.
As each canonical path defines $O(1)$ regions, the number of regions,
and therefore the number of vertices in $H$ follows.

To bound the number of edges, we again count the number of sets 
$S_1(A)$ and $S_2(A)$ that a single site $t \in S$ can be contained in.
Let $P$ be a canonical path, 
let $\sigma$ and $\tau$ be the smallest and largest cell
of  $P$, and let $A$ be an 
associated region.
Suppose that $t \in S_1(A)$.
We claim that $N(t)$ contains a cell 
that belongs to $P$.
By property (ii) in the definition of $S_1(A)$, 
we have $|\sigma| \leq r_t < 2|\tau|$.
If $|\sigma|\leq r_t < 2|\sigma|$, the claim holds by 
\cref{lem:bounded:neighborhood} (we have $\sigma \in N(t)$).
Thus, suppose that $2|\sigma| \leq r_t < 2|\tau|$,
and let $\rho \in \G$ be the cell  with 
$2|\sigma| \leq |\rho| \leq r_t < 2|\rho|$ and 
$\sigma \subset \rho \subseteq \tau$. 
By property (ii) in the definition of  $S_1(A)$, we have 
 \[
 \| ta(\sigma) \| \leq r_t + \frac{5}{2}|\sigma|
  \leq 2|\rho| + \frac{5}{4}|\rho|  
 = \frac{13}{4}|\rho|,
 \]
since $r_t < 2 |\rho|$  and  $2|\sigma| \leq |\rho|$. 
Hence, by the triangle inequality
$\| ta(\rho) \|  \leq \frac{15}{4}|\rho|$, as $\sigma \subseteq \rho$.
By the same considerations as in the proof of 
\cref{lem:bounded:neighborhood}, it follows that $\rho \in N(t)$, 
and $\rho$ is in $\Q$ and thus in $P$.

Let $\rho$ be a cell of $N(t)$, and let $R$ 
be the heavy path containing $\rho$.
Then, $t$ is considered for the set $S_1(A)$ 
for all regions that are defined by a canonical path containing $\rho$.
These are the $O(\log n)$ paths along the search path  for $\rho$ 
in the binary search tree on the cells of $R$, and thus each site 
can be part of at most $O(\log n)$ sets $S_1(A)$.

The number of sets $S_2(A)$ that contain a fixed site $t \in S$ 
is at most the number of canonical paths that partition the 
path from the root to $\sigma_t$.
By \cref{lem:pathunion}, there are $O(\log n)$ such paths.
This yields $\sum_{A\in\mathcal{A}}(|S_1(A)| + |S_2(A)|) = 
O\left(n\log n\right)$ as an upper bound for the number of edges.
\end{proof}

\subsection{Decremental data structure}\label{subsec:unbounded:deletion}

The approach for the decremental structure is the same as in 
\cref{subsec:bounded:deletion}, see \cref{fig:bounded:dec}. 
We again store for each region $A \in \mathcal{A}$ 
the sets $S_1(A)$ and $S_2(A)$ together with an
associated \RDS.
The set $B$ for the \RDS is again $S_1(A)$, and the set $R$ is $S_2(A)$.
Furthermore, we store the graph proxy graph $H$ 
in an HLT-structure $\mathcal{H}$.

Both queries and deletions work exactly as in 
\cref{subsec:bounded:deletion}, but we repeat them here for 
completeness.
Queries are performed directly to $\mathcal{H}$. 
To delete of a site $s$ from $S$, we first remove all edges 
incident to $s$ from $\mathcal{H}$. 
Then, $s$ is removed from all the sets $S_1(A)$ containing it, 
as well as from the associated \RDS{}s.
The sites $U$ reported as revealed by the \RDS are then removed from 
the corresponding sets $S_2(A)$, and the edges $uA$, for 
$u\in U \setminus S_1(A)$, are removed from $\mathcal{H}$.
Finally, the site $s$ is removed from all the sets $S_2(A)$ and all 
corresponding \RDS{}s.

\begin{lemma}\label{lem:unbounded:preprocessing}
The data structure can be preprocessed in $O(n\log^4 n)$ expected
time and requires $O(n \log^2 n)$ space.
\end{lemma}

\begin{proof}
To find the regions $\mathcal{A}$, we first compute the extended 
compressed quadtree $\Q$.
As described in \cref{sec:more_prelims},
this takes $O(n \log n)$ time.

After that, we need $O(n)$ additional time to find the heavy paths in
$\Q$, by \Cref{lem:heavypath:constr} of \cref{lem:heavypath}.
It also takes $O(n)$ time to compute the biased binary search trees over the heavy paths by standard techniques\cite{Bent1985BiasedST}. 
This gives us the set of regions $\mathcal{A}$.
To find the sets $S_1(A)$, recall that a site $t$ can only belong to the set $S_1(A)$ for a region defined by a canonical path that 
contains a cell in $N(t)$.
Therefore the sets $S_1(S)$ can  be found as follows:
for each site $t \in S$, find the cells in $N(t)$,
and for each $\rho \in N(t)$, find the heavy path $R$ that contains it.
In the biased binary search tree defined on $R$, follow the search path for 
$\rho$ and for each canonical path along this search path, 
explicitly find the region containing $t$,
and check if the distance condition (condition (iii) for the set \(S_1(A)\)) holds.
A naive implementation of this step takes $O\left(n\log n\right)$ 
time, which is fast enough for our purposes. 

As in \cref{lem:bounded:preprocessing}, we construct an additively 
weighted Voronoi diagram for each set $S_1(A)$, again assigning the weight 
$-r_s$ to each site $s\in S$.
Recall that the time needed to build a single 
Voronoi diagram is $O(|S_1(A)| \log n)$~\cite{berg_computational_2008,fortune_sweepline_1987,sharir_intersection_1985}.
Since $\sum_{A\in\mathcal{A}} |S_1(A)| = O(n\log n)$ by to \cref{lem:unbounded:size}, the construction of 
all diagrams takes \(O\left(n\log^2 n \right)\) time. 
The query time however remains $O(\log n)$ in each diagram, 
as such a diagram contains at most $O(n)$ sites.
This also limits the temporary space required for all Voronoi diagrams to $O(n \log n)$. 

Let $s \in S$. Recall that $\pi_s$ is the path in $\Q$ to $\sigma_s$. 
To find all sets $S_2(A)$ containing $s$, we obtain the decomposition 
of $\pi_s$ into canonical paths, and we query the Voronoi diagrams 
with $s$ for all regions defined by these paths.
As there are $O\left(\log n\right)$ canonical paths in the decomposition, 
this takes an additional $O\left(n\log^2 n\right)$ time for all sites.
Inserting the $O\left(n\log n\right)$ edges into $\mathcal{H}$ takes 
$O\left(\log^2 n\right)$ amortized time each, for a total of 
$O\left(n\log^3 n\right)$.

Again, the step dominating the preprocessing time is the construction of 
the $\RDS$. 
For a single region $A \in \mathcal{A}$, this is expected
$O\left(|S_1(A)|\log^2 n + |S_2(A)| \log^3 n \right)$, 
by \cref{lem:reveal:cor}. 
We have $\sum_{A\in\mathcal{A}} |S_1(A)| = O(n\log n)$ and 
$\sum_{A\in\mathcal{A}} |S_2(A)|=O\left(n\log n\right)$ by \cref{lem:unbounded:size}. 
The claimed preprocessing time follows.

Following \cref{lem:unbounded:size} again, $\mathcal{H}$ requires $O(n \log n)$ space and both the quadtree and all canonical paths require $O(n)$ space.
Each of the \RDS requires $O(|S_1(A)| \log |S_1(A)| + |S_2(A)|)$ space by \cref{lem:reveal:cor}.
Using \cref{lem:unbounded:size} we get that the total space usage for the \RDS{}s is 
\begin{align*}
  &~~\phantom{=} O\left(\sum_{A\in\mathcal{A}} |S_1(A)| \log |S_1(A)| + |S_2(A)| \right)\\
  &= O\left(\left( \sum_{A\in\mathcal{A}} |S_1(A)| \right) \log n + \left( \sum_{A\in\mathcal{A}} |S_2(A)| \right) \right) \\
  &= O(n \log^2 n)
\end{align*}
This  dominates the overall space usage.
\end{proof}

\begin{theorem}\label{thm:unbounded:deletion}
The  data structure described above correctly answers connectivity queries in 
$O(\log n/\log \log n)$ time.
It requires 
$O(n\log^{5} n)$ 
overall expected update time
, assuming the deletions are oblivious of the internal random choices of the data structure.
The data structure requires $O(n \log^2 n)$ space.
\end{theorem}
\begin{proof}
As the only difference with the structure from \cref{thm:bounded:deletion}
is the definition
of the sets $S_1(A)$ and $S_2(A)$, 
the correctness follows from \cref{thm:bounded:deletion}. 
The $O(\log n/\log\log n)$ bound for the queries in $\mathcal{H}$ 
also carries over.
The preprocessing takes 
$O\left(n\log^4 n\right)$ time by 
\cref{lem:unbounded:preprocessing}.
Deletions from $\mathcal{H}$ take $O\left(\log^2n\right)$ 
amortized time for each of the $O\left(n\log n\right)$ edges, 
for an overall time of $O\left(n\log^3 n\right)$.
The time for the sequence of deletions is again dominated by the 
time needed for the updates in the \RDS{}s.
For a region $A \in \mathcal{A}$, let $k_A$ be the number of 
sites deleted from $S_1(A)$. 
Then, the time needed by the \RDS associated with $A$ is 
$O\left(|S_1(A)| \log^2 n  + k_A\log^4 n +
|S_2(A)|\log^4 n \right)$ under the assumption of an oblivious adversary, by \cref{lem:reveal:cor}.
We have  
$\sum_{A\in \mathcal{A}} |S_1(A)| = O(n\log n)$,
$\sum_{A\in\mathcal{A}} k_A = O(k\log n)$,
and 
$\sum_{A\in \mathcal{A}} |S_2(A)| = O(n\log n)$, following \cref{lem:unbounded:size}.
Summing over all 
$A\in\mathcal{A}$, we get a running time of
$O\left(n\log^5 n\right)$ 
as claimed.
The space usage is unchanged from \cref{lem:unbounded:preprocessing}.
\end{proof}

It should be noted that following a similar line of argument as Klost~\cite{KLOST2023101979}, \cref{lem:unbounded:preprocessing} implies an efficient static data structure.
\begin{lemma}\label{lem:static:general}
There is a static connectivity data structure for disk graphs with \(O(n\log^2n)\) preprocessing time and \(O(1)\) query time.
\end{lemma}
\begin{proof}
The part in the proof of \cref{lem:unbounded:preprocessing} that computes the proxy graph has a running time of \(O(n\log^2n)\) and yields a graph with the same connected components and only \(O(n\log^2n)\) edges. 
Using a graph traversal to annotate each vertex with a label of its connected component then takes an additional \(O(n\log^2n)\) time and allows \(O(1)\) query time by comparing the labels.
\end{proof}

\FloatBarrier
\section{Conclusion}
We discussed several problems related to dynamic connectivity in 
disk graphs. 
First of all, we significantly improved the state of 
the art for unit disk graphs, by developing data structures tailored 
to this case.
Furthermore, in the general bounded radius ratio case, we 
were able to improve the dependency on $\Psi$ for updates. 
We then considered the incremental and decremental setting. For 
the incremental setting with bounded radius ratio, we gave 
a data structure with an amortized update time that is logarithmic in $\Psi$ 
and polylogarithmic in $n$ and near constant query time.

In order to obtain a similarly efficient data structure 
in the decremental setting, we first considered problems 
related to the lower envelopes of planes and more general two-dimensional
surfaces. 
Using these, we were able to describe a dynamic revealing
data structure that is fundamental for our decremental 
data structure and might be of independent interest.
Using the \RDS, we were able to give data structures with 
$O(\log n/\log\log n)$ query time. 
In the bounded setting, the update time is again logarithmic in $\Psi$ 
and polylogarithmic in $n$, while for the setting of unbounded radius 
ratio, we managed to achieve a data structure whose update 
time depends only on $n$.

For the semi-dynamic data structures described in 
\cref{sec:bounded,sec:unbounded}, this significantly 
improves the previously best time bounds that can be derived from 
the fully dynamic data structure of Chan et al.~\cite{chan_dynamic_2011}.

There are still several open questions. 
First of all, our result in the incremental 
setting only deals with the setting of bounded radius ratio. 
We are currently working on extending it to the case of general disk graphs. 
Furthermore, as the incremental and decremental data structures we 
developed are significantly faster than the fully dynamic data 
structures, an interesting question would be if similar bounds 
can be achieved in the fully dynamic setting.

\bibliography{connectivity}

\appendix

\section{Implementation on the real RAM}

\subsection{Quadforests}\label{sec:app:quadforest}

Given an arbitrary set of sites \(S\) with \(\Phi = \max_{s,t\in S} \Vert st\Vert\), the quadtree as defined in \cref{sec:hier-grid} has height \(O(\log \Phi)\).
In \cref{sec:poly-dependence,sec:bounded}, we developed data structures whose update time depend on the radius ratio \(\Psi\) instead.
This is done by considering a quadforest instead of a quadtree.
The quadforest consists of a set of quadtrees whose root have diameter \(\Psi' = 2^{\lfloor \log_2 \Psi \rfloor}\).
Note that given \(\Psi\), the value \(\Psi'\) can be found with a doubling search in \(O(\log\Psi)\) time.
In \cref{sec:unit,sec:poly-dependence,sec:bounded} we assumed that all roots of quadtrees are aligned to a global grid, where the cells of \cref{sec:unit} can be considered as quadtrees with $\Psi \in O(1)$. 
These assumptions cannot easily be satisfied on the real RAM where no floor function is available.
In the following, we show how we can achieve the same query, update, and space bounds as with the results of \cref{sec:unit,sec:poly-dependence,sec:bounded} without explicitly aligning the quadtree roots to a global grid.

The general idea is to \emph{locally} align quadtree roots instead.
To avoid costly rebuilding when sets of roots with incompatible local alignments get too close after several insertions, we allow a small overlap of sets with different alignments.
In such an overlap quadtree roots of different sets are handled separately.
They may not be aligned to each other and will overlap each other as well.
Disks may fall into such an overlap and will then be inserted into quadtrees of all involved sets according to the respective data structures of \cref{sec:unit,sec:poly-dependence,sec:bounded}, but into at most $O(1)$ many.
A single set of aligned quadtrees might not contain the full neighborhood
of a disk or its cell (according to the respective neighborhood definition applied to it), but all involved sets together cover the required space.

First, we reduce the problem to a single dimension:
we describe how to construct (potentially overlapping) intervals, which allows finding locally aligned quadtree roots in one dimension.
Afterwards, we apply this approach for the $x$- and $y$-dimension separately and combine the results to construct the sets of aligned quadtree roots in two dimensions.
The approaches are dynamic.
Disks can be added or deleted over time, requiring the creation or deletion of new quadtree roots. 
As a byproduct of handling the dimensions separately, some cells close in one but not two dimensions will be (partially) aligned as well, but that does not cause any problems.

\subsubsection{Covering points in one dimension}
\label{subsub:local-alignment-1d}

In the following we describe how to maintain a set of $O(n)$ \emph{intervals} of equal size to cover a dynamic set of $n$ input points---which will be centers of disks later on---in one dimension.
Each $v \in \mathbb{R}$ is covered by $O(1)$ intervals and the intervals are constructed in such a way that local connectivity of disks is preserved when using them later in \cref{subsub:local-align-2d}.
Their size is a multiple of the side length \(\Psi'/\sqrt{2}\) of a quadtree root, inducing an alignment of quadtree roots in one dimension in their interior.

Let $c \in O(\Psi)$ be the width of the neighborhood of a disk or quadtree root of maximum size in the respective connectivity data structure described in \cref{sec:unit,sec:poly-dependence,sec:bounded}.
Then, each interval has a width of $w > 8 c \in O(\Psi)$ with $w$ a multiple of \(\Psi'/\sqrt{2}\).
Each interval saves its left and right endpoint and the number of contained points.

During updates of the point set we adjust the set of intervals to maintain the following invariants:
\begin{description}
    \item[$\mathcal{I}1$] Each input point is covered by at least one interval.
    \item[$\mathcal{I}2$] Each $v \in \mathbb{R}$ is covered by at most two intervals.
    \item[$\mathcal{I}3$] The length of the overlap of two intersecting intervals is at least \(c\) and at most \((w-c)/2\).
    \item[$\mathcal{I}4$] Non-intersecting intervals have a distance of at least $c$.\footnote{Note that this is automatically given for two intervals overlapping a common third interval assuming $\mathcal{I}2$ and $\mathcal{I}3$ hold.}
\end{description}

The left and right endpoints of all intervals are maintained in two separate red-black trees~\cite{cormen_introduction_2009}.
Assuming the invariants hold, the red-black trees allow the retrieval of the up to two covering intervals for a given point in $O(\log n)$ time.
Similarly, given a point, the nearest intervals to the left and to the right can be obtained in $O(\log n)$ time.
Additionally, a red-black tree containing the current input point set is maintained which contains in each inner node the number of descendant leaves, allowing range counting queries in $O(\log n)$ time.

When a new input point \(p\) is added to the point set we might need to create new intervals to fulfill $\mathcal{I}1$.
For this, the endpoint trees are queried for intervals containing the new point.
If one or two intervals are returned, all invariants are already fulfilled for the updated point set and we just need to increase their point counter by one.
If no interval is returned we have to create a new interval to cover the newly added point.
For this, the endpoint trees are queried again for the closest intervals \(I_L\) to the left and \(I_R\) to the right. We denote by \(\Vert I_L I_R \Vert\) the size of the gap between these disjoint intervals.
Then, four different cases can arise, depending on the returned intervals.
In the first three cases a new interval will be constructed which leaves a gap at one or both sides, whereas in the last case a gap will be covered with the new interval.
The cases are visualized in \cref{fig:interval-cases}.

\begin{figure}
    $\ $ 
    \subparagraph*{Case 1a:} $w + 2c < \Vert I_L I_R \Vert$ and the point has distance $> c$ to $I_L$, $I_R$
    \begin{center}
        \includegraphics[page=1]{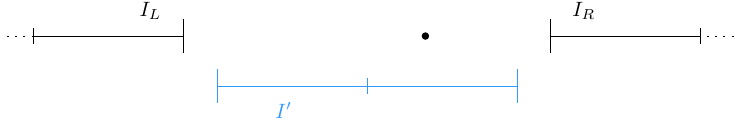}
    \end{center}
    \subparagraph*{Case 1b:} $w + 2c < \Vert I_L I_R \Vert$ and the point has distance $\le c$ to $I_R$
    \begin{center}
        \includegraphics[page=2]{./figures/quadtree1dcases.pdf}
    \end{center}
    \subparagraph*{Case 2:} $w - 2c < \Vert I_L I_R \Vert \le w + 2c$
    \begin{center}
        \includegraphics[page=3]{./figures/quadtree1dcases.pdf}
    \end{center}
    \subparagraph*{Case 3:} $\Vert I_L I_R \Vert \le w - 2c$
    \begin{center}
        \includegraphics[page=4]{./figures/quadtree1dcases.pdf}
    \end{center}
    \caption{Examples for cases 1--4 with $w = 9c$ and distance $\Vert I_L I_R \Vert$ between $I_L$ and $I_R$. Existing intervals are shown at the top and the newly created intervals at the bottom. The points to cover are drawn as black dots.}
    \label{fig:interval-cases}
\end{figure}

\begin{description}
\item[Case 1: \(I_L\) or \(I_R\) does not exist or \( w + 2c < \Vert I_L I_R\Vert \).] This case has two subcases:

\begin{description}
\item[Case 1a: \(\Vert p I_L\Vert > c\) and \(\Vert p I_R\Vert > c\).]
An arbitrary interval covering the new input point is created which has a distance of at least $c$ to the closest intervals.

\item[Case 1b:  \(\Vert p I_L\Vert \le c\) or \(\Vert p I_R\Vert \le c\).]
Let $I$ be the interval with a distance of $c$ or less to the new input point.
Then, a new interval $I'$ is created such that $I'$ overlaps $I$ by $c$ and contains the new input point.

\end{description}
\item[Case 2: $w - 2c <\Vert I_L I_R\Vert\leq w+2c $.]
Assume without loss of generality that the distance from the new input point to $I_R$ is at most the distance to $I_L$.
Then, we create an interval $I'$ which overlaps $I_R$ by $3c$ on the left side.
This new interval must overlap the new input point due to $(w + 2c) / 2 < (w - 3c)$ and has a distance to $I_L$ of more than $(w - 2c) - (w - 3c) = c$.

\item[Case 3:  $\Vert I_L I_R\Vert \leq w - 2c$.]

A new interval $I'$ can be created centrally between them.
Assuming the distance between $I_L$ and $I_R$ is at least $c$,
then $I'$ must overlap $I_L$ and $I_R$ each by at least $c$ and by at most $(w - c) / 2$.
\end{description}

Additionally, a range counting query is done on the existing point set, to find all points contained in the new interval. The counter of the new interval is then initialized with the result plus one.
Finally, \(p\) and the new interval are added to all relevant red-black trees. 

When an input point is deleted, we query for all intervals containing it.
In all obtained intervals the internal counter is decremented and intervals reaching zero are deleted.
Afterwards, the corresponding red-black trees are updated.

The intervals maintained with the construction described fulfill the invariants $\mathcal{I}1$--$\mathcal{I}4$ by induction.
Additionally, every interval covers at least one input point.
All operations for maintaining the intervals require $O(\log n)$ time per point update, and we maintain three red-black trees and an interval set of size $O(n)$.
Hence, we can conclude:

\begin{lemma}
    \label{lem:covering-intervals-1d}
    We can maintain $O(n)$ intervals for a dynamic set of points in one dimension which fulfill invariants $\mathcal{I}1$--$\mathcal{I}4$, where $n$ is the maximum number of points at any time.
    To maintain the intervals $O(\log n)$ additional time per point update operation and $O(n)$ additional space is required.
\end{lemma}

\subsubsection{Local alignment in two dimensions}
\label{subsub:local-align-2d}

To locally align quadtree roots in two dimensions we construct interval sets according to \cref{lem:covering-intervals-1d} separately for the $x$- and $y$-dimension where the sites form the point set of the lemma.
Then, each pair of intervals with one interval from the $x$-set together with one interval from the $y$-set forms a square.
Inside each such square, we can for any given point identify a quadtree root aligned to the square's boundary using binary search in $O(\log \Psi)$ time.
This allows us to apply the data structures of \cref{sec:unit,sec:poly-dependence,sec:bounded} for each square separately\footnote{For correctness, time bounds, and space bounds it is irrelevant whether quadtree roots and its descendant cells normally created by the data structures outside the square are discarded or created nonetheless.}
and a site will be assigned to the data structures of the up to four squares containing it.
In this approach we have to take special care to reflect the global connectivity by carefully handling overlapping squares, and we have to make sure to only handle interval pairs whose squares contain at least one site to keep the space requirement low.

To be more precise, we maintain two interval sets of \cref{lem:covering-intervals-1d}, one for the $x$-coordinates of the sites and one for their $y$-coordinates.
Additionally, we maintain a partial mapping of pairs of $x$- and $y$-intervals to sets of aligned quadtree roots inside the square induced by the respective interval pair. 
This mapping is implemented as a red-black tree~\cite{cormen_introduction_2009} which maps such an interval pair (using e.g.\@ the tuple of their left endpoints as key with lexicographical ordering) to a red-black tree that contains the aligned quadtree roots in the square, with the quadtree root's lower left corner's coordinate as the key in lexicographical order.

By similar arguments as for the global grid $\G_{\lfloor \log \Psi \rfloor}$, given a query point the corresponding quadtree root in an interval pair can be obtained in $O(\log n)$ time due to the alignment of the quadtree roots. 
Hence, given a query point, it is possible to obtain all (up to four) quadtree roots containing a point in $O(\log n)$ with the following approach.
First, query the red-black trees of the interval endpoints used in \cref{lem:covering-intervals-1d} for both dimensions.
Due to $\mathcal{I}2$ of \cref{lem:covering-intervals-1d} at most two intervals are returned in each dimension.
Then, for each of the at most four resulting interval pairs in \(O(\log n)\) time the corresponding square can be found and, as described above, with additional time \(O(\log n)\) the quadtree root can be identified.

This construction already captures all local connectivity information of the corresponding data structure when inserting sites in all interval pairs which overlap them. 
To extend this to global connectivity, we will describe in the following step by step how the data structures from the main part of the paper can be adapted to work with the alignment framework described above.
We will refer to the data structures of the main body as the \emph{original} data structures and to the new data structure that is obtained by using the original data structures with the interval as the \emph{adapted} data structures.

\begin{lemma}
    \label{lem:real-ram-local-connectivity}
    Fix a connectivity data structure \(\mathcal{X}\) as described Theorem~\ref{thm:dynamicUDG}, \ref{thm:dg-dynamicDG}, \ref{thm:bounded-radius-ratio-psi}, \ref{thm:bounded:deletion}, or \ref{thm:bounded:insertion}.
    Construct a new data structure by applying the mapping based on interval pairs of \cref{lem:covering-intervals-1d} to a set of sites, constructing an instance of \(\mathcal{X}\) for every interval pair that covers at least one site, and inserting all sites into all instances of \(\mathcal{X}\) assigned to interval pairs covering it. For this new data structure the following holds:
    \begin{enumerate}
    \item For each pair of sites that are connected by an edge in the disk graph, there exist an interval pair covering both sites.
    \item If \(\mathcal{X}\) is not the data structure from \cref{thm:bounded-radius-ratio-psi}, then the direct connectivity between two sites in the instance of \(\mathcal{X}\) of each covering pair is maintained as in the corresponding Lemma~\ref{lem:gridreachability}, \ref{lem:dg-graph-properties}, or \ref{lem:bounded:conn}.
    \item If \(\mathcal{X}\) is the data structure from \cref{thm:bounded-radius-ratio-psi}, then in the instance $x$ of \(\mathcal{X}\) of each covering pair of intervals for two directly connected sites
    either (a) one of the sites' disk Minkowski covers the cell of the other site, or (b) they share an edge in the connectivity graph constructed inside $x$.
    \item If two sites are not connected in the disk graph, then they are not connected in all instances of \(\mathcal{X}\).
\end{enumerate}
\end{lemma}
\begin{proof}
    Let $c$ be the width of a neighborhood in the fixed data structure as defined in \cref{subsub:local-alignment-1d}.
    We first consider each dimension and its associated interval set of \cref{lem:covering-intervals-1d} separately.

    By $\mathcal{I}1$ of \cref{lem:covering-intervals-1d} each site is covered by at least one interval of the set.
    We show by contradiction that two sites whose disks intersect are covered by a common interval.
    Hence, assume that they are not covered by a common interval.
    Then, two different cases can arise: Either all intervals covering one site do not intersect any interval covering the other site, or at least one interval covering one site intersects an interval covering the other site.

    By the definition of the neighborhood in \(\mathcal{X}\), the centers of two intersecting disks must have a distance of \emph{less} than $c$ in each dimension.
    Hence, if no intersection between covering intervals exists, invariant $\mathcal{I}4$ of \cref{lem:covering-intervals-1d} is contradicted, as their distance must be at least $c$.
    If there is an overlap between an interval pair, then at least one interval must contain both disk centers as the overlap is at least of size $c$ due to $\mathcal{I}3$ and the length of the intervals is larger than $2c$, contradicting the assumption.
    Hence, there exists an interval in each dimension covering the centers of each pair of intersecting disks.

    Following this, there exists a pair of an $x$-interval and a $y$-interval, such that their associated two-dimensional area covers both sites.
    Therefore, if \(\mathcal{X}\) is not of \cref{thm:bounded-radius-ratio-psi}, then its instance of a covering pair contains all relevant quadtree roots, cells, and other associated parts required to show their connectivity according to Lemma~\ref{lem:gridreachability}, \ref{lem:dg-graph-properties}, or \ref{lem:bounded:conn}.
    In case the data structure is of \cref{thm:bounded-radius-ratio-psi} and no site's disk Minkowski covers the cell of the other site in the instance of \(\mathcal{X}\) of the covering pair, then the instance has both disks connected in its internal HLT-structure due to the edge conditions in \cref{sub:dg-limit-insertions-to-mbm} and the same reasoning as for the other data structures.

    Similarly, the instances of \(\mathcal{X}\) are constructed according to Theorem~\ref{thm:dynamicUDG}, \ref{thm:dg-dynamicDG}, \ref{thm:bounded-radius-ratio-psi}, \ref{thm:bounded:deletion}, or \ref{thm:bounded:insertion}.
    Hence, two sites in different components of the disk graph are not connected in any instance of \(\mathcal{X}\).
\end{proof}

Note that in the construction of instances of \cref{thm:bounded-radius-ratio-psi} for interval pairs the proxy vertices for cells are also determined locally only.
In the original, global approach of \cref{thm:bounded-radius-ratio-psi} or in instances of the data structure for other interval pairs there might exist other, possibly larger, sites' disks which Minkowski cover a site's cell in a local instance of \cref{thm:bounded-radius-ratio-psi} in \cref{lem:real-ram-local-connectivity}.
Due to \cref{lem:real-ram-local-connectivity} they again share an interval pair, as they have intersecting disks.
Hence, multiple candidates for proxy vertices according to \cref{lem:dg-gnore-fully-connected} can exist in different interval pairs, which requires additional attention later on.

We can use the approach of \cref{lem:real-ram-local-connectivity} as a basis to adapt the disk connectivity data structures of \cref{sec:unit,sec:poly-dependence,sec:bounded} on the real RAM.
It remains to bridge the individual data structures with local connectivity into a global connectivity data structure.
This can be done by using a global proxy graph $H$, which is shared between all the data structures, instead of many local ones.
Additionally, individual subcomponents in the global proxy graph need to be linked appropriately.
For the data structures of \cref{thm:bounded:deletion,thm:bounded:insertion} no further action is required.
For the other data structures we add a vertex to $H$ for every site and link this vertex to all cells of different interval pairs this site gets assigned to.

\begin{lemma}
    \label{lem:real-ram-2d-local-to-global-bridges}
    Fix a connectivity data structure \(\mathcal{X}\) of Theorem~\ref{thm:dynamicUDG}, \ref{thm:dg-dynamicDG}, \ref{thm:bounded-radius-ratio-psi}, \ref{thm:bounded:deletion}, or \ref{thm:bounded:insertion}.
    Additionally, assume that the interval pairs and instances of \(\mathcal{X}\) are constructed as described in \cref{lem:real-ram-local-connectivity}, except that the proxy graph $H$ is augmented and shared globally as described above.
    For this new data structure the following holds:
    \begin{enumerate}
        \item If \(\mathcal{X}\) is not the data structure from \cref{thm:bounded-radius-ratio-psi}, then two sites are connected in $H$ iff they are connected in the disk intersection graph.
        \item If \(\mathcal{X}\) is the data structure from \cref{thm:bounded-radius-ratio-psi}, then at least one pair of proxy cells of two sites is connected in $H$ iff the two sites are connected in the disk intersection graph.
    \end{enumerate}
\end{lemma}
\begin{proof}
    First assume that \(\mathcal{X}\) is not the data structure from \cref{thm:bounded-radius-ratio-psi}.
    Then, by \cref{lem:real-ram-local-connectivity} the connectivity of each pair of sites with intersecting disks is reflected in $H$ through at least one instance of \(\mathcal{X}\), possibly indirectly through the cells they are assigned to.
    Due to the (potentially added) vertices and edges for each site which link the sites' cells or regions between different instances of \(\mathcal{X}\), a site is connected to all other sites in $H$ to which it has an edge in the disk intersection graph.
    This allows applying the proofs of Lemma~\ref{lem:gridreachability}, \ref{lem:dg-graph-properties}, or \ref{lem:bounded:conn}.

    Now, assume \(\mathcal{X}\) is the data structure from \cref{thm:bounded-radius-ratio-psi}.
    Let $D_{s}$, $D_t$ be the disk pair of two connected sites $s$, $t$ and let $\{\sigma_i\}$, $\{\tau_j\}$ be the up to four cells of instances of \(\mathcal{X}\) the sites are assigned to.
    Then, let $\{\sigma_i'\}$ be the up to four cells which are proxy vertices of the cells in $\{\sigma_i\}$ and let $\{D_{s_i'}\}$ be the up to four disk which define the proxy vertices, i.e.\@ the disks of maximum radius which Minkowski cover the respective $\overline{\sigma_i}$.
    Every site whose disk contains $D_s$ must be in at least one instance of \(\mathcal{X}\) which also contains $s$ by \cref{lem:real-ram-local-connectivity}.
    Hence, at least one disk $D_{s_{k}'}$ must be inclusion maximal.
    Similarly, at least one disk $D_{t_{\ell}'}$ defining a proxy vertex of a $\tau_j$ must be inclusion maximal.

    Similar to the other case, by \cref{lem:real-ram-local-connectivity} and the bridging edges, the connectivity of all paths of intersecting inclusion maximal disks is preserved in $H$.
    This allows applying the proof of \cref{lem:dg-gnore-fully-connected} with $\sigma_{k}'$ and $\tau_{\ell}'$ as the proxy vertices.
\end{proof}

As in the original data structures, this global proxy graph is represented by an instance of \cref{thm:generalinc} when extending \cref{thm:bounded:insertion} and the HLT-structure of \cref{thm:dynamicspanningtree} for the remaining cases.
This directly implies that the query bounds from the original data structures carry over to the adapted data structures.
For all adapted data structures except the one from \cref{thm:bounded-radius-ratio-psi} the data structure of the proxy graph is queried directly with the sites.
For the adapted data structure derived from the original data structure from \cref{thm:bounded-radius-ratio-psi} all up to four proxy vertices are obtained for both sites according to the instances of the original data structure
and then the HLT-structure is queried for connectivity of all up to 16 pairs of proxy vertices.

To insert a new site (unless the original data structure is of \cref{thm:bounded:deletion}), we obtain the intervals overlapping the corresponding coordinate of the site from both sets of intervals.
When at least one interval per dimension could be found, then the corresponding instances of the original data structure for the up to four pairs are obtained (or initialized if required) and the site is added to them.
New aligned quadtree roots can be created and added to the respective red-black trees managing them in $O(\log \Psi + \log n)$ time each, which is always overshadowed by the insertion routine of the respective original data structure.

If instead for at least one dimension no interval could be obtained, then we have to create a new interval as described in \cref{subsub:local-alignment-1d}.
This new interval can overlap previously existing intervals and contain already inserted sites, requiring us to update various data structures.
We have to insertion of the sites contained in the new interval into the original data structure assigned to the interval.
First, all sites contained in the new interval in the corresponding dimension are obtained with a range query using the red-black tree for points of \cref{subsub:local-alignment-1d} in $O(\log n + k)$ time, where $k$ is the number of returned sites.
Then, for each of the $k$ sites, the covering intervals of the other dimension are obtained.
For each pair formed by the new interval with an existing interval a new original data structure is initialized, and the pair is added to the red-black tree of pairs.
The $k$ sites are then added to the new original data structure.
This requires in total time for $k$ insertions into the instances of the original data structure and $O(k \log n)$ for the remaining updates.
The cost for the $k$ additional insertions can be charged to the original insertion of the $k$ sites, as they occur at most once per dimension due to $\mathcal{I}3$ of \cref{subsub:local-alignment-1d}.

Afterwards, the new site and the bridging edges are added to $\mathcal{H}$ if required by the original data structure.
In total, we have to pay for up to four insertions into instances of the original data structure and the additional operations require $O(\log n)$ time in total.

Deletions work similar as insertions, assuming the original data structure supports deletions.
First, the corresponding $x$- and $y$-intervals are obtained in $O(\log n)$ time, then the up to four instances of the original data structure.
In each of them the site is deleted.
When a quadtree root or even an instance of the data structure for an interval is no longer used, then it is destroyed.
In the latter case the corresponding interval pair is removed from the red-black tree of pairs.
Additionally, the counter for covered sites of the intervals is decremented.
In case one interval's counter reaches a value of zero, then the interval is deleted as described in \cref{lem:covering-intervals-1d} and removed from all relevant red-black trees.
Note that in that case only one or two interval pairs containing that interval can exist, all covering the deleted site.
Furthermore, the site vertex and its bridging edges required for \cref{lem:real-ram-2d-local-to-global-bridges} are removed from $\mathcal{H}$ if they were introduced, depending on the original data structure.
All steps require the time for at most four deletions in the instances of the original data structures, four deletions in $\mathcal{H}$, and additional $O(\log n)$ overhead for the remaining work.
Hence, asymptotically the run time is unchanged.

Although we require a relatively large amount of auxiliary data structures, all of them require $O(n)$ additional storage.
Similarly, we require $O(n)$ space for the interval data structures due to \cref{lem:covering-intervals-1d}.
The number of additional vertices and edges in comparison to the original data structures is also bound by $O(n)$.
This results in an asymptotic space usage identical to the respective original data structure. 

Combining everything above, we can conclude the following:
\begin{theorem}
    \label{thm:real-ram-everything-works}
    (Semi-)dynamic disk connectivity can be implemented with the bounds of \cref{thm:dynamicUDG,thm:dg-dynamicDG,thm:bounded-radius-ratio-psi,thm:bounded:deletion,thm:bounded:insertion} on a real RAM.
\end{theorem}

\subsection{Compressed quadtrees with neighbor pointers}\label{sec:app:quadtrees}
In this section we describe how the compressed quadtrees can be adapted in order to allow the same running times as given in \cref{sec:unbounded} in the real RAM model without using the floor function.
The relevant operations that are needed for the data structure in \cref{sec:unbounded} are constructing a compressed quadtree that contains the cells in \(N(s)\) for any given \(s\in S\) and efficiently finding the neighbors on the same level of a given cell.

The algorithm of Buchin et al.\ \cite{BuchinLMM11} to construct a compressed quadtree from a set of point sites takes \(O(n\log n)\) time on the real RAM without rounding. It achieves this running time with the disadvantage of not having all cells aligned to the hierarchical grid.
On a high level, Buchin et al.\ start building a non-compressed quadtree and whenever splitting a cell only lead to one new non-empty cell for \(c\) iterations, a compressed vertex is added and a new bounding box is computed for the vertices in the cell.
This process is then continued recursively.

This section will focus on how the data structure in \cref{sec:unbounded} can be easily adapted to work with a non-aligned grid.
We still want to have a cell \(\sigma_s\) for each site \(s\in S\) with \(s\in \sigma\) and \(|\sigma_s|\leq r_s < 2 |\sigma_s|\) but relax the definition in so far that \(\sigma_s\) does not have to be part of a global hierarchical grid but is the cell of the (not necessary aligned) compressed quadtree on the sites with this property.
Below we will describe a process to guarantee that such a cell exists.
We also relax the notion of the neighborhood \(N(s)\). We will not use a \(15\times 15\) subgrid on the same level as \(|\sigma_s|\), but instead use all cells of the compressed quadtree that would intersect this \(15\times 15\) subgrid and whose diameter lies in the interval \(\left[|\sigma_s|, 2|\sigma_s|\right)\).
The size of the neighborhood is still \(O(1)\).

\begin{observation}\label{obs:app:neighborhooddisjoint}
Given a site \(s\), the cells in \(N(s)\) are disjoint.
\end{observation}

\begin{lemma}
Let \(S\subseteq \R^2\) be a set of sites with associated radii. A compressed quadtree that contains a suitable neighborhood \(N(s)\) for every \(s\in S\) can be constructed in \(O(n\log n)\) time on a real RAM.
Furthermore, the constructed quadtree allows to return the neighborhood of a cell in \(O(1)\) time, assuming a pointer to the cell is given.
\end{lemma}
\begin{proof}
We use the algorithm of Buchin et al.\ \cite{BuchinLMM11} to construct the quadtree. The algorithm expects a set of points as its input, so we construct a point set that makes sure that a cell \(\sigma_s\) and its neighborhood are contained in the quadtree.

For this we take all sites \(s\in S\) and surround them with a \(61\times 61\) grid of points, where the vertical and horizontal distance between neighboring points is \(\frac{r_s}{2\sqrt{2}}\). 
The distance between these points enforces a cell \(\sigma_s\) with \(s\in\sigma_s\) and \(|\sigma_s|\leq r_s < 2|\sigma_s|\) together with the cells in its \(15\times 15\) neighborhood to be split and thus be present in the compressed quadtree, see \cref{fig:app:compressedQuadtree}.
As these are \(O(n)\) points, the construction takes \(O(n\log n)\) time.

Now we describe how to find the neighborhood. This can be done by constructing a set of pointers in a traversal of the constructed quadtree.
Here the main observation is, that the neighborhood of a given cell is fully contained in the neighborhood of its parent.
So when traversing the quadtree cells in increasing size, one can always traverse the neighborhood of the parent in order to find the neighborhood of any cell.
This traversal can be done using a priority queue in \(O(n\log n)\) overall time.
After this traversal, the stored neighborhood of a cell can be returned in \(O(1)\) time.
\end{proof}

\begin{figure}
\centering
\includegraphics{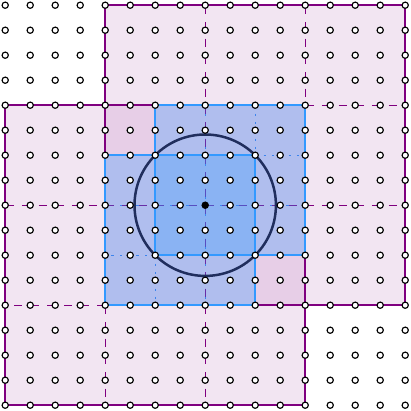}
\caption{A site \(s\) together with the dummy points for constructing the compressed quadtree that includes the \(3\times 3\) neighborhood. For the \(13\times 13\) neighborhood more dummy points are added.}
\label{fig:app:compressedQuadtree}
\end{figure}

\end{document}